\documentclass{article}

\usepackage[T1]{fontenc}
\usepackage{amsmath}
\usepackage{amsthm}
\usepackage{amssymb}
\usepackage{graphicx}
\usepackage{float}
\usepackage{dsfont}
\usepackage{xcolor}

\newtheorem{theorem}{Theorem}[section]
\newtheorem{lemma}[theorem]{Lemma}
\newtheorem{corollary}[theorem]{Corollary}

\theoremstyle{definition}

\newtheorem{example}[theorem]{Example}
\newtheorem{assumption}[theorem]{Assumption}
\theoremstyle{remark}
\newtheorem{remark}[theorem]{Remark}

\newcommand{\R}{\mathbb{R}}  
\newcommand{\Z}{\mathbb{Z}}
\newcommand{\E }{\mathbb{E}}  
\renewcommand{\P}{\mathbb{P}}  
\newcommand{\Var}{\operatorname{Var}}  

\newcommand{\eqdis}{\stackrel{d}{=}}
\newcommand{\distr}{\stackrel{d}{\to}}

\newcommand{\SC}{\operatorname{SC}} 
\newcommand{\SCB}{\operatorname{SCB}} 
\newcommand{\Pois}{\operatorname{Pois}} 
 
\newcommand{\NegBin}{\operatorname{NegBin}}

\newcommand{\calL}{\mathcal{L}}

\newcommand{\calE}{\mathcal{E}}
\newcommand{\calF}{\mathcal{F}}
\newcommand{\calG}{\mathcal{G}}
\newcommand{\calH}{\mathcal{H}}
\newcommand{\calI}{\mathcal{I}}

\def\bz{\boldsymbol{z}}

\def\b0{\boldsymbol{0}}

\title{Claims processing and costs under capacity constraints}
\author{Filip Lindskog and Mario V.~W\"uthrich}
\begin{document}
\maketitle

\begin{abstract}
Random delays between the occurrence of accident events and the corresponding reporting times of insurance claims is a standard feature of insurance data. The time lag between the reporting and the processing of a claim depends on whether the claim can be processed without delay as it arrives or whether it remains unprocessed for some time because of temporarily insufficient processing capacity that is shared between all incoming claims. We aim to explain and analyze the nature of processing delays and build-up of backlogs. 
We show how to select processing capacity optimally in order to minimize claims costs, taking delay-adjusted costs and fixed costs for claims settlement capacity into account. Theoretical results are combined with a large-scale numerical study that demonstrates practical usefulness
of our proposal.
\end{abstract}



\section{Introduction}
Accident events give insurance policyholders the right to financial compensation. In most cases, the effective costs are not immediately known to the insurance company. Reporting delay is a standard feature of insurance data, and once a claim is reported, it is not necessarily settled quickly
since claims may take time to process. 
With an unlimited processing capacity, claims can be processed fast(er), however, economic reasons allow insurance companies
to only allocate a limited capacity to the claims settlement unit (process). Naturally, this
capacity should be bigger than the average claims volume, otherwise there
will be growing an infinitely large backlog of unprocessed claims. This is a consideration on average, which is essentially distorted by the fact
that claims occurrence and reporting can cluster, i.e., there may be
peaks of claims reportings, but also quiet periods where reportings are
below average. The question we study is how can the capacity be set
optimally so that backlogs of processing are not too large, and at the same time periods of low reportings do still not lead to very quiet times in
claims handling units. The latter implies high fixed costs, as an
inactive claims handling unit still needs to be compensated.
Generally too high backlogs also lead to additional costs, it is verified
that late claims settlements typically increase the claims costs. Thus, there
is a trade-off in costs between low and high claims processing capacity, and our aim is to study an optimal balance between the two.

The study of systems where constraints on processing capacity induces processing delays and dependence is at the heart of queueing theory. 
If we would be content with studying the system with incoming reported claims and outgoing processed claims without labeling the incoming reported claims, then our study could essentially be reduced to an application of standard queueing theory. However, insurance applications require the input to be labeled in the sense that incoming reported claims belong to different contract groups and we need to keep track of the evolution of processed claims for each such group that shares the processing capacity.    
This topic seems new to the actuarial literature because we did not find any literature on this topic. It has come to our attention because we have been approached by an insurer facing a backlog of unprocessed claims that needed to be worked off in a cost efficient way. This question is related to queueing theory from where we borrow some mathematical results. Nevertheless, many parts significantly differ from queueing theory, mathematically as well as from an interpretation and terminology point of view. 


Optimal capacity sizing for systems where users share the system's capacity has been studied in the operations research literature. Optimality may be considered in terms of stability of the system or in terms of maximization of profits generated by the system. The study by Maglaras and Zeevi \cite{Maglaras-Zeevi-2003} provides one example from this area of research.

In the actuarial literature there does not seem to exist works that study problems close to the one we consider. However, economic consequences of delays in claims settlement has indeed been studied. 
Boogaert and Haezendonck~\cite{Boogaert-Haezendonck-1989} consider claims arriving according to a homogeneous Poisson process. To the sequence of claims an i.i.d.~sequence of triplets $(X_n,H_n,V_n)$ is added, where $X_n$ is the claim size, $H_n$ the handling delay, and $V_n$ the payment delay. By considering an economic environment with time-varying inflation and interest rate, the present value of insurance liabilities is affected by possible dependence between the elements of the triplet $(X_n,H_n,V_n)$, such as positive dependence between claim size and handling delay.     
Huynh et al.~\cite{Huynh-et-al-2015} consider reported claims according to a compound Poisson model where incoming claims are either processed (and paid) immediately or investigated by the claim handler. Incoming claims are handled independently with equal probability of being investigated. An investigation causes a delay in the processing of the claim and also a claim cost whose distribution differs from that of claims that do not undergo investigation. The resulting surplus process can be seen as the output from a queueing system. The effects on the ruin probability of investigating claims, through delayed processing and modified claim cost, are investigated. 
Related studies of effects of delays in claims settling on ruin probabilities are 
Waters and Papatriandafylou \cite{Waters-Papatriandafylou-1985} and 
Albrecher et al.~\cite{Albrecher-et-al-2011}; see also references therein. 
Although there are studies on effects of processing and payment delays on liability values or ruin probabilities, we have not found literature that analyzes effects of capacity constraints on processing delays and delay-adjusted claims costs for contract groups that share processing capacity.     

The paper is organized as follows. 
Section \ref{sec:rep_pro_bac} defines the key variables for our study, motivates their relationships and explains the queueing system that arises. 
Section \ref{sec:costs} discusses costs for processing capacity, delay-adjusted claims costs, and introduces joint models for combined delay-adjusted claims costs and settlement costs that form the basis of the minimization problems we study. 
Section \ref{sec:processing_procedure} introduces the procedure for how claims are processed, taking into account that at each time different contract groups share the current processing capacity. 
Section \ref{sec:comp_backlog_exp} derives expressions for unconditional and conditional expectations of the number of claims in the backlog for different contract groups. These expectations are key ingredients in the cost minimization problems. 
Section \ref{sec:neg_bin_model} introduces the stochastic model that allows us to consider a realistic large-scale application of the framework and results presented earlier.     
Section \ref{sec:NN} explains in detail how we approximate terms appearing in the expressions for the backlog expectations by recurrent neural networks.  
Section \ref{sec:costs_negbin} solves the cost minimization problems numerically.  
Finally, we summarize in Section \ref{sec:summary}.

\section{Reported, processed, and backlog claims}\label{sec:rep_pro_bac}
\subsection{Definitions and assumptions}
We index the number of reported (R), processed (P), and backlog (B) claims by occurrence period and development period. 
Let $\{i_0,i_0+1,\dots,1,2,\dots\}$ be the index set for occurrence periods and let $\{0,1,\dots\}$ be the index set for development periods. 
Periods may refer to months, quarters, years, etc. The general index sets allow us to define and study the numbers of reported, processed and backlog claims as stochastic processes. 

Occurrence period $i$ starts at (calendar) time $i-1$ and ends at (calendar) time $i$.  
Let $R_{i,j}$ denote the number of reported claims due to accident events during occurrence period $i$ that are reported during development period $j$. Hence, these reportings occur between time $i-1+j$ and time $i+j$.     
We assume the existence of a non-random integer $J$ such that $R_{i,j}=0$ for any $j>J$. Hence, we assume a maximal reporting delay of $J+1$. 
Let $P_{i,j}$ denote the number of processed claims due to events during occurrence period $i$ that are processed during development period $j$. A claim cannot be processed before it has been reported. 
$R_{i,j}$ and $P_{i,j}$ are observable at time $i+j$. 
Let $B_{i,j}$ denote the number of backlog claims due to events during occurrence period $i$ that are already reported but in the backlog and, therefore, have not yet been processed by development period $j$. Set $B_{i,0}\equiv 0$, meaning that each new occurrence period starts without any backlog. 
$B_{i,j}$ is observable at time $i+j$ (and may be inferred at time $i+j-1$ depending on what other variables are observed, see \eqref{ax:Bdynamics} below). 
Let $C_t>0$ denote the number of claims that the insurer has the capacity to process during time period $t$, i.e., the time period between time $t-1$ and $t$. In general, $C_t$ is a random variable.    

The total number of reported claims during time period $t$ is 
\begin{align*}
R_t:=\sum_{i=t-J}^{t}R_{i,t-i}=\sum_{j=0}^{J}R_{t-j,j}.
\end{align*}
The total number of processed and backlog claims, respectively, during time period $t$ are  
\begin{align*}
P_t&:=\sum_{i=i_0}^{t}P_{i,t-i}=\sum_{j=0}^{t-i_0}P_{t-j,j}, \\ 
B_t&:=\sum_{i=i_0}^{t}B_{i,t-i}=\sum_{j=0}^{t-i_0}B_{t-j,j}. 
\end{align*} 
There is sufficient capacity during time period $t$ to process all not-yet-processed claims reported during period $t$ or earlier if the event $\SC_t$, given by 
\begin{align*}
\SC_t:=\big\{B_t+R_t\leq C_t\big\}, 
\end{align*}
is true. The inequality means that there is sufficient capacity to process both the backlog $B_t$ and the newly reported claims $R_t$. If there is not sufficient capacity, $\SC_t^c$, some backlog is forwarded to the next time period.   

We emphasize some general properties for reported, processed, and backlog claims. 
Each reported claim will at some point be processed, possibly by temporarily contributing to the backlog. A processed claim remains processed. A claim in the backlog either remains temporarily in the backlog or transitions into a processed claim (terminal state). In mathematical terms, any procedure for processing claims should satisfy the following axioms \eqref{ax:positive}-\eqref{ax:processed}:  
\begin{align}
& R_{i,j}\geq 0, \, B_{i,j}\geq 0, \, B_{i,0}=0, \, P_{i,j}\geq 0, \, \text{for all } i,j, \label{ax:positive} \\ 
& \sum_{j=0}^{\infty}R_{i,j}=\sum_{j=0}^{\infty}P_{i,j}, \, \text{for all } i, \label{ax:closed} \\
& B_{i,j+1}=B_{i,j}+R_{i,j}-P_{i,j}, \, \text{for all } i,j. \label{ax:Bdynamics}
\end{align}
An immediate consequence of \eqref{ax:positive} and \eqref{ax:Bdynamics} is 
\begin{align*}
\sum_{j=0}^{k}R_{i,j}=B_{i,k+1}+\sum_{j=0}^{k}P_{i,j}, \, \text{for all } i,k. 
\end{align*} 
The total number of processed claims during time period $t$ is 
\begin{align}\label{ax:processed}
P_t=\min(B_t+R_t,C_t), 
\end{align}
i.e., the sum of the numbers of backlog and reported claims if the sum does not exceed the capacity $C_t$ of period $t$, otherwise $C_t$. 
From \eqref{ax:Bdynamics} it follows that $B_{t+1}=B_t+R_t-P_t$ which together with \eqref{ax:processed} give 
 \begin{align}\label{eq:backlog_recursion}
B_{t+1}=\max\big(B_t+R_t-C_t,0\big).
\end{align}
The recursion \eqref{eq:backlog_recursion} for the total size of the backlog is an example of the Lindley recursion that is well studied in queueing theory, see Example I.5.7 and Chapter III.6 in \cite{Asmussen-2003}. If $(R_t)_{t=0}^{\infty}$ and $(C_t)_{t=0}^{\infty}$ are i.i.d.~sequences, the Lindley process $(B_t)_{t=0}^{\infty}$ is a Markov process given by 
\begin{align*}
B_0=b, \quad B_{t+1}=\max\big(B_t+R_t-C_t,0\big), \quad t \ge 0,
\end{align*}
which is the waiting time process for a GI/G/1 queue. If $\E[R_t]<\E[C_t]$ and $\E[(R_t-C_t)^2]<\infty$, there is a stationary distribution with finite mean to which $B_t$ converges in distribution as $t\to\infty$. 

\begin{assumption}\label{assump:axioms}
The stochastic system $\{(B_{i,j},R_{i,j},P_{i,j},C_{i+j}):i\geq i_0,j\geq 0\}$ satisfies \eqref{ax:positive}, \eqref{ax:closed}, \eqref{ax:Bdynamics} and \eqref{ax:processed}. 
\end{assumption}

To understand the stochastic system, we specify $\sigma$-algebras that play a natural role in conditional probabilities and expectations that will appear. 
Let
\begin{align*} 
\calE_t&:=\sigma\big(R_{s},B_{s}:s\leq t\big),\\
\calF_t&:=\sigma\big(R_{s},B_{s},P_{s}:s\leq t\big),\\
\calG_t&:=\sigma\big(R_{i,j},B_{i,j}:i+j\leq t, j\geq 0\big),\\
\calH_t&:=\sigma\big(R_{i,j},B_{i,j},P_{i,j}:i+j\leq t, j\geq 0\big).
\end{align*} 
By construction, the $\sigma$-algebras obviously satisfy $\calE_t\subset \calF_t\subset \calH_t$ and $\calE_t\subset \calG_t\subset \calH_t$. From \eqref{ax:Bdynamics}, it follows that $\calF_t \subset \calE_{t+1} \subset \calF_{t+1}$  and $\calH_t \subset \calG_{t+1} \subset \calH_{t+1}$. Note that $B_{t+1}$ is $\calF_t$-measurable but in general not $\calE_t$-measurable. Similarly, $B_{i,j+1}$ is $\calH_{i+j}$-measurable but in general not $\calG_{i+j}$-measurable. 
From \eqref{ax:Bdynamics} and \eqref{ax:processed}, it follows that $\SC_t\in\calF_t$. However, we emphasize that from Assumption \ref{assump:axioms} alone it does not follow that $C_t$ is measurable w.r.t.~any of the $\sigma$-algebras $\calE_t,\calF_t,\calG_t,\calH_t$. This is not surprising: if we are not observing the capacity $C_t$ but only its effect on the number of backlog claims and processed claims, then the actual capacity may be larger than the capacity used to the process claims. However, if $C_t$ is non-random, then \eqref{ax:Bdynamics} and \eqref{ax:processed} together imply that $B_{t+1}$ is $\calE_t$-measurable. 

\begin{remark}
In queueing theory, the Lindley recursion \eqref{eq:backlog_recursion} describes the waiting time $B_t$ for the $t$th customer arriving to a single service station, with $R_t$ the service time for the $t$th customer and $C_t$ the time between the arrival of the $t$th customer and that of the $(t+1)$th customer. If both $(R_t)$ and $(C_t)$ are i.i.d.~sequences,  the queueing system is denoted GI/G/1. The special case when both service times $R_t$ and inter-arrival times $C_t$ are exponentially distributed is denoted M/M/1. The special case when $C_t=c$ is constant is denoted D/G/1. We emphasize that the interpretation of the variables $B_t$, $R_t$ and $C_t$ in our setting is quite different although many results known for GI/G/1 queues are also useful to understand the dynamics of the total number of backlog claims. The expectation $\E[B]$ of the stationary distribution of $(B_t)$ can in general not be obtained explicitly. However, it can be approximated numerically; see, e.g., \cite{Janssen-Leeuwaarden-2005} and \cite{Janssen-Leeuwaarden-2018} for the analysis of the D/G/1 queue.     
\end{remark}

\subsection{Stationary behavior}\label{sec:stationary}
The recursion \eqref{eq:backlog_recursion} for the total size of the backlog is an example of the Lindley  recursion that is well studied in queueing theory, see Chapter III.6 \cite{Asmussen-2003}. 
As a direct consequence of the recursion \eqref{eq:backlog_recursion}, 
\begin{align*}
B_{t+1}=\max\bigg(B_{0}+\sum_{s=0}^{t}(R_s-C_s),\sum_{s=1}^{t}(R_s-C_s),\dots,R_t-C_t,0\bigg), \quad t \ge 0. 
\end{align*}
The asymptotic behavior of $B_{t+1}$ is well understood when 
$(R_s-C_s)_{s=0}^{\infty}$ 
is an i.i.d.~sequence with $\E[R_s]<\E[C_s]$. 
If $(C_t)_{t=0}^{\infty}$ is an i.i.d.~sequence, if all $R_{i,j}$ are independent, and if $R_{i',j}\eqdis R_{i,j}$ for all $j$, then 
$(R_s-C_s)_{s=0}^{\infty}$ is an i.i.d.~sequence if $i_0\leq -J$.    
We write $\E[R]$ for the common expected value for any element of the i.i.d.~sequence 
$(R_s)_{s=0}^{\infty}$.   
Note that in this case, $\E[R]=\sum_{j=0}^{J}\E[R_{i,j}]$ is simply the expected total number of reported claims for any occurrence period $i$.   
By Corollary III.6.5 in \cite{Asmussen-2003}, we have convergence in distribution
\begin{align*}
B_t \distr \max_{0\leq s< \infty}\bigg(0,\sum_{r=0}^{s}(R_{r}-C_{r})\bigg) \quad \text{as } t\to\infty.
\end{align*}
From $\sum_{r=0}^{s}(R_{r}-C_{r})\to-\infty$,  a.s., as $s\to\infty$, 
it follows that the limit variable is well defined.  
Dropping the subscript, we denote by $B$ a random variable whose distribution is the stationary distribution of 
$(B_t)_{t=0}^{\infty}$.  
By Proposition VIII.4.5 \cite{Asmussen-2003},  
\begin{align*}
\E[B]=\sum_{k=1}^{\infty}\frac{1}{k}\,\E\left[\max(S_k,0)\right], \qquad S_k:=\sum_{s=0}^{k-1} (R_s-C_s). 
\end{align*} 
Unfortunately, it is rarely possible to compute $\E[B]$ explicitly and numerical evaluation is non-trivial in general; see, e.g., \cite{Stadje-1997} for an approach to computing stationary probabilities for integer-valued $B$.    

Upper bounds for $\E[B]$ are studied in \cite{Chen-Whitt-2020} by considering a scaled version of the Lindley recursion in the setting of GI/G/1 queues. Since $(R_t)$ and $(C_t)$ are i.i.d.~sequences in the GI/G/1 queueing setting, we write $R$ and $C$ for an arbitrary element of these sequences.  
The recursion \eqref{eq:backlog_recursion} can be written 
 \begin{align*}
\widetilde{B}_{t+1}=\max\big(\widetilde{B}_t+\widetilde{R}_t-\widetilde{C}_t,0\big),
\end{align*}
where $\widetilde{B}_t:=B_t/\E[C]$, and similarly for $\widetilde{R}_t$ and $\widetilde{C}_t$. 
By construction $\E[\widetilde{C}_t]=1$, and $\E[\widetilde{B}_t]=\E[B]/\E[C]$ and $\E[\widetilde{R}_t]=\E[R]/\E[C]=:\rho$. The parameter $\rho$ is called the traffic intensity in queueing theory. 
The so-called heavy-traffic approximation of $\E[B]$ (obtained by considering the behavior as $\rho\to 1$) is 
\begin{align}\label{eq:HTA}
\frac{\E[B]}{\E[C]}\,\approx\, \frac{\rho^2}{2(1-\rho)}\bigg(\frac{\Var[R]}{\E[R]^2}+\frac{\Var[C]}{\E[C]^2}\bigg), 
\end{align}
this is equivalent to expression (2.9) in \cite{Chen-Whitt-2020}. 
In the D/G/1 setting $\Var[C]=0$, since $C=c$ is constant, and \eqref{eq:HTA} takes the form 
\begin{align}\label{eq:HTA_constant_C}
\E[B]\,\approx\, \frac{\E[R]}{c-\E[R]}\,\frac{\Var[R]}{2\E[R]}.
\end{align}
It can be shown ((2.6) and (2.7) in \cite{Chen-Whitt-2020}) that the right-hand side in \eqref{eq:HTA_constant_C} coincides with the upper bounds for $\E[B]$ obtained by Kingman in \cite{Kingman-1962} and by Daley in \cite{Daley-1977}.  
In the M/G/1 setting, where $C$ is exponentially distributed with mean $c$, $\Var[C]/\E[C]^2=1$ and the heavy-traffic approximation for $\E[B]$ in \eqref{eq:HTA} equals 
\begin{align*}
\frac{c\rho^2}{2(1-\rho)}\bigg(\frac{\Var[R]}{\E[R]^2}+1\bigg)=\frac{\E[R]}{c-\E[R]}\,\frac{\E[R^2]}{2\E[R]},
\end{align*}
which, in fact, is an exact expression for $\E[B]$ referred to as the Pollaczek-Khintchine formula; see VIII.(5.6) in \cite{Asmussen-2003}.  

\begin{remark}\label{rem:poisson_backlog}
In the D/G/1 setting, an upper bound for $\E[B]$ is given by the right-hand side in \eqref{eq:HTA_constant_C}. If we write $c=\eta \E[R]$ for some $\eta>1$, then 
\begin{align*}
\E[B]\leq \frac{1}{\eta-1}\,\frac{\Var[R]}{2\E[R]}.
\end{align*}
If $\E[R]$ is fairly large (which is the case for realistic insurance applications), then, unless $\eta\approx 1$, $\Var(R)$ needs to be substantially larger than $\E[R]$ in order for $\E[B]$ and $\E[R]$ to be of similar size. In particular, if $R$ is Poisson distributed, then $\Var[R]=\E[R]$, and no substantial backlog will appear unless $\eta\approx 1$, corresponding to a close to non-stationary system. In our examples below, we consider a negative binomial
distribution for the number of reported claims $R$.
\end{remark}

\section{Cost implications of capacity constraints}\label{sec:costs}

Capacity constraints are mainly due to limited financial resources, and considerations of how much of these financial resources should be allocated to the claims settlement unit. Claims settlement costs are belonging to the unallocated loss adjustment expenses (ULAE) meaning that these costs are not specific to an individual claim, but they are rather overhead costs that are necessary to run the claims handling unit. ULAE then need to be allocated to occurrence periods (or individual insurance contracts) in order to be able to perform a profit analysis for occurrence periods (or individual insurance contracts); see Buchwalder et al.~\cite{Buchwalder} for a method supporting the chain-ladder claims reserving method.
In most cases, the allocation is chosen to be proportional either to claims costs, claims counts or a linear combination of the two. 

We will perform an analysis of expected claims costs where we consider unconditional expectations of the numbers of reported, processed, and backlog claims. This means that we consider i.i.d.~sequences $(R_t)_{t=0}^{\infty}$ and $(C_t)_{t=0}^{\infty}$ under the stochastic system in Assumption \ref{assump:axioms} in its stationary state. In particular, the backlog process $(B_t)_{t=0}^{\infty}$ is given by \eqref{eq:backlog_recursion} with initial state $B_0$ drawn from the stationary distribution, see Section \ref{sec:stationary}. We write $\E[R]$ and $\E[B]$ for the expected number of reported claims and size of the backlog, respectively, in any given period for the stationary system. We write $\E[C]$ for the expected maximal number of claims that can be processed in any given period. Recall that $\E[C]>\E[R]$ is assumed since we are considering the system in its stationary state. We further assume that claims costs for individual claims form an i.i.d.~sequence independent of the stochastic system in Assumption \ref{assump:axioms}.    

{\bf Claims settlement costs.} 
Making claims processing capacity available generates ULAE and these expenses need to be allocated to occurrence periods to have an integrated cost view. 
Suppose that any claims occurrence period $\tau$ uses the processing capacity during periods $\tau, \tau+1,\dots,\tau+J_P-1$, where $J_P\geq 1$. The cost allocated to occurrence period $\tau$ is the total expected cost for these processing capacities multiplied by the fraction of expected number of processed claims for occurrence period $\tau$ during these periods divided by the expected number of processed claims for all occurrence periods that share the capacity during these periods.
Let us formalize this. The set of index pairs $(i,j)$ for the number of processed claims $P_{i,j}$ during the periods $\tau, \tau+1,\dots,\tau+J_P-1$ is 
\begin{align*}
\calI:=\{(i,j):\tau\leq i+j\leq \tau+J_P-1,0\leq j\leq J_P-1\},
\end{align*}
and due to stationarity of the numbers of processed claims ($\E[P_{i,j}]=\E[P_{i',j}]$) we may sum over rows instead of over diagonals
\begin{align*}
\sum_{(i,j)\in \calI}\E\big[P_{i,j}\big]&=\sum_{k=0}^{J_p-1}\sum_{i=\tau+k-J_p+1}^{\tau+k}\E\big[P_{i,\tau+k-i}\big]\\
&=\sum_{i=\tau}^{\tau+J_P-1}\sum_{j=0}^{J_P-1}\E\big[P_{i,j}\big]=J_P\sum_{j=0}^{J_P-1}\E\big[P_{\tau,j}\big].
\end{align*}
We see that to any occurrence period $\tau$ we should allocate a fraction $1/J_P$ of the cost for processing capacity during $J_P$ periods. The integer $J_P$ cancels out when multiplying the two numbers and we conclude that the capacity cost allocated to any occurrence period equals the full cost for processing capacity during a single period. We emphasize that this is the stationary case.

{\bf Linearly delay-adjusted claims costs.}
The total expected ground-up claims costs for any occurrence period $i$ 
in stationarity is  
\begin{align*}
\kappa_g \sum_{j=0}^{J}\E[R_{i,j}]=\kappa_g \sum_{j=0}^{J}\E[R_{t-j,j}]=\kappa_g \E[R], 
\end{align*}
where $\kappa_g>0$ denotes the expected claims cost for an individual claim if paid without any delay. We assume that delayed processing generally makes claims more expensive. In the most simple linear model ($\ell$) we assume an additional constant $\kappa_b>0$ that originates from late processing. The expected total claims costs of occurrence period $i$ are then in this linear cost model given by 
\begin{align*}
\kappa_g \E[R] + \sum_{j\geq 0} \kappa_b \E\big[B_{i,j}\big] = \kappa_g \E[R] + \sum_{j\geq 0} \kappa_b \E\big[B_{t-j,j}\big] =  \kappa_g \E[R] + \kappa_b \E[B], 
\end{align*}
where we used the assumption of the system in stationarity.  

{\bf Non-linearly delay-adjusted claims costs.}
Alternatively, we could consider a claims cost model where we rather have the view of an inflation-adjusted cost ($\iota$). In that case we consider an additional constant
$\lambda_b>1$ and set
\begin{align*}
 \kappa_g \sum_{j\geq 0} \lambda_b^{j}\,\E\big[P_{i,j}\big]
=\kappa_g \sum_{j\geq 0} \lambda_b^{j}\,\E\big[B_{i,j}+R_{i,j}-B_{i,j+1}\big],
\end{align*}
with $R_{i,j} \equiv 0$ for $j>J$. This model considers processed claims and inflates ground-up costs $\kappa_g \lambda_b^{j}$ by its processing delay from the end of the occurrence period. Note that $\lambda_b^{j}$ should be seen as a super-imposed delay inflation which is different from economic inflation. In fact, by assuming constant ground-up costs we implicitly assume that all costs have been adjusted for economic inflation so that they live on the same scale, and additional backlog costs are then concerned with super-imposed claims inflation and costs related to an increased expense due to late processing and settlements.

{\bf Combining delay-adjusted claims costs and settlement costs.}  
Making claims processing capacity with expectation $\E[C]$ available generates ULAE. We have explained above that the ULAE allocated to any given occurrence period corresponds to the full single-period ULAE. For simplicity, we assume that ULAE have a fixed component
that corresponds to a minimal expected capacity $\E[R]$ ensuring a stationary system, and for the excess expected capacity $\E[C]-\E[R]$ we consider
a proportional cost $\kappa_c (\E[C]-\E[R])$ with $\kappa_c>0$. 
Adding the expected costs for the excess capacity to the delay-adjusted expected claims costs gives the following expected costs for any occurrence period $i$: 
\begin{align}\label{eq:linear_cost_case}
\mu_i^{(\ell)} := \kappa_g \,\E[R] + \kappa_b \, \E[B] + \kappa_c \,(\E[C]-\E[R]),
\end{align}
for the linear cost model, and 
\begin{align}\label{eq:inflating_cost_case}
\mu_i^{(\iota)}: = \kappa_g \sum_{j\ge 0} \lambda_b^{j}\,
\E\big[B_{i,j}+R_{i,j}-B_{i,j+1}\big] + \kappa_c \,\big(\E[C]-\E[R]\big),
\end{align}
for the model with non-linear delay-inflated costs. We emphasize that $\mu_i^{(\ell)}$ for the linear cost model does not depend on how processing capacity is shared between occurrence periods requiring processing capacity (the indexes $i$ and $j$ do not show up in the expression \eqref{eq:linear_cost_case}).  
However, $\mu_i^{(\iota)}$ for the non-linear cost model does indeed depend on how processing capacity is shared. It is known that $\E[B]$ is a convex function of $\E[C]$, see, e.g., \cite{Harel-1990}, and therefore $\mu_i^{(\ell)}$ is a convex function of $\E[C]$. The main question then
is what is the optimal expected capacity to minimize the overall
expected costs.
In general, this expected cost minimization has to be performed numerically, and we return to this in Section \ref{sec:cost_optim_negbin}. 

\begin{example}
As shown in Section \ref{sec:stationary}, a heavy-traffic approximation for $\E[B]$ may give an explicit expression for $\mu_i^{(\ell)}$ of the form  
\begin{align*}
\kappa_g\E[R]+\kappa_b\E[B]+\kappa_c(\E[C]-\E[R])=\kappa_g\E[R]+\kappa_b\frac{\alpha\E[R]}{c-\E[R]}+\kappa_c(c-\E[R]), 
\end{align*}
which can be minimized explicitly with minimizer $c=\E[R]+\sqrt{\alpha\E[R]\kappa_b/\kappa_c}$ and minimum $\kappa_g\E[R]+2\sqrt{\alpha\kappa_b\kappa_c\E[R]}$. 
The D/G/1 queueing model setting corresponds to $\alpha=\Var[R]/(2\E[R])$ which leads to the minimizer $c=\E[R]+\sqrt{\Var[R]\kappa_b/(2\kappa_c)}$ and minimum $\kappa_g\E[R]+\sqrt{2\kappa_b\kappa_c\Var[R]}$.
\end{example}

\section{Sharing processing capacity}\label{sec:processing_procedure}

Many procedures can be considered describing how claims are processed and how the backlog for individual occurrence periods evolves over time. 
We focus on the simple procedure where claims are processed by first processing the backlog and then, if there is processing capacity left after the backlog has been processed, the newly reported claims are processed. However, we do this without assuming continuous-time monitoring of claim arrivals (reporting times). 
Let 
\begin{align*}
\SCB_t:=\big\{B_t\leq C_t\big\},
\end{align*}
denote the event that there is, at time $t$, sufficient capacity to process the backlog. Note that 
\begin{align*}
\SC_t^c\cap \SCB_t=\big\{B_t\leq C_t < B_t+R_t\big\}
\end{align*}
denotes the event that the capacity is insufficient to process all claims waiting to be processed, but sufficient to process the backlog. 

The number of processed claims is the sum of the number of processed backlog claims and the number of processed newly reported claims
\begin{align}
P_{i,t-i}&=P^{B}_{i,t-i}+P^{R}_{i,t-i}. \label{eq:processed_tot}
\end{align}
We assume that, given $\calF_t \vee \calG_t:=\sigma\{\calF_t,\calG_t\}$, claims in the backlog at the beginning of period $t$ will be processed during period $t$ independently with (conditional) probability 
\begin{align}\label{eq:cond_prob_pb}
\mathds{1}_{\SCB_t}+\frac{C_t}{B_t}\mathds{1}_{\SCB_t^c}.
\end{align} 
From \eqref{ax:Bdynamics} and \eqref{ax:processed} it follows that this expression for the conditional probability is $\calF_t$-measurable. This is seen as follows. If $P_t\geq B_t$, then $\mathds{1}_{\SCB_t}=1$ and $C_t\mathds{1}_{\SCB_t^c}=0$. If $P_t<B_t$, then $\mathds{1}_{\SCB_t}=1$ and $C_t\mathds{1}_{\SCB_t^c}=P_t$. Hence, the conditional probability is fully determined by $P_t$ and $B_t$ being both $\calF_t$-measurable. 

The interpretation of the conditional probability \eqref{eq:cond_prob_pb} is straightforward: If there is sufficient capacity to process the backlog, then the probability is equal to one. Otherwise the probability equals the proportion of the capacity to the size of the backlog. 
Given $\calF_t \vee \calG_t$, we know the number $B_{i,t-i}$ of backlog claims for occurrence period $i$. Hence, it follows from the assumption of independently processing the backlog claims that the conditional distribution $\calL(P^{B}_{i,t-i}\mid\calF_t\vee \calG_t)$ is a binomial distribution. Therefore,  
\begin{align}
\E\big[P^{B}_{i,t-i}\mid\calF_t\vee \calG_t\big]=B_{i,t-i}\bigg(\mathds{1}_{\SCB_t}+\frac{C_t}{B_t}\mathds{1}_{\SCB_t^c}\bigg). \label{eq:processed_b}
\end{align}
We assume that, given $\calF_t \vee \calG_t$, claims reported during period $t$ will be processed during period $t$ independently with (conditional) probability 
\begin{align}\label{eq:cond_prob_pr}
\mathds{1}_{\SC_t}+\frac{C_t-B_t}{R_t}\mathds{1}_{\SC_t^c \cap \SCB_t}.
\end{align} 
From \eqref{ax:Bdynamics} and \eqref{ax:processed} it follows that this expression for the conditional probability is $\calF_t$-measurable. That is, if $P_t=B_t+R_t$, then $\mathds{1}_{\SC_t}=1$. Otherwise, if $P_t<B_t+R_t$, then $\mathds{1}_{\SC_t}=0$.  
If $P_t<B_t+R_t$ and $P_t\geq B_t$, then $P_t=C_t\mathds{1}_{\SC_t^c \cap \SCB_t}$. Otherwise, if $P_t=B_t+R_t$ or $P_t<B_t$, then $\mathds{1}_{\SC_t^c \cap \SCB_t}=0$. Hence, the conditional probability is fully determined by $P_t$, $B_t$ and $R_t$ which are all $\calF_t$-measurable. 

The interpretation of the conditional probability \eqref{eq:cond_prob_pr} is as follows: If there is sufficient capacity to process first the backlog and then the newly reported claims, then the probability is equal to one. Otherwise the probability equals the proportion of the remaining capacity (after processing the backlog) to the number of reported claims.  
Similarly to above for the backlog claims, the newly reported claims are processed independently with the corresponding remaining capacity, giving another binomial distribution. Thus, 
given $\calF_t \vee \calG_t$, we know the number $R_{i,t-i}$ of reported claims for occurrence period $i$, the conditional distribution $\calL(P^{R}_{i,t-i}\mid\calF_t\vee \calG_t)$ is a binomial distribution and 
\begin{align}
\E\big[P^{R}_{i,t-i}\mid\calF_t\vee \calG_t\big]=R_{i,t-i}\bigg(\mathds{1}_{\SC_t}+\frac{C_t-B_t}{R_t}\mathds{1}_{\SC_t^c \cap \SCB_t}\bigg).  \label{eq:processed_r}
\end{align}
Equations \eqref{eq:processed_tot}, \eqref{eq:processed_b} and \eqref{eq:processed_r} together define the procedure for processing claims. By summing up the number of processed backlog claims and the number of processed newly reported claims we obtain the conditionally
expected number of processed claims
\begin{align}
\begin{split}
\E\big[P_{i,t-i}\mid\calF_t\vee \calG_t\big]&=B_{i,t-i}\bigg(\mathds{1}_{\SCB_t}+\frac{C_t}{B_t}\mathds{1}_{\SCB_t^c}\bigg) \label{eq:processed_b_and_r} \\
&\quad+R_{i,t-i}\bigg(\mathds{1}_{\SC_t}+\frac{C_t-B_t}{R_t}\mathds{1}_{\SC_t^c \cap \SCB_t}\bigg). 
\end{split}
\end{align}

\begin{assumption}\label{assump:processing}
The stochastic system $\{(B_{i,j},R_{i,j},P_{i,j},C_{i+j}):i\geq i_0,j\geq 0\}$ satisfies \eqref{eq:processed_b_and_r}. 
\end{assumption}

From \eqref{ax:Bdynamics} and \eqref{eq:processed_b_and_r} it follows immediately that 
\begin{align}
\begin{split}
\E\big[B_{i,t-i+1}\mid\calF_t\vee \calG_t\big]&=B_{i,t-i}\bigg(1-\frac{C_t}{B_t}\bigg)\mathds{1}_{\SCB_t^c} \label{eq:backlog_b_and_r} \\
&\quad+R_{i,t-i}\bigg(\mathds{1}_{\SCB_t^c}+\bigg(1-\frac{C_t-B_t}{R_t}\bigg)\mathds{1}_{\SC_t^c \cap \SCB_t}\bigg).
\end{split}
\end{align} 
Note that recursion \eqref{eq:backlog_recursion} (which holds regardless of the choice of procedure for processing claims) can be written
\begin{align*}
B_{t+1}=B_{t}\bigg(1-\frac{C_t}{B_t}\bigg)\mathds{1}_{\SCB_t^c}
+R_{t}\bigg(\mathds{1}_{\SCB_t^c}+\bigg(1-\frac{C_t-B_t}{R_t}\bigg)\mathds{1}_{\SC_t^c \cap \SCB_t}\bigg).
\end{align*}
Note also that by summing over occurrence periods we obtain the same recursion from \eqref{eq:backlog_b_and_r} since in stationarity
\begin{align*}
\sum_i \E\big[B_{i,t-i+1}\mid\calF_t\vee \calG_t\big]=\E\big[B_{t+1}\mid\calF_t\vee \calG_t\big]=B_{t+1}. 
\end{align*}
Hence, whereas \eqref{eq:backlog_recursion} explains the backlog dynamics on the aggregate level, \eqref{eq:backlog_b_and_r} adds information by explaining the backlog dynamics on an individual occurrence period level. We emphasize that \eqref{eq:backlog_recursion} holds for any procedure for processing claims, whereas \eqref{eq:backlog_b_and_r} is a consequence of a particular choice of the claims processing procedure. If, as an example, we would require that newly reported claims should be processed before backlog claims, then \eqref{eq:backlog_recursion} would still hold, whereas \eqref{eq:backlog_b_and_r} would result in another expression.  

\section{Computation of backlog expectations}\label{sec:comp_backlog_exp}

The aim of this section is to provide explicit expressions for unconditional and conditional expectations of the number of backlog claims $B_{i,j}$. Considering the unconditional expectation $\E[B_{i,j}]$ makes most sense if we consider the backlog process in its stationary state. We will therefore (see Assumption \ref{assump:stationary_system}) assume that $(R_t)$ and $(C_t)$ are i.i.d.~sequences with $\E[R_t]<\E[C_t]$ which ensures a stationary distribution and that the Markov chain $(B_t)$ approaches stationarity from an arbitrary fixed initial state. When the initial state $B_0$ of the backlog process $(B_t)_{t=0}^{\infty}$ is drawn from its stationary distribution, the distribution of $B_{i,j}$ does not depend on $i$. On the other hand, we will consider conditional expectations $\E[B_{\tau-i,i+k+1}\mid \calG_{\tau}]$ for $k\geq 0$. For such conditional expectations, stationarity issues for $(B_t)$ do not play a role, but the assumption of i.i.d.~sequences $(R_t)$ and $(C_t)$ is imposed in order to obtain an explicit expression for the conditional expectation. Conditioning on $\calG_{\tau}$ means that we consider a situation where $B_{i,j}$ and $R_{i,j}$ for $i+j\leq \tau$ are observable, e.g., this may address the question of optimally planing capacities when currently facing a large backlog $B_\tau$. This corresponds to data likely available to an actuary.   

\begin{assumption}\label{assump:stationary_system}
The stochastic system $\{(B_{i,j},R_{i,j},P_{i,j},C_{i+j}):i\geq i_0,j\geq 0\}$ satisfies: 
\begin{itemize} 
\item[(i)] $(R_t)$ and $(C_t)$ are i.i.d.~sequences with $0<\E[R_t]<\E[C_t]<\infty$.   
\item[(ii)] There exist constants $\mu_j>0$ such that $\E[R_{i,j}\mid\calF_{i+j}]=\E[R_{i,j}\mid R_{i+j}]=(\mu_j/\mu)R_{i+j}$ for all $i$, where $\mu=\sum_{j=0}^{J_R}\mu_j$.
\end{itemize}
\end{assumption}

Assumption \ref{assump:stationary_system} (ii) holds for a wide family of stochastic models. On a sufficiently rich probability space, the requirement is essentially that $R_{i,j}$ and $R_{i+j}-R_{i,j}$ can be represented as independent increments of a non-negative L\'evy process. We refer to \cite{Sato-1999} for more on L\'evy processes. 

\begin{theorem}\label{thm:id_model}
Assume that $R_{i,j}$ and $R_{i+j}-R_{i,j}$ are independent non-negative random variables with $\E[R_{i+j}]=\mu\in (0,\infty)$ and $\E[R_{i,j}]=\mu_j\in (0,\mu)$, and assume that there exists a L\'evy process $(X_t)_{t\geq 0}$ such that $X_1=R_{i+j}$ and $X_{\mu_j/\mu}=R_{i,j}$. Then $\E[R_{i,j}\mid R_{i+j}]=(\mu_j/\mu)R_{i+j}$.   
\end{theorem}

\begin{proof}[Proof of Theorem \ref{thm:id_model}]
First, note that for any i.i.d.~random variables $Z_1,\dots,Z_n$ with sum $S_n$, by symmetry it holds that $\E[Z_1\mid S_n]=S_n/n$. 
Consider a sequence $((m_k,n_k))_{k\geq 1}$, with $(m_k,n_k)\in \Z^2$, such that $\lim_{k\to\infty}m_k/n_k=\mu_j/\mu$. 
For each $k$, let 
\begin{align*}
Z^k_{v}:=X_{v/n_k}-X_{(v-1)/n_k}, \quad v=1,\dots,n_k, \quad S^k_{u}:=\sum_{v=1}^{u}Z^k_{v}.  
\end{align*}
Note that 
\begin{align*}
\E\big[S^k_{m_k}\mid S^k_{n_k}\big]=\frac{m_k}{n_k}S^k_{n_k}=\frac{m_k}{n_k}X_1, 
\end{align*}
and by the stochastic continuity property for L\'evy processes 
\begin{align*}
S^k_{m_k}\stackrel{\P}{\to} X_{\mu_j/\mu} \text{ as } k\to\infty. 
\end{align*}
Hence, there is a subsequence $k'\to\infty$ such that $S^{k'}_{m_{k'}}\stackrel{a.s.}{\to} X_{\mu_j/\mu}$. 
By Theorem 34.2(v) in \cite{Billingsley-1995}, 
\begin{align*}
\E\big[S^{k'}_{m_{k'}}\mid S^{k'}_{n_{k'}}\big]\stackrel{a.s.}{\to} \E\big[X_{\mu_j/\mu}\mid X_1\big],
\end{align*}
from which the conclusion follows. 
\end{proof}

We now turn to the computation of backlog expectations. 
In order to avoid unnecessarily lengthy expressions we introduce the notation 
\begin{align*}
 F_t&:=R_{t}\bigg(\mathds{1}_{\{B_{t}>C_{t}\}}+\bigg(1-\frac{C_{t}-B_{t}}{R_{t}}\bigg)\mathds{1}_{\{B_{t}\leq C_{t}<B_{t}+R_{t}\}}\bigg),\\
 G_t&:=\bigg(1-\frac{C_{t}}{B_{t}}\bigg)\mathds{1}_{\{B_{t}>C_{t}\}}, 
\end{align*}
 and note that $F_t$ and $G_t$ are $\calF_t$-measurable, see Section \ref{sec:processing_procedure}. 
 
In Theorems \ref{thm:uncond_backlog_exp} and \ref{thm:cond_backlog_exp} below we derive expressions for unconditional and conditional expectations of backlogs $B_{i,j}$. From these expressions we note that the expectations are fully determined once the corresponding expectations for quantities $F_{\tau+1}$, $F_{\tau}G_{\tau+1}$, $F_{\tau}G_{\tau+1}G_{\tau+2}$, etc., can be evaluated. We will consider numerical computation of the latter expectations in Section \ref{sec:NN}. 
We use the notation $\wedge$ for the minimum $x\wedge y=\min(x,y)$. 
 
\begin{theorem}\label{thm:uncond_backlog_exp}
Assume that Assumptions \ref{assump:axioms}, \ref{assump:processing} and \ref{assump:stationary_system} hold. 
For $j\geq 0$, 
\begin{align*}
\E\big[B_{i,j}\big]=\sum_{k=1}^{j\wedge (J+1)}\frac{\mu_{k-1}}{\mu}\E\bigg[F_{i+k-1}\prod_{l=k}^{j-1}G_{i+l}\bigg],
\end{align*}
where an empty sum is equal to $0$ and an empty product is equal to $1$. 
\end{theorem}

The identity $B_{i,j+1}=B_{i,j}+R_{i,j}-P_{i,j}$ implies the following expression for $\E[P_{i,j}]$: 

\begin{corollary}\label{expected processing corollary}
Assume that Assumptions \ref{assump:axioms}, \ref{assump:processing} and \ref{assump:stationary_system} hold. 
For $j\geq 0$, 
\begin{align*}
\E\big[P_{i,j}\big]&=\sum_{k=1}^{j\wedge (J+1)}\frac{\mu_{k-1}}{\mu}\E\bigg[F_{i+k-1}\bigg(\prod_{l=k}^{j-1}G_{i+l}\bigg)\big(1-G_{i+j}\big)\bigg] \\
&\quad+\mathds{1}_{\{0\leq j\leq J\}}\frac{\mu_j}{\mu}\big(\mu-\E\big[F_{i+j}\big]\big), 
\end{align*}
where an empty sum is equal to $0$ and an empty product is equal to $1$. 
\end{corollary}


In the proof of Theorem \ref{thm:uncond_backlog_exp} we will use conditional independence properties together with the fact that products of the kind $G_{t+1}\cdots G_{t+h}$, $h\geq 1$, are $\sigma\{B_s,R_s,C_s:s>t\}$-measurable.  

\begin{lemma}\label{lem:cond_independence}
Assume that Assumptions \ref{assump:axioms}, \ref{assump:processing} and \ref{assump:stationary_system} hold. 
Then $\calH_{t}$ and $\{B_s,R_s,C_s:s>t\}$ are conditionally independent given $\calF_t$. 
\end{lemma}

\begin{proof}[Proof of Lemma \ref{lem:cond_independence}] 
Note that 
\begin{align*}
\sigma\{B_s,R_s,C_s:s>t\}=\sigma(B_{t+1})\vee \sigma\{R_s,C_s:s>t\}, 
\end{align*}
where $\sigma(B_{t+1})\subset \calF_t$ and $\sigma\{R_s,C_s:s>t\}$ are independent, and $\sigma\{R_s,C_s:s>t\}$ and $\calH_{t}$ are independent. 
Hence, $\calH_t$ and $\{B_s,R_s,C_s:s>t\}$ are dependent only through $B_{t+1}$, and therefore conditionally independent given $\calF_t$. 
\end{proof}

\begin{proof}[Proof of Theorem \ref{thm:uncond_backlog_exp}]
We prove the statement for $j=0,1,2,3$ in detail assuming $j\leq J+1$. From this, it is obvious how to repeat the argument recursively to prove the statement for larger values of $j$. 
$B_{i,0}=0$ by definition. For $j=1$, 
\begin{align*}
\E\big[B_{i,1}\big]&=\E\big[\E\big[B_{i,1}\mid\calF_{i}\vee \calG_{i}\big]\big]=\E\bigg[\frac{R_{i,0}}{R_i}F_{i}\bigg]\\
&=\E\bigg[\E\bigg[\frac{R_{i,0}}{R_i}F_{i}\,\bigg|\,\calF_{i}\bigg]\bigg]
=\E\bigg[\frac{\E[R_{i,0}\mid\calF_{i}]}{R_i}F_{i}\bigg]=\frac{\mu_{0}}{\mu}\E\big[F_{i}\big].
\end{align*}
For $j=2$, 
\begin{align*}
\E\big[B_{i,2}\big]=\E\big[\E\big[B_{i,2}\mid\calF_{i+1}\vee \calG_{i+1}\big]\big]=\E\big[B_{i,1}G_{i+1}\big]+\E\bigg[\frac{R_{i,1}}{R_{i+1}}F_{i+1}\bigg],
\end{align*}
where for the newly reported claims 
\begin{align*}
\E\bigg[\frac{R_{i,1}}{R_{i+1}}F_{i+1}\bigg]=\E\bigg[\E\bigg[\frac{R_{i,1}}{R_{i+1}}F_{i+1}\,\bigg|\,\calF_{i+1}\bigg]\bigg]
=\frac{\mu_{1}}{\mu}\E\big[F_{i+1}\big], 
\end{align*}
and for the previous backlog
\begin{align}
\begin{split}
\E\big[B_{i,1}G_{i+1}\big]&=\E\big[\E\big[B_{i,1}G_{i+1}\mid \calF_{i}\big]\big] \label{eq:cond_indep}
=\E\big[\E\big[B_{i,1}\mid \calF_{i}\big]\E\big[G_{i+1}\mid \calF_{i}\big]\big] \\
&=\E\big[\E\big[\E\big[B_{i,1}\mid \calF_{i}\vee \calG_{i}\big]\mid\calF_{i}\big]\E\big[G_{i+1}\mid \calF_{i}\big]\big] \\
&=\E\bigg[\E\bigg[\frac{R_{i,0}}{R_i}F_{i}\,\bigg|\,\calF_{i}\bigg]\E\bigg[G_{i+1}\,\bigg|\,\calF_{i}\bigg]\bigg] \\
&=\E\bigg[\E\bigg[\frac{R_{i,0}}{R_i}\,\bigg|\,\calF_{i}\bigg]\E\bigg[F_{i}G_{i+1}\,\bigg|\,\calF_{i}\bigg]\bigg]=\frac{\mu_{0}}{\mu}\E\big[F_{i}G_{i+1}\big],
\end{split}
\end{align}
where we used Lemma \ref{lem:cond_independence} noting that $B_{i,1}$ is $\calH_{i}$-measurable. 
For $j=3$, 
\begin{align*}
\E\big[B_{i,3}\big]=\E\big[\E\big[B_{i,3}\mid\calF_{i+2}\vee \calG_{i+2}\big]\big]=\E\big[B_{i,2}G_{i+2}\big]+\E\bigg[\frac{R_{i,2}}{R_{i+2}}F_{i+2}\bigg],
\end{align*}
where 
\begin{align*}
\E\bigg[\frac{R_{i,2}}{R_{i+2}}F_{i+2}\bigg]=\E\bigg[\E\bigg[\frac{R_{i,2}}{R_{i+2}}F_{i+2}\,\bigg|\,\calF_{i+2}\bigg]\bigg]=\frac{\mu_{2}}{\mu}\E\big[F_{i+2}\big],
\end{align*}
and, similarly to \eqref{eq:cond_indep},  
\begin{align*}
\E\big[B_{i,2}G_{i+2}\big]&=\E\big[\E\big[\E\big[B_{i,2}\mid\calF_{i+1}\vee \calG_{i+1}\big]\mid\calF_{i+1}\big]\E\big[G_{i+2}\mid\calF_{i+1}\big]\big]\\
&=\E\bigg[\E\bigg[B_{i,1}G_{i+1}+\frac{R_{i,1}}{R_{i+1}}F_{i+1}\,\bigg|\,\calF_{i+1}\bigg]\E\big[G_{i+2}\mid\calF_{i+1}\big]\bigg] \\
&=\E\big[B_{i,1}G_{i+1}G_{i+2}\big]+\E\bigg[\frac{R_{i,1}}{R_{i+1}}F_{i+1}G_{i+2}\bigg], 
\end{align*}
where we used Lemma \ref{lem:cond_independence}.   
The last two terms are simplified as follows
\begin{align*}
\E\bigg[\frac{R_{i,1}}{R_{i+1}}F_{i+1}G_{i+2}\bigg]
&=\E\bigg[\E\bigg[\frac{R_{i,1}}{R_{i+1}}F_{i+1}G_{i+2}\,\bigg|\,\calF_{i+1}\bigg]\bigg]\\
&=\E\bigg[\E\bigg[\frac{R_{i,1}}{R_{i+1}}F_{i+1}\,\bigg|\,\calF_{i+1}\bigg]\E\bigg[G_{i+2}\,\bigg|\,\calF_{i+1}\bigg]\bigg]\\ 
&=\E\bigg[\E\bigg[\frac{R_{i,1}}{R_{i+1}}\,\bigg|\,\calF_{i+1}\bigg]\E\bigg[F_{i+1}G_{i+2}\,\bigg|\,\calF_{i+1}\bigg]\bigg]\\ 
&=\frac{\mu_{1}}{\mu}\E\big[F_{i+1}G_{i+2}\big], 
\end{align*}
where we used Lemma \ref{lem:cond_independence}. 
For the remaining term we recall that $B_{i,1}$ is $\calH_{i}$-measurable and use Lemma \ref{lem:cond_independence} to receive
\begin{align*}
\E\big[B_{i,1}G_{i+1}G_{i+2}\big]&=\E\big[\E\big[B_{i,1}G_{i+1}G_{i+2}\mid\calF_{i}\big]\big]\\
&=\E\big[\E\big[\E\big[B_{i,1}\mid\calF_{i}\vee \calG_{i}\big]\mid\calF_{i}\big]\E\big[G_{i+1}G_{i+2}\mid\calF_{i}\big]\big]\\
&=\E\bigg[\E\bigg[\frac{R_{i,0}}{R_{i}}F_{i}\,\bigg|\,\calF_{i}\bigg]\E\big[G_{i+1}G_{i+2}\mid\calF_{i}\big]\bigg]\\
&=\E\bigg[\E\bigg[\frac{R_{i,0}}{R_{i}}\,\bigg|\,\calF_{i}\bigg]\E\big[F_{i}G_{i+1}G_{i+2}\mid\calF_{i}\big]\bigg]\\ 
&=\frac{\mu_{0}}{\mu}\E\big[F_{i}G_{i+1}G_{i+2}\big].
\end{align*}
For $j>3$ the statement follows by reusing the above arguments.  
\end{proof}
  
Theorem \ref{thm:uncond_backlog_exp} considered unconditional backlog expectations. Below follows the corresponding result for conditional expectations.  
 
\begin{theorem}\label{thm:cond_backlog_exp}
Assume that Assumptions \ref{assump:axioms}, \ref{assump:processing} and \ref{assump:stationary_system} hold. 
For $k\geq 0$ such that $\tau-i+k\geq 0$, 
\begin{align}
&\E\big[B_{i,\tau-i+k+1}\mid\calG_{\tau}\big] \nonumber \\
&\quad=\sum_{m=1}^{k}\mathds{1}_{\{i\leq \tau+m\leq i+J\}}\frac{\mu_{\tau-i+m}}{\mu}\E\bigg[F_{\tau+m}\prod_{l=m+1}^{k}G_{\tau+l}\,\bigg|\,\calE_{\tau}\bigg] \label{eq:thm_cond_be_t1}\\
&\quad\quad+\mathds{1}_{\{i\leq \tau\leq i+J\}}\frac{R_{i,\tau-i}}{R_{\tau}}\E\big[F_{\tau}\mid\calE_{\tau}\big]\E\bigg[\prod_{l=1}^{k}G_{\tau+l}\,\bigg|\, \calE_{\tau}\bigg] \label{eq:thm_cond_be_t2}\\
&\quad\quad+\mathds{1}_{\{i\leq \tau\}}B_{i,\tau-i}\E\big[G_{\tau}\mid\calE_{\tau}\big]\E\bigg[\prod_{l=1}^{k}G_{\tau+l}\,\bigg|\, \calE_{\tau}\bigg], \label{eq:thm_cond_be_t3}
\end{align}
where an empty sum is equal to $0$ and an empty product is equal to $1$. 
We can equally replace all conditions $\calE_\tau$ by $\calG_\tau$.
\end{theorem}
 
 \begin{proof}[Proof of Theorem \ref{thm:cond_backlog_exp}]
 The statement is proved by the same arguments as in the proof of Theorem \ref{thm:uncond_backlog_exp}.
 \end{proof}
 
 \begin{remark}
 For past occurrence periods, $i\leq\tau$, all three terms  
 \eqref{eq:thm_cond_be_t1}, \eqref{eq:thm_cond_be_t2}, \eqref{eq:thm_cond_be_t3} contribute to the conditional backlog expectation. 
 For future occurrence periods, $i>\tau$, the terms \eqref{eq:thm_cond_be_t2} and \eqref{eq:thm_cond_be_t3} vanish.  
 \end{remark}
 
 \begin{remark}\label{rem:constant_capacity}
 If $C_t=c$ is constant, then $F_{t}$ and $G_{t}$ are fully determined by $B_t$ and $R_t$. Hence, $\E\big[F_{\tau}\mid\calE_{\tau}\big]=F_{\tau}$ and $\E\big[G_{\tau}\mid\calE_{\tau}\big]=G_{\tau}$. Moreover, since $(R_t)$ is an i.i.d.~sequence, if $C_t=c$ is constant, then $B_{t+1}$ is fully determined by $B_t$ and $R_t$, and 
 \begin{align*}
 \E\bigg[F_{\tau+m}\prod_{l=m+1}^{k}G_{\tau+l}\,\bigg|\,\calE_{\tau}\bigg]&=\E\bigg[F_{\tau+m}\prod_{l=m+1}^{k}G_{\tau+l}\,\bigg|\, B_{\tau+1}\bigg], \\
 \E\bigg[\prod_{l=1}^{k}G_{\tau+l}\,\bigg|\, \calE_{\tau}\bigg]&=\E\bigg[\prod_{l=1}^{k}G_{\tau+l}\,\bigg|\, B_{\tau+1}\bigg].
 \end{align*}
 \end{remark}

Remark \ref{rem:constant_capacity} implies the following corollary to Theorem \ref{thm:cond_backlog_exp}.

\begin{corollary}\label{cor:cond_backlog_exp}
Assume that Assumptions \ref{assump:axioms}, \ref{assump:processing} and \ref{assump:stationary_system} hold. Assume that there exists a non-random $c>0$ such that $C_t=c$ for all $t$. For $k\geq 0$ such that $\tau-i+k\geq 0$, 
\begin{align*}
\E\big[B_{i,\tau-i+k+1}\mid\calG_{\tau}\big]
&=\sum_{m=1}^{k}\mathds{1}_{\{i\leq \tau+m\leq i+J\}}\frac{\mu_{\tau-i+m}}{\mu}\E\bigg[F_{\tau+m}\prod_{l=m+1}^{k}G_{\tau+l}\,\bigg|\, B_{\tau+1}\bigg] \\
&\quad+\mathds{1}_{\{i\leq \tau\leq i+J\}}\frac{R_{i,\tau-i}}{R_{\tau}}F_{\tau}\E\bigg[\prod_{l=1}^{k}G_{\tau+l}\,\bigg|\, B_{\tau+1}\bigg] \\
&\quad+\mathds{1}_{\{i\leq \tau\}}B_{i,\tau-i}G_{\tau}\E\bigg[\prod_{l=1}^{k}G_{\tau+l}\,\bigg|\, B_{\tau+1}\bigg], 
\end{align*}
where an empty sum is equal to $0$ and an empty product is equal to $1$. 
\end{corollary}

\section{The negative binomial model}\label{sec:neg_bin_model} 
It may seem natural to consider a Poisson model for the number of reported claims, where all $R_{i,j}$ are independent and $R_{i,j}\sim \Pois(\mu_j)$.
Such a model is consistent with Assumption \ref{assump:stationary_system}. However, as explained in Remark \ref{rem:poisson_backlog}, for constant capacity $C_t=c$ and $\mu<c$ reasonably large, say $1000$, the expected total size of the backlog will be close to $0$ unless $c\approx \mu$. The event $\{B_t>c\}$, which appears in expected backlog calculations in Section \ref{sec:comp_backlog_exp}, is a very unlikely event for the Poisson model. Markov's inequality together with Remark \ref{rem:poisson_backlog} give
\begin{align*}
\P[B_t>c]\leq \frac{1}{c/\mu-1}\, \frac{1}{2c}\approx 0. 
\end{align*}
Thus, long-term backlogs are not an issue in Poisson settings.
We need a model for the numbers of reported claims with considerable over-dispersion compared to the Poisson model to have
an interesting backlog behavior.

We start with a generic negative binomial random variable $R$. 
Assume that $R$ is conditionally Poisson distributed with conditional
mean $\Lambda$, and $\Lambda \sim \Gamma(\alpha, \beta)$. It follows
that $R$ has an unconditional negative binomial distribution $\NegBin(\alpha,\beta)$ with moment
generating function
\begin{align*}
\E[\exp \{x R \}]= \E\big[\exp \big\{\Lambda (e^x-1)\big\}\big]
=
\bigg(\frac{\beta}{\beta - (e^x-1)}\bigg)^\alpha,
\quad \text{for $x<\log(\beta+1)$.}
\end{align*}
From this it follows that we can aggregate independent negative binomial
random variables $R_{i,0}, \ldots, R_{i,J}$ (or $R_{t,0}, R_{t-1,1}, \ldots, R_{t-J,J}$) as long as they share the same scale
parameter $\beta$, and we remain in the family of negative binomial random variables
\begin{align*}
\E\bigg[\exp \bigg\{x \sum_{j=0}^{J}R_{i,j} \bigg\}\bigg]=\prod_{j=0}^{J} \E[\exp \{x R_{i,j}\}]
=\bigg(\frac{\beta}{\beta - (e^x-1)}\bigg)^{\sum_{j=0}^{J}\alpha_j}.
\end{align*}
We consider the following model for the number of reported claims: all $R_{i,j}$ are independent with $R_{i,j}\sim \NegBin(\alpha_j,\beta)$. 
Hence, 
\begin{align*}
R_t := \sum_{j=0}^J R_{t-j,j} \eqdis \sum_{j=0}^{J} R_{i,j} \sim \NegBin(\alpha,\beta), \quad \alpha:=\sum_{j=0}^J \alpha_j, 
\end{align*}
with means and variances
\begin{align*}
\E[R_t] = {\alpha}/{\beta}
\quad \text{ and } \quad
\Var[R_t] = {\alpha}/{\beta}+{\alpha}/{\beta^2}
=\E[R_t] \big( 1 + 1/\beta\big). 
\end{align*}
The over-dispersion term $1/\beta$ distinguishes the negative binomial model
from the Poisson model. In practical applications, it is often reasonable
to assume that the model has a coefficient of variation roughly on the unit scale,
typically, it is bigger for claim size modeling than for claim counts modeling.
The coefficient of variation is in this negative binomial model given by
\begin{align*}
\frac{\sqrt{\Var[R_t]}}{\E[R_t]}
=\sqrt{\frac{\beta}{\alpha}+\frac{1}{\alpha}}
=\sqrt{\frac{1}{\E[R_t]}+\frac{1/\beta}{\E[R_t]}}.
\end{align*}
This requires that $\beta>0$ lives on the same scale as $1/\E[R_t]$.
We select $\alpha=2$ and $\beta=2/1000$, this gives
an expected value of $\mu=\E[R_t]=1000$ and a coefficient of variation $\approx 0.7$. 
We select $J=3$, and we split the expected numbers of reported claims according to 
\begin{align}\label{expected claims split}
(\E[R_{t,0}], \E[R_{t,1}], \E[R_{t,2}], \E[R_{t,3}])
&=(\alpha_0, \alpha_1, \alpha_2, \alpha_3)/\beta
\\&=(500, 300, 150, 50). \nonumber
\end{align}
We select a constant period capacity $C_t=c$ with $c=\eta\mu$, where the capacity ratio $\eta$ is chosen as $\eta=1.2$.  

The following figures show the results from this negative binomially generated
claims reportings data $R_t$, $1\le t \le T$, over the
$T=120$ (monthly) time periods. 
We start the process with a zero backlog $B_1 = 0$.

\begin{figure}[htb!]
\begin{center}
\begin{minipage}[t]{0.48\textwidth}
\begin{center}
\includegraphics[width=\textwidth]{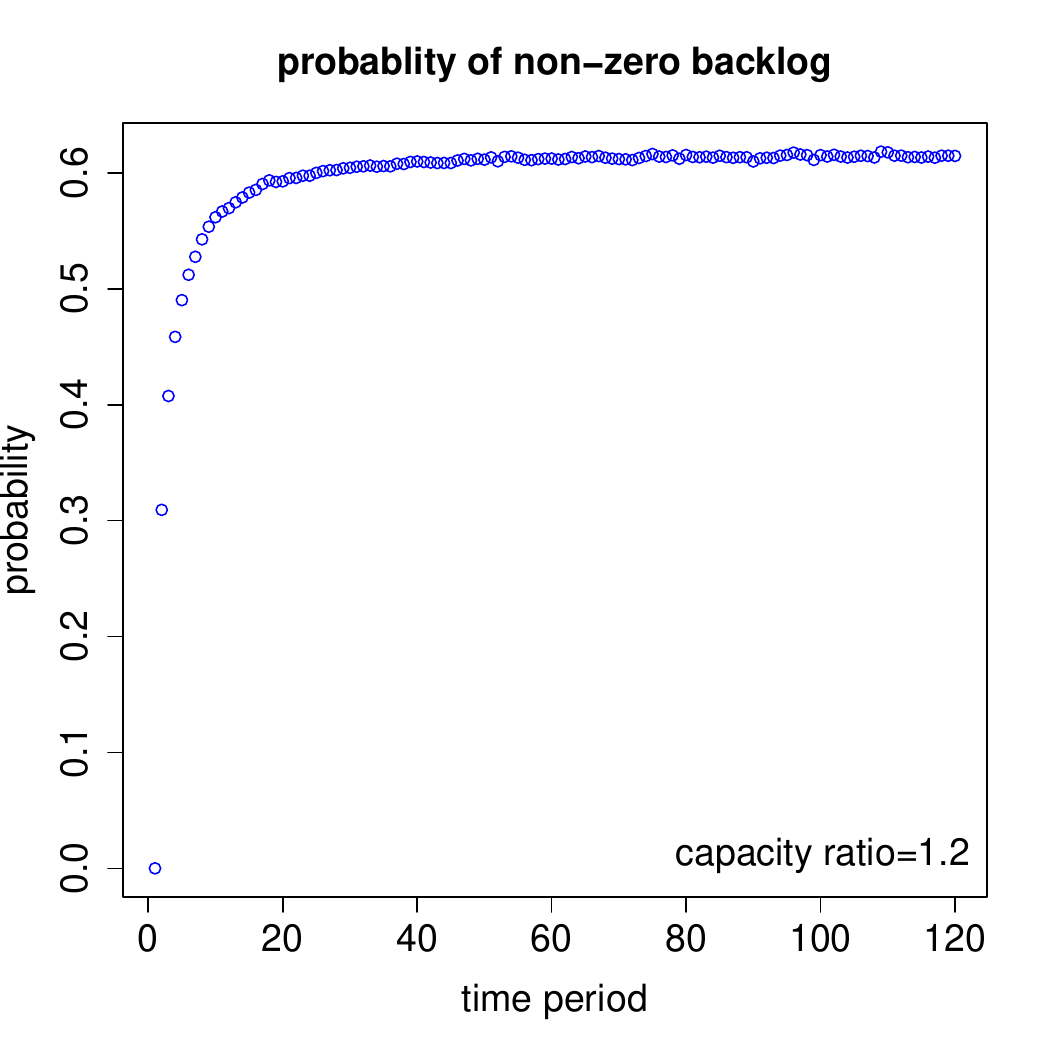}
\end{center}
\end{minipage}
\begin{minipage}[t]{0.48\textwidth}
\begin{center}
\includegraphics[width=\textwidth]{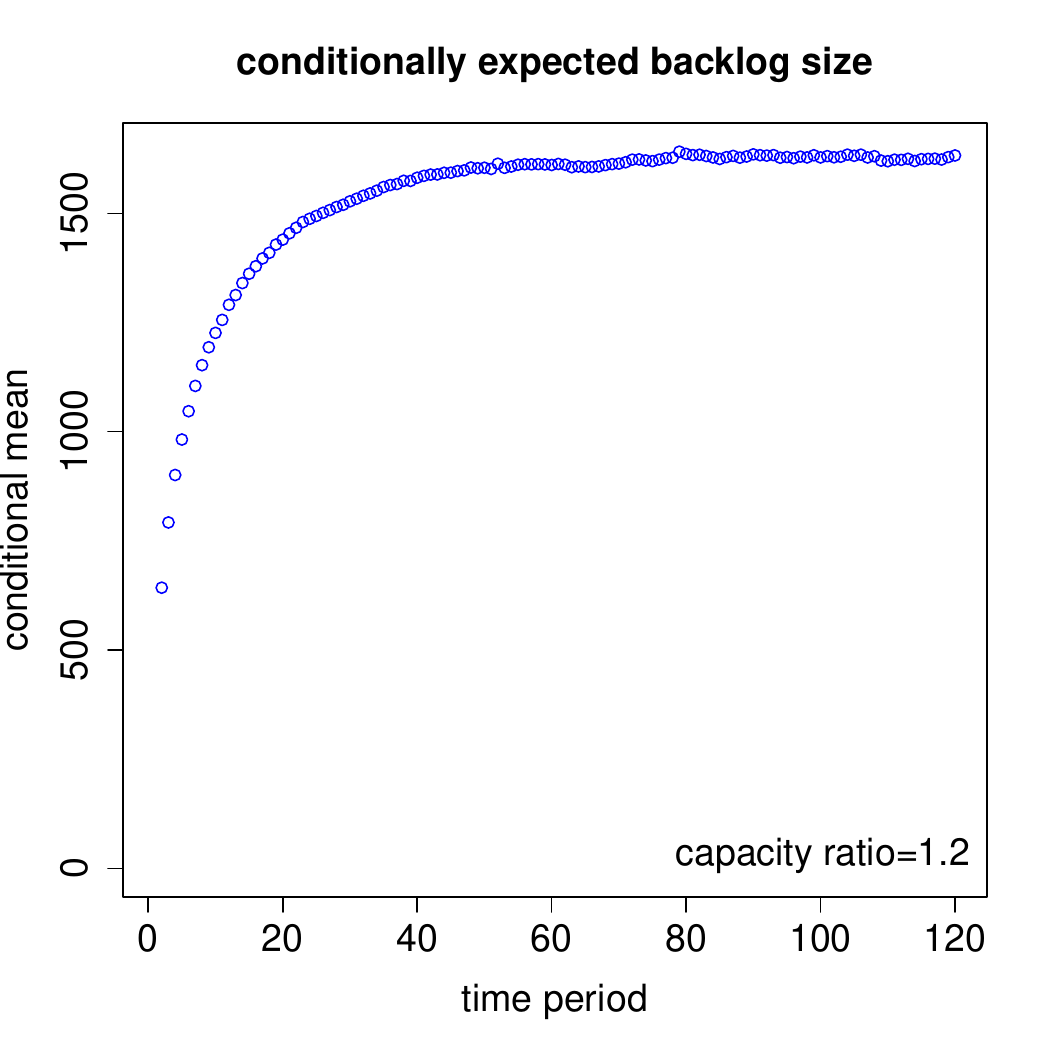}
\end{center}
\end{minipage}
\end{center}
\vspace{-.65cm}
\caption{Negative binomial model with capacity ratio $\eta=1.2$:
(lhs) probabilities of a non-zero backlog $\P[B_t>0 \mid B_1=0]$ for
different time periods $1\le t \le T$; (rhs) conditionally expected
backlog sizes $\E[B_t\mid B_t>0, B_1=0]$.}
\label{NB backlog analysis 120}
\end{figure}

Figure \ref{NB backlog analysis 120} (lhs) shows the
probabilities of a non-zero backlog when starting 
with a zero backlog, $\P[B_t>0 \mid B_1=0]$, for different time periods
$1\le t \le T$, and the corresponding conditionally expected backlog sizes
$\E[B_t \mid B_t>0, B_1=0]$ for a capacity ratio of $\eta=1.2$ in this negative
binomial model. 
The conditionally expected long-term backlog $\lim_{t \to \infty}\E[B_t \mid B_t>0, B_1=0]$
is of magnitude 1600, see Figure \ref{NB backlog analysis 120} (rhs). This implies that we have an expected long-term backlog
of $\lim_{t\to \infty}\E[B_t\mid B_1=0]\approx 1000$, which is just slightly below the selected capacity 
$c=\eta \mu =1200$. This results in frequent carry forward of old backlogs.

\begin{figure}[htb!]
\begin{center}
\begin{minipage}[t]{0.48\textwidth}
\begin{center}
\includegraphics[width=\textwidth]{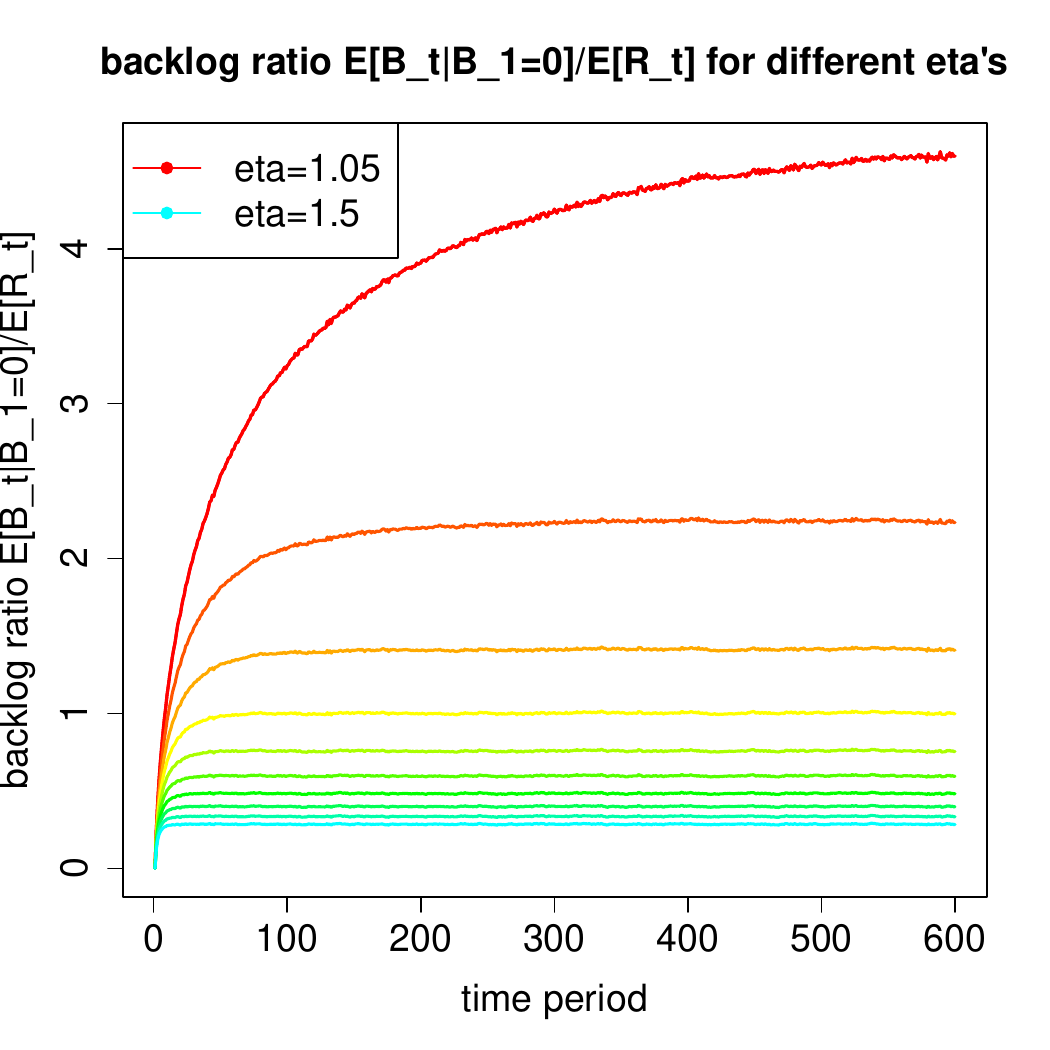}
\end{center}
\end{minipage}
\end{center}
\vspace{-.65cm}
\caption{Relative expected backlogs $\E[B_t\mid B_1=0]/\E[R_t]$
for capacity ratios $\eta\in \{1.05, 1.10, \ldots, 1.50\}$ if one starts
with a zero backlog.}
\label{development of processing}
\end{figure}

Figure \ref{development of processing} shows the
relative expected backlogs $\E[B_t\mid B_1=0]/\E[R_t]$ if we start the
process with a zero backlog $B_1=0$. The different colors correspond
to capacity ratios $\eta\in\{1.05, 1.10, \ldots, 1.50\}$. For a capacity
ratio $\eta=1.1$ the average backlog is roughly twice the expected
number of reported claims, and for $\eta=1.2$ we have a factor
of roughly 1.

\begin{figure}[htb!]
\begin{center}
\begin{minipage}[t]{0.48\textwidth}
\begin{center}
\includegraphics[width=\textwidth]{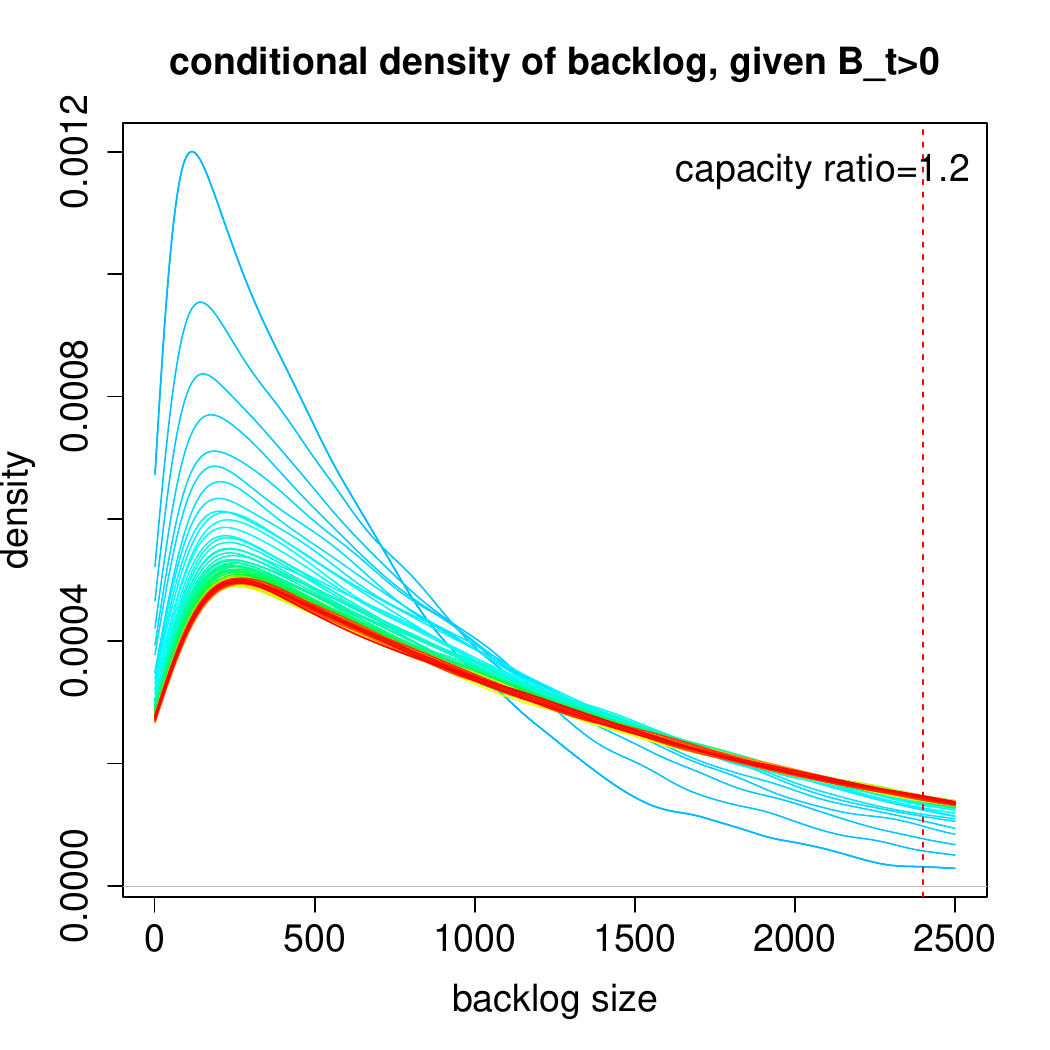}
\end{center}
\end{minipage}
\begin{minipage}[t]{0.48\textwidth}
\begin{center}
\includegraphics[width=\textwidth]{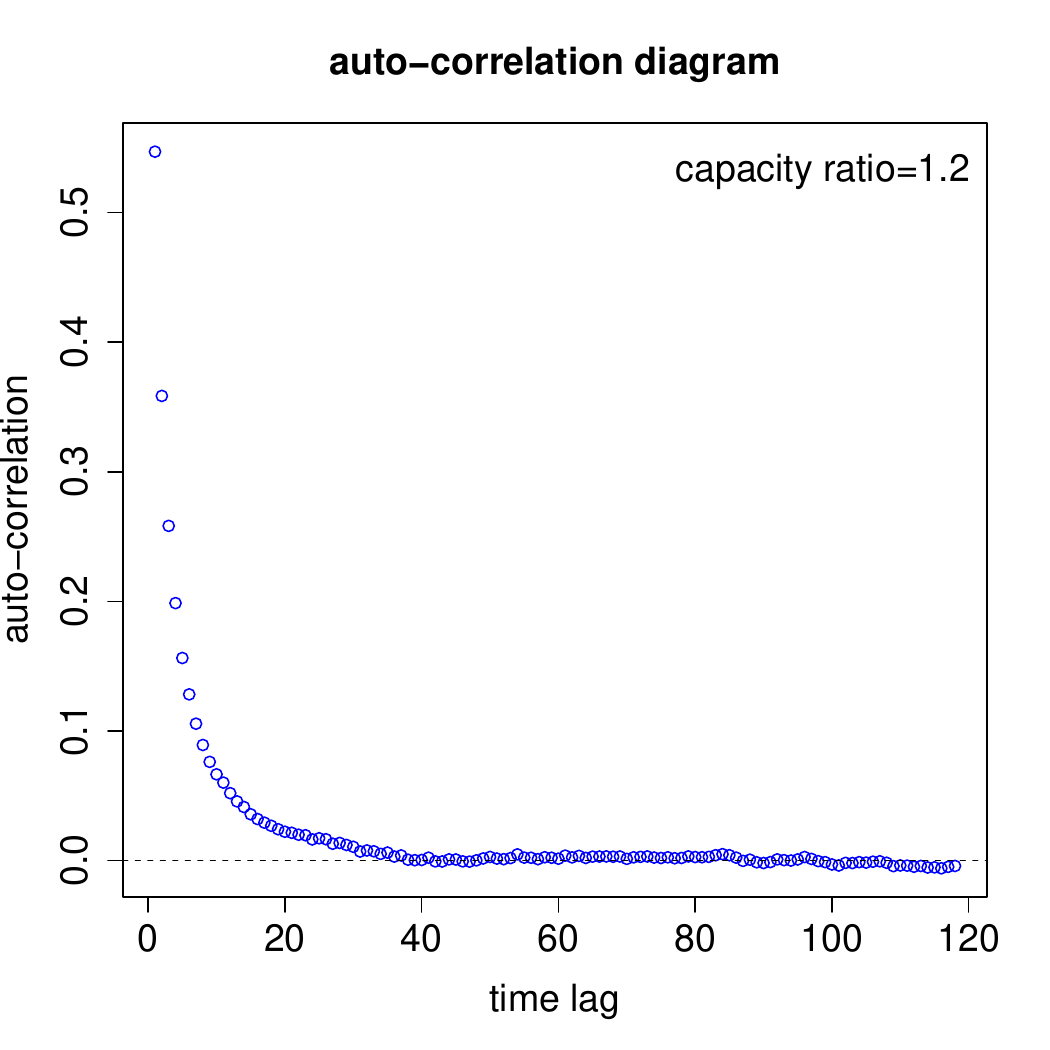}
\end{center}
\end{minipage}
\end{center}
\vspace{-.65cm}
\caption{Negative binomial model: (lhs) empirical densities
of the backlogs, conditioned on $\{B_t>0, B_1=0\}$ for
different time periods $2\le t \le T$; (rhs) auto-correlation diagram of backlogs
$B_2$ and $B_{2+s}$ for time lags $s\ge 1$ when starting with
a zero backlog $B_1=0$.}
\label{NB autocor 101}
\end{figure}

Figure \ref{NB autocor 101} (lhs) gives the empirical conditional densities
of the backlog $B_t$, conditioned on $\{B_t>0, B_1=0\}$, for $2\le t \le T$. The light-blue
density with the highest maximum corresponds to $t=2$, and with increasing time
$t$ we follow the rainbow colors (from light-blue to red). We observe that this
conditional backlog size has still a significant positive probability that exceeds
twice the capacity $2c$ in the stationary limit (vertical red dotted line). This indicates
that we likely have carry forwards of backlogs over multiple periods.
Figure \ref{NB autocor 101} (rhs) shows the auto-correlation diagram of
backlogs $B_2$ and $B_{2+s}$ for time lags $s\ge 1$ when starting with
a zero backlog $B_1=0$. This auto-correlation
$\operatorname{Corr}(B_2,B_{2+s}\mid B_1=0)$
has vanished after 40 periods.

\begin{figure}[htb!]
\begin{center}
\begin{minipage}[t]{0.48\textwidth}
\begin{center}
\includegraphics[width=\textwidth]{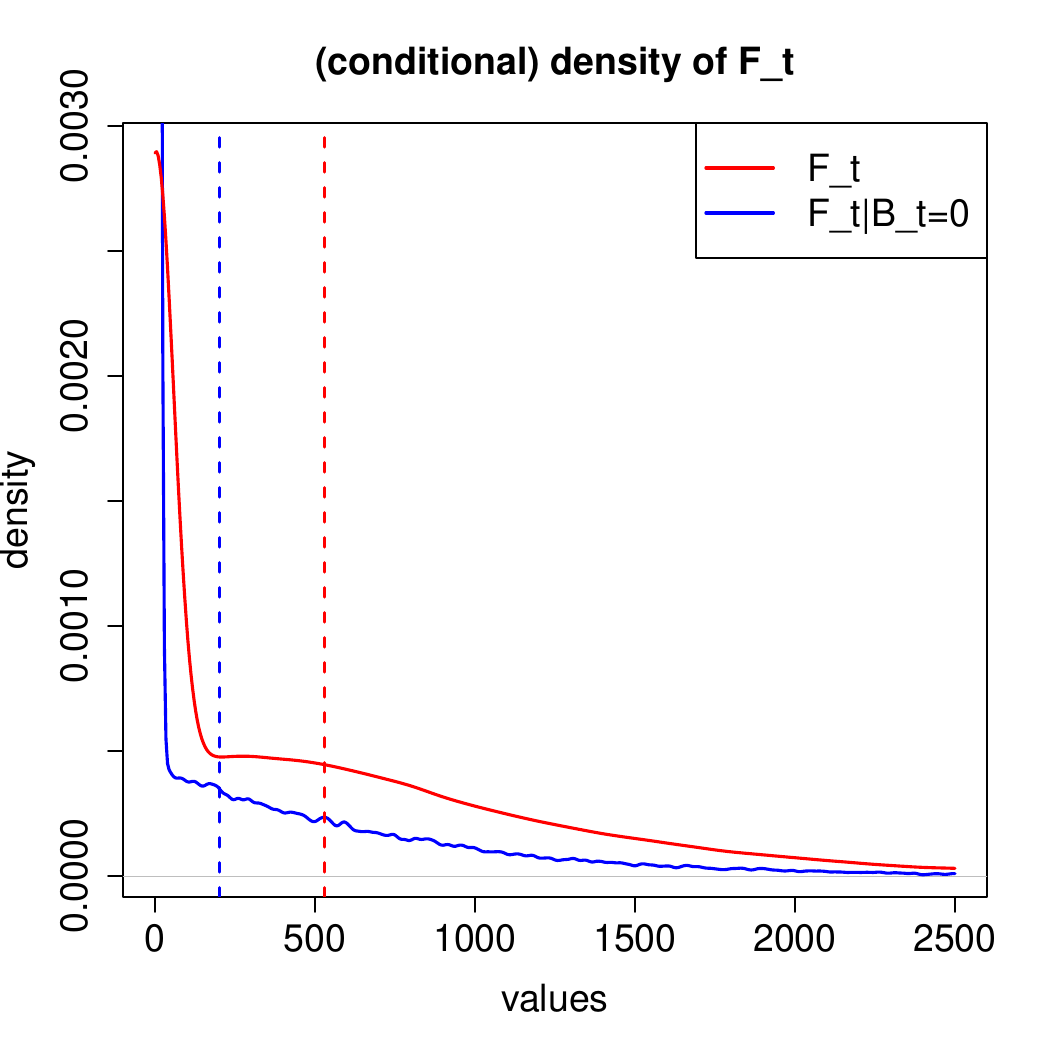}
\end{center}
\end{minipage}
\begin{minipage}[t]{0.48\textwidth}
\begin{center}
\includegraphics[width=\textwidth]{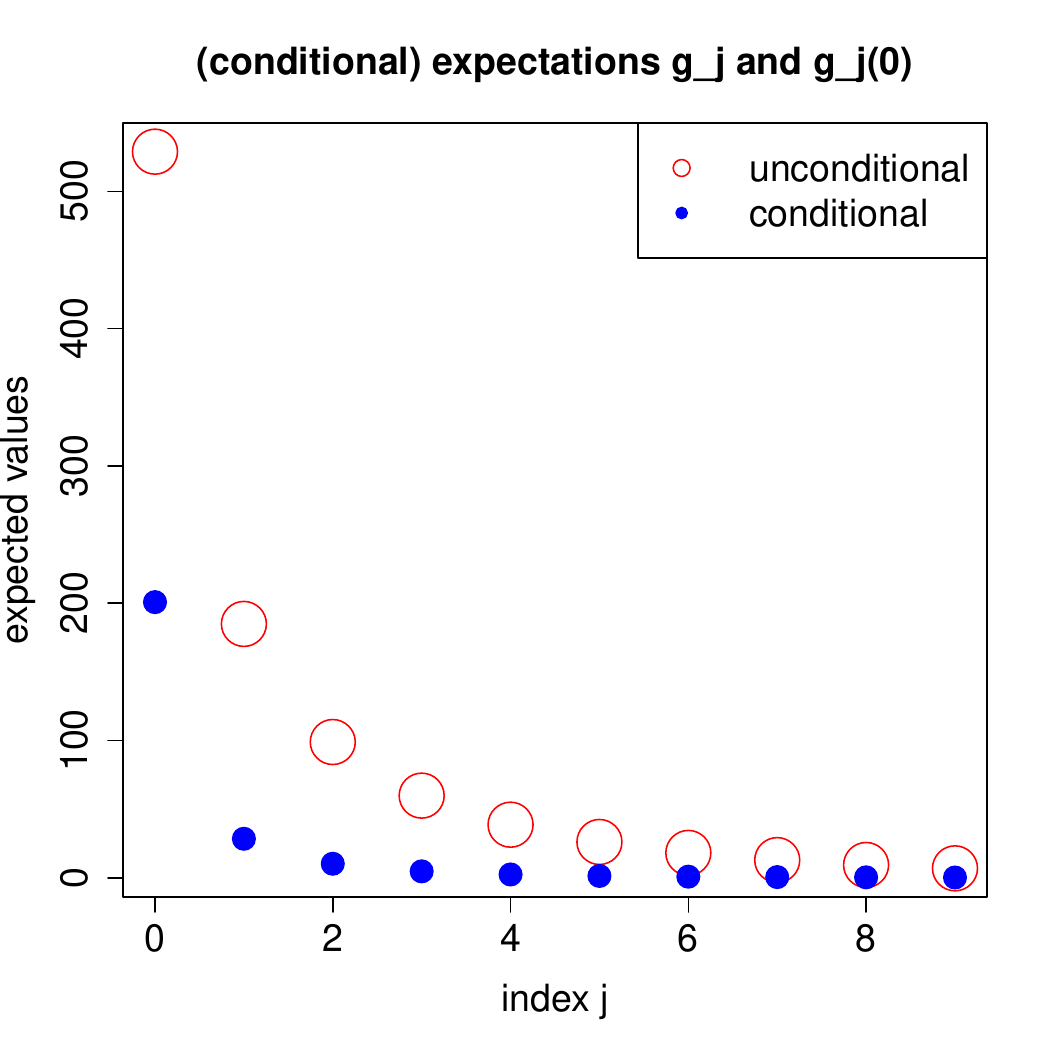}
\end{center}
\end{minipage}
\end{center}
\vspace{-.65cm}
\caption{Negative binomial model: (lhs) (conditionally) densities
and expected
values $\E[F_t]$ and $\E[F_t\mid B_t=0]$;
(rhs) (conditional) expectations given
in equation \eqref{G expectations} for $0 \le j \le 9$, where
$j=0$ corresponds to the empty product only considering $F_t$;
red refers to the unconditional case averaging over the stationary distribution
for the initial value $B_t$, and blue corresponds to starting the process
in a zero backlog $B_t=0$.}
\label{NB density Ft 101}
\end{figure}

Figure \ref{NB density Ft 101} (lhs) shows the empirical density of $F_t$ and the conditional empirical density of $F_t$, given $B_t=0$. The latter is directly simulated, given $B_t=0$, using the aggregate negative binomial model assumption for $R_t$. For the former, we select $B_t$ from the stationary limit
distribution of Figure \ref{NB autocor 101} (lhs), and then simulate $F_t$ conditional on this value. The vertical dotted lines show
the (conditionally) expected values $\E[F_t]$ and $\E[F_t\mid B_t=0]$ in red and blue.
 
This carries over to Figure \ref{NB density Ft 101} (rhs) which shows the (conditional) expectations $g_j$ and $g_j(0)$ of $F_t \prod_{l=t+1}^{t+j} G_l$, where 
\begin{align}\label{G expectations}
g_j:=\E \left[F_t \prod_{l=t+1}^{t+j} G_l \right]
\quad \text{ and } \quad 
g_j(b):=\E \bigg[F_t \prod_{l=t+1}^{t+j} G_l\,\bigg|\, B_t=b \bigg].
\end{align}
We obtain expected values that are significantly bigger than zero for small $j$'s, and this precisely differentiates the backlog behavior in the negative binomial model from the Poisson case. 

\section{Neural network approximation}\label{sec:NN}
To take into account backlog costs in cost optimization problems, we need to be able to study the sensitivities of the functions $(g_j(\cdot))_{j\ge 0}$ defined in \eqref{G expectations} in the capacity ratio parameter $\eta>1$. 
We therefore slightly modify the corresponding notation by adding upper indices $^{(\eta)}$ to the variables that directly depend on the selected capacity ratio $\eta>1$. 
In particular, we consider the following quantity for backlog computations
\begin{align}\label{G expectations 2}
g_j(b;\eta)=\E\bigg[F_1^{(\eta)} \prod_{l=2}^{j+1} G^{(\eta)}_{l}\,\bigg|\, B^{(\eta)}_1=b \bigg],
\end{align}
for $j\ge 0$, with an empty product being set to one.

Since the quantities \eqref{G expectations 2} cannot be computed explicitly, as a function of $\eta$, we fit a recurrent neural network (RNN) $\bz^{\rm RNN}_{\theta}: \R^{2}\to \R^T$ that takes the inputs $(b,\eta)$ to approximate the function $(b,\eta) \mapsto (g_j(b;\eta))_{j=0}^{T-1}\in \R^T$, for a fixed large $T$.
The fitted RNN $\bz^{\rm RNN}_{\widehat{\theta}}$ is obtained by minimizing the square loss in network parameter $\theta$
\begin{align}\label{gradient descent 1st}
\widehat{\theta} ~ \in ~
\underset{\theta}{\arg\min}\,
\frac{1}{n} \sum_{k=1}^n
\sum_{j=0}^{T-1}
\bigg(F_1^{k} \prod_{l=2}^{j+1} G_l^{k} - \bz^{\rm RNN}_{\theta}(B_1^k,\eta^k)_j \bigg)^2,
\end{align}
where the observations $(F_1^{k}, G^{k}_2, \ldots, G_{T-1}^{k})$, $1\le k \le n$, are simulated by using i.i.d.~randomized uniform capacity ratios $\eta^k \in (1.05,1.50)$, randomized initial backlogs $B_1^k$ (from the stationary limit distribution corresponding to the simulated $\eta^k$), simulated reporting processes $(R^k_t)_{t=1}^T$, and the resulting backlog processes $(B_t^k)_{t=1}^T$ for the simulated capacity ratios $\eta^k$. 
From this we compute the observations  $(F_1^{k}, G^{k}_2, \ldots, G_{T-1}^{k})$, which then enter the square loss \eqref{gradient descent 1st}.

\begin{figure}[htb!]
\begin{center}
\begin{minipage}[t]{0.48\textwidth}
\begin{center}
\includegraphics[width=\textwidth]{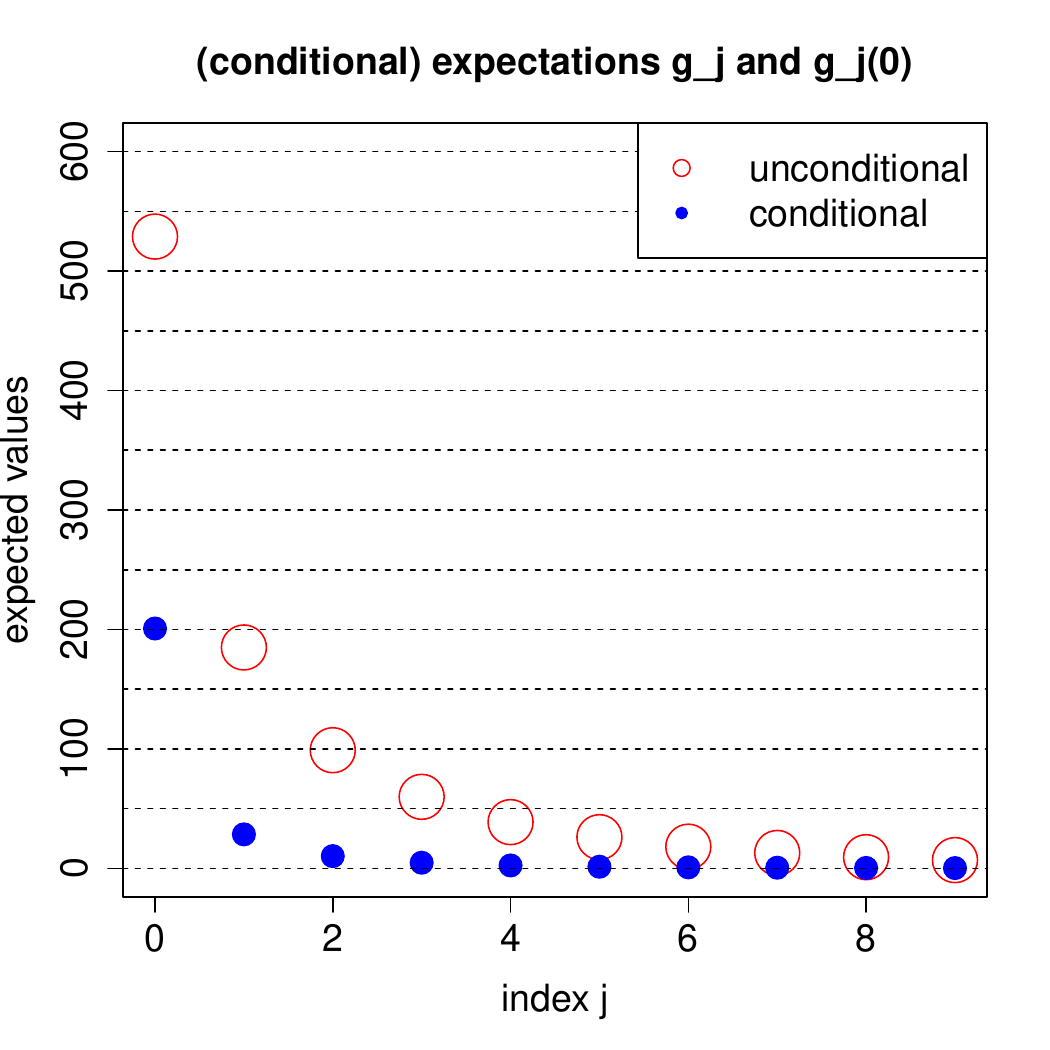}
\end{center}
\end{minipage}
\begin{minipage}[t]{0.48\textwidth}
\begin{center}
\includegraphics[width=\textwidth]{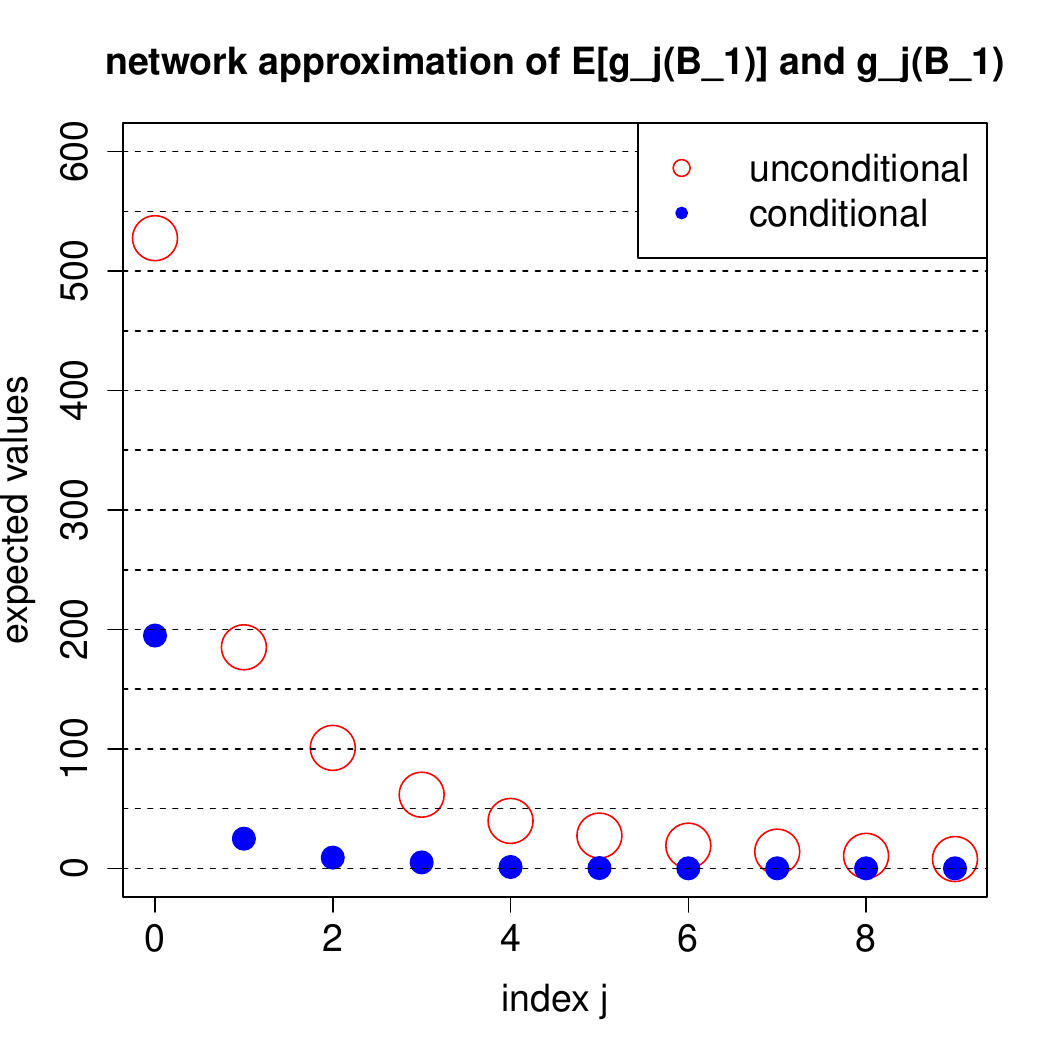}
\end{center}
\end{minipage}
\end{center}
\vspace{-.65cm}
\caption{(Conditional) expectations
$\E[g_j(B^{(\eta)}_1;\eta)]$ and $g_j(0;\eta)$, $0\le j\le 9$, for capacity
ratio $\eta=1.2$: (lhs) empirical means taken from Figure
\ref{NB density Ft 101} (rhs), and (rhs) RNN approximations $\bz^{\rm RNN}_{\widehat{\theta}}$.}
\label{NB network approximation}
\end{figure}
 
Figure \ref{NB network approximation} (rhs) shows the results of the fitted RNN  $\bz^{\rm RNN}_{\widehat{\theta}}$ approximation, and they are compared to the empirical means on the left-hand side in this figure; these are taken from Figure \ref{NB density Ft 101} (rhs), and both figures have the same $y$-scale and the same selected capacity ratio $\eta=1.2$. For the conditional version we start with a zero backlog $B^{(\eta)}_1=0$, thus, we consider $\bz^{\rm RNN}_{\widehat{\theta}}(0,\eta)$, and for the unconditional version $\E[\bz^{\rm RNN}_{\widehat{\theta}}(B^{(\eta)}_1,\eta)]$ we average over the stationary limit distribution for $B^{(\eta)}_1$ that corresponds to the capacity ratio $\eta=1.2$; note that this averaging is done purely empirically by selecting stationary samples from the backlogs.
From this figure we conclude that the RNN approximations are very close to the empirical means. 
Thus, overall the RNN approximations $\bz^{\rm RNN}_{\widehat{\theta}}$ seem very accurate. 

\begin{figure}[htb!]
\begin{center}
\begin{minipage}[t]{0.4\textwidth}
\begin{center}
\includegraphics[width=\textwidth]{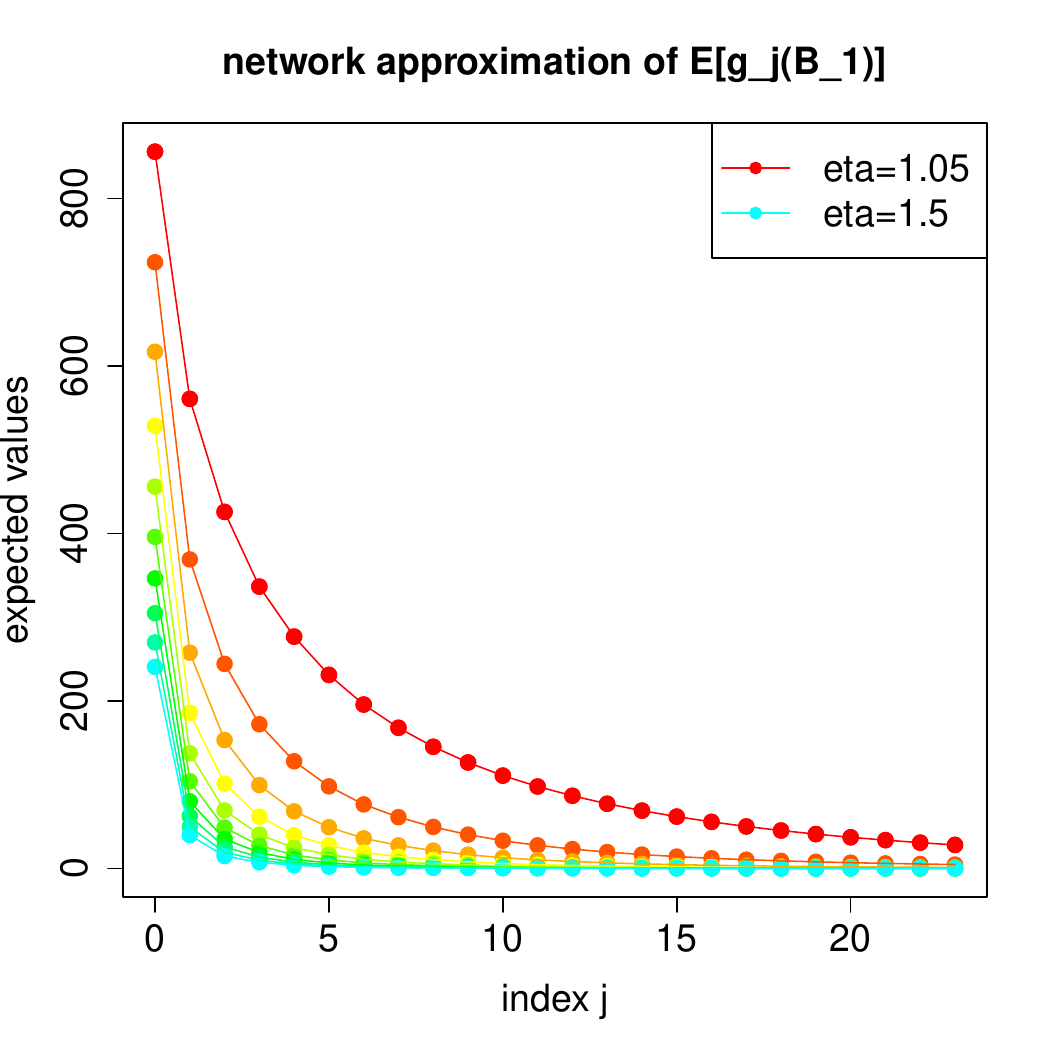}
\end{center}
\end{minipage}
\begin{minipage}[t]{0.4\textwidth}
\begin{center}
\includegraphics[width=\textwidth]{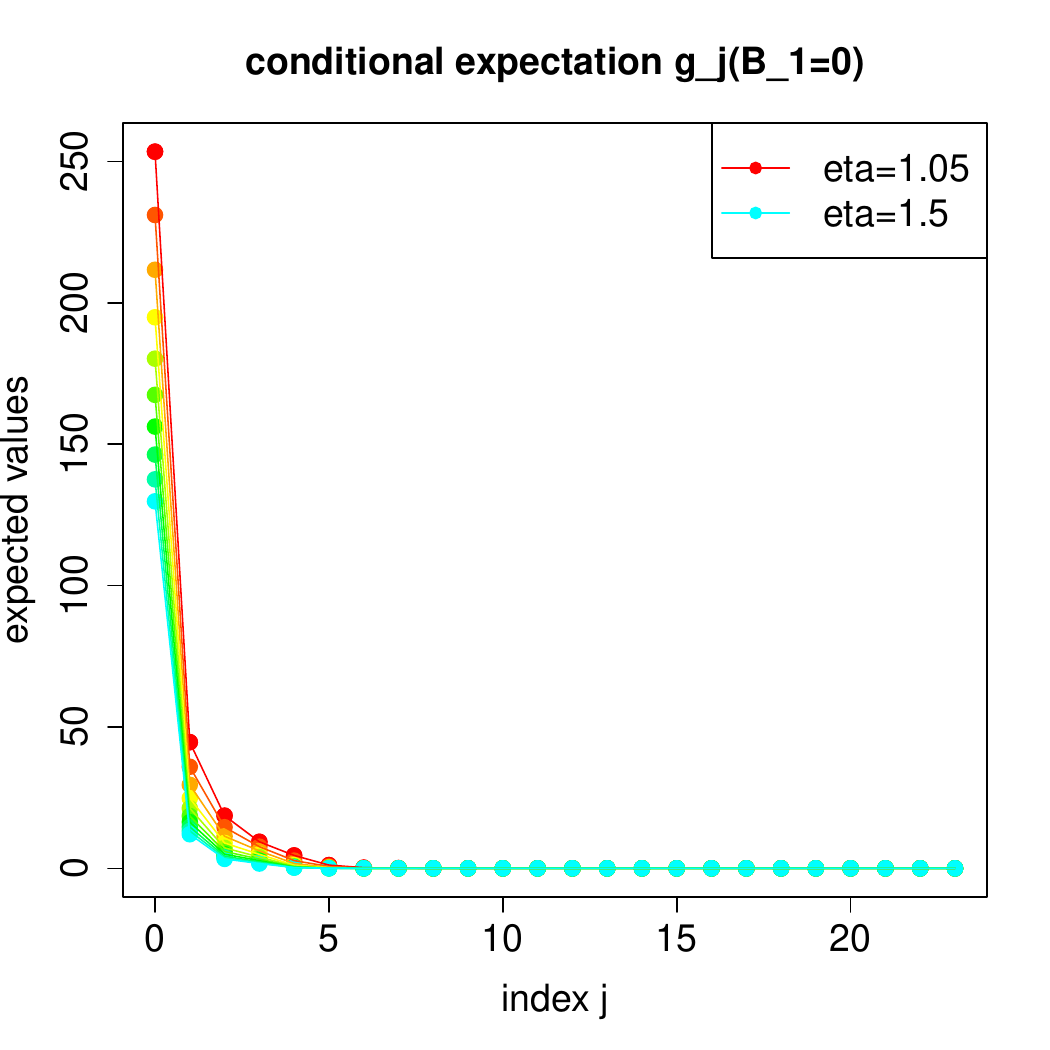}
\end{center}
\end{minipage}
\begin{minipage}[t]{0.4\textwidth}
\begin{center}
\includegraphics[width=\textwidth]{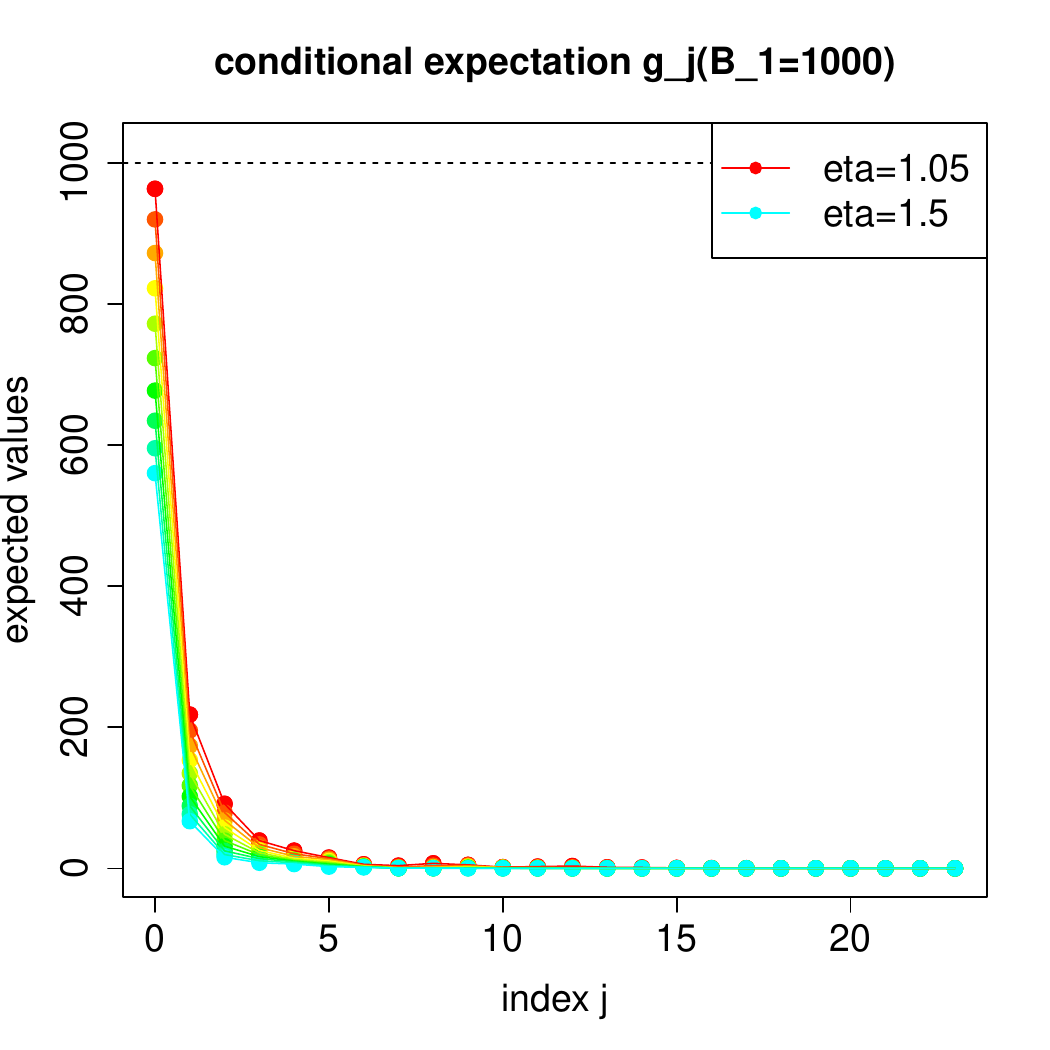}
\end{center}
\end{minipage}
\begin{minipage}[t]{0.4\textwidth}
\begin{center}
\includegraphics[width=\textwidth]{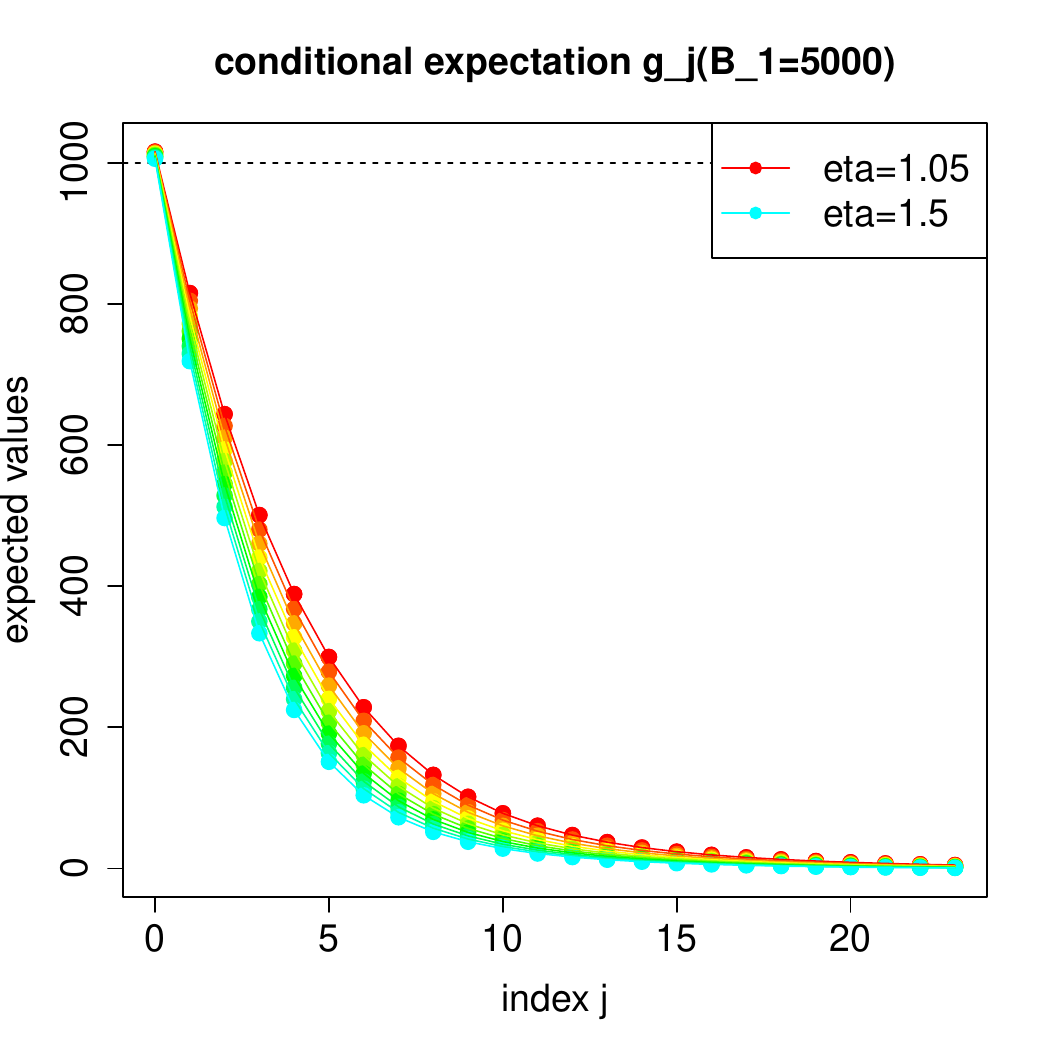}
\end{center}
\end{minipage}
\end{center}
\vspace{-.65cm}
\caption{RNN approximations of the means $\E[g_j(B^{(\eta)}_1;\eta)]$ and conditional means $g_j(B^{(\eta)}_1;\eta)$ for capacity ratios $\eta\in \{1.05, 1.10, \ldots, 1.50\}$: (top-left) unconditional
version averaged over the stationary limit distribution of $B_1^{(\eta)}$; remaining versions are the conditional means for starting backlogs $B^{(\eta)}_1 \in \{0,1000, 5000\}$.}
\label{NB network approximation 2}
\end{figure}

Figure \ref{NB network approximation 2} shows the RNN approximations for different choices of the capacity ratio $\eta \in \{1.05, 1.10, \ldots, 1.50\}$ and for different starting backlogs $B^{(\eta)}_1$. 
All these plots are obtained from the (single) fitted RNN $(b,\eta)\mapsto\bz^{\rm RNN}_{\widehat{\theta}}(b,\eta)$, i.e., we can now simultaneously evaluate $g_j(b,\eta)$ for any initial backlog $b\in \{0,\ldots,  40,000\}$ and any capacity ratio $\eta \in [1.05,1.50]$, this is the input domain on which the network $\bz^{\rm RNN}_{\widehat{\theta}}$ has been trained on.

Figure \ref{NB network approximation 2} (top-left) gives the unconditional versions $\E[\bz^{\rm RNN}_{\widehat{\theta}}(B^{(\eta)}_1;\eta)]$, where we average 
over the stationary limit distribution of the backlogs under the given capacity ratio $\eta$; note that this averaging is again done purely empirically by selecting stationary samples
from the backlogs. 
Figure \ref{NB network approximation 2} (top-right) shows the conditional versions $\bz^{\rm RNN}_{\widehat{\theta}}(0;\eta)$ starting 
in a zero backlog.  
The remaining plots show the conditional versions for starting backlogs 
$B_1^{(\eta)}\in \{1000, 5000\}$; 
these additional plots have all identical $y$-scales and the horizontal dotted lines is at the expected number of reported claims level $\mu=\E[R_1]=1000$.
Recall the random variable
\begin{align*}
F^{(\eta)}_1 = R_1 \bigg(\mathds{1}_{\{B^{(\eta)}_1 > c^{(\eta)}\}}
+ \bigg(1- \frac{c^{(\eta)}-B^{(\eta)}_1}{R_1}\bigg)
\mathds{1}_{\{B^{(\eta)}_1 \le  c^{(\eta)} <B^{(\eta)}_1 +R_1\}}\bigg).
\end{align*}
Clearly, the first indicator is zero for $B^{(\eta)}_1=1000$ and $c^{(\eta)}= \eta \mu= \eta 1000>1000$.
Thus, only the second indicator contributes to $g_0(b,\eta)$ for $b<\eta \mu$. On the other hand, for any starting backlog $B^{(\eta)}_1 \ge \eta\mu$ we have a non-vanishing
first indicator, saying that $g_0(b,\eta) \ge \mu=\E[R_1]$ for $b\ge \eta \mu$.
This is how the initial values $g_0(b,\eta)$ of Figure \ref{NB network approximation 2} 
(bottom) 
are interpreted, and this is highlighted by the horizontal darkgray dotted line. For $j\ge 1$, we can then (simply) verify that it takes more time to run off the initial backlog $B^{(\eta)}_1$ for smaller capacity ratios $\eta>1$.

Fitting the RNN $\bz^{\rm RNN}_{\theta}$ becomes increasingly difficult the closer the capacity ratio $\eta>1$ approaches its limit $1$ of a non-explosive model
which requires $c^{(\eta)}>\mu =\E[R_1]$. 
The unconditional version Figure \ref{NB network approximation 2} (top-left) is more sensitive in $\eta$ than its conditional counterparts $g_0(b,\eta)$, because for the former we average the initial backlog $B^{(\eta)}_1$ over its stationary limit distribution, and the additional variability enters through the increasing volatility of this limiting distribution
for decreasing $\eta$, i.e., we have slower convergence to the stationary limit distribution of the backlog process the closer $\eta>1$ is to the critical value of one. 

\subsection{Approximation of unconditional backlog expectations}
In view of the unconditional cost allocation problem, see Section \ref{sec:costs}, the conditional version $g_j(b;\eta)$, given in \eqref{G expectations 2}, is not fully suitable
because it still needs averaging over the stationary limit distribution of the backlog, which can only be done numerically; this is exactly how Figure \ref{NB network approximation 2} (top-left) has been obtained.  
Here we would rather like to have a function 
\begin{align}\label{G expectations 2B}
\eta ~\mapsto ~g_j(\eta)=
\E \Big[g_j(B^{(\eta)}_1;\eta) \Big]=
\E \bigg[F_1^{(\eta)} \prod_{l=2}^{j+1} G^{(\eta)}_l \bigg]
\quad \text{for } j \ge 0.
\end{align} 
We fit a different second RNN to directly approximate $\eta \mapsto (g_j(\eta))_{j=0}^{T-1}$ rather than $(b,\eta) \mapsto (g_j(b,\eta))_{j=0}^{T-1}$. 
For receiving a suitable RNN approximation we need to ensure that the initial backlog $B^{(\eta)}_1$ is sampled from the stationary limit distribution. 
For this we start a process that has some burn-in until it generates stationary samples. We first study empirically the convergence rate to the stationary phase. 
Figure \ref{NB burn in} shows the results for two different capacity ratios $\eta=1.05, 1.10$ starting from a zero backlog. From these plots we conclude that we need to simulate roughly 1200 iteration steps to arrive at an empirical approximation to the stationary limit distribution. All following results are obtained by using this burn-in of 1200 iterations, and the subsequent samples are then taken as an empirical approximation to the stationary limit distribution. 

\begin{figure}[htb!]
\begin{center}
\begin{minipage}[t]{0.48\textwidth}
\begin{center}
\includegraphics[width=\textwidth]{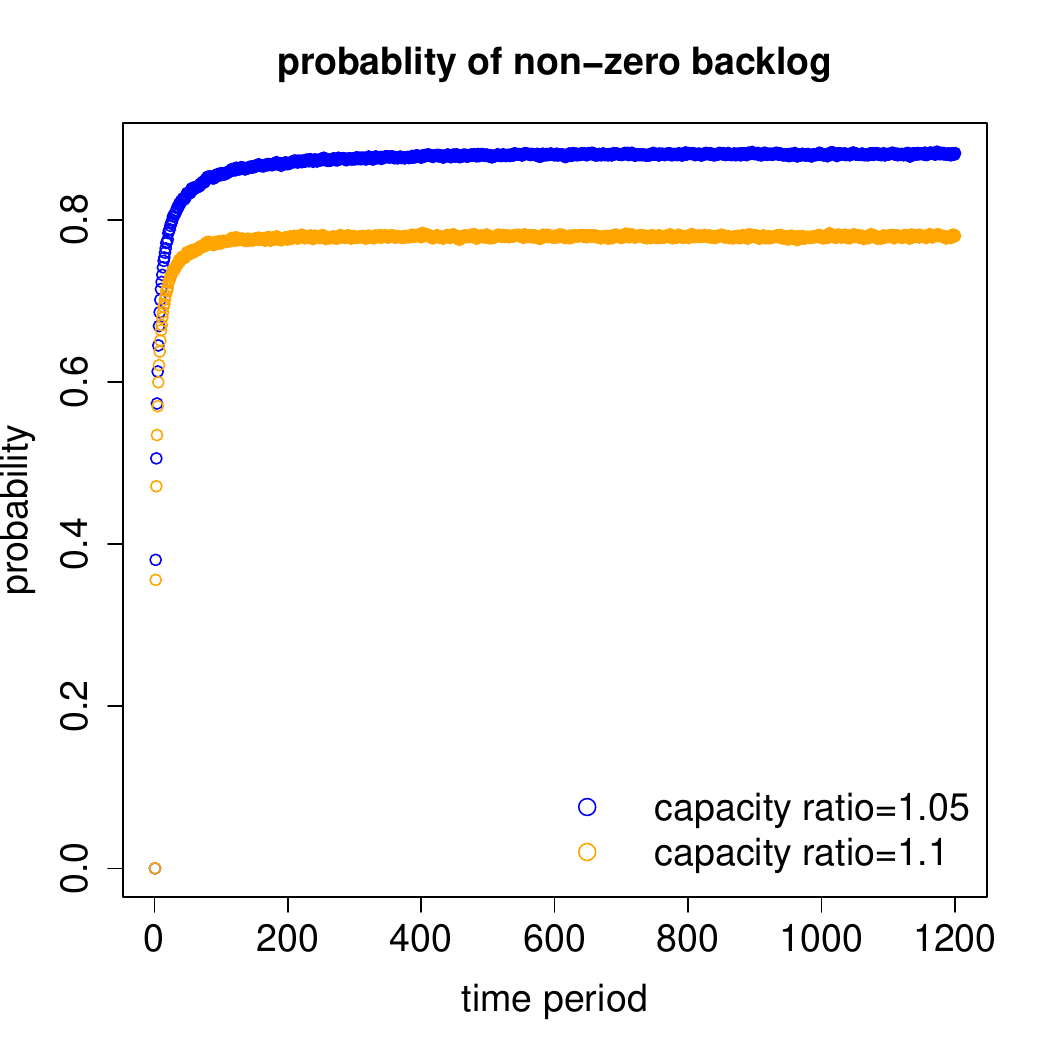}
\end{center}
\end{minipage}
\begin{minipage}[t]{0.48\textwidth}
\begin{center}
\includegraphics[width=\textwidth]{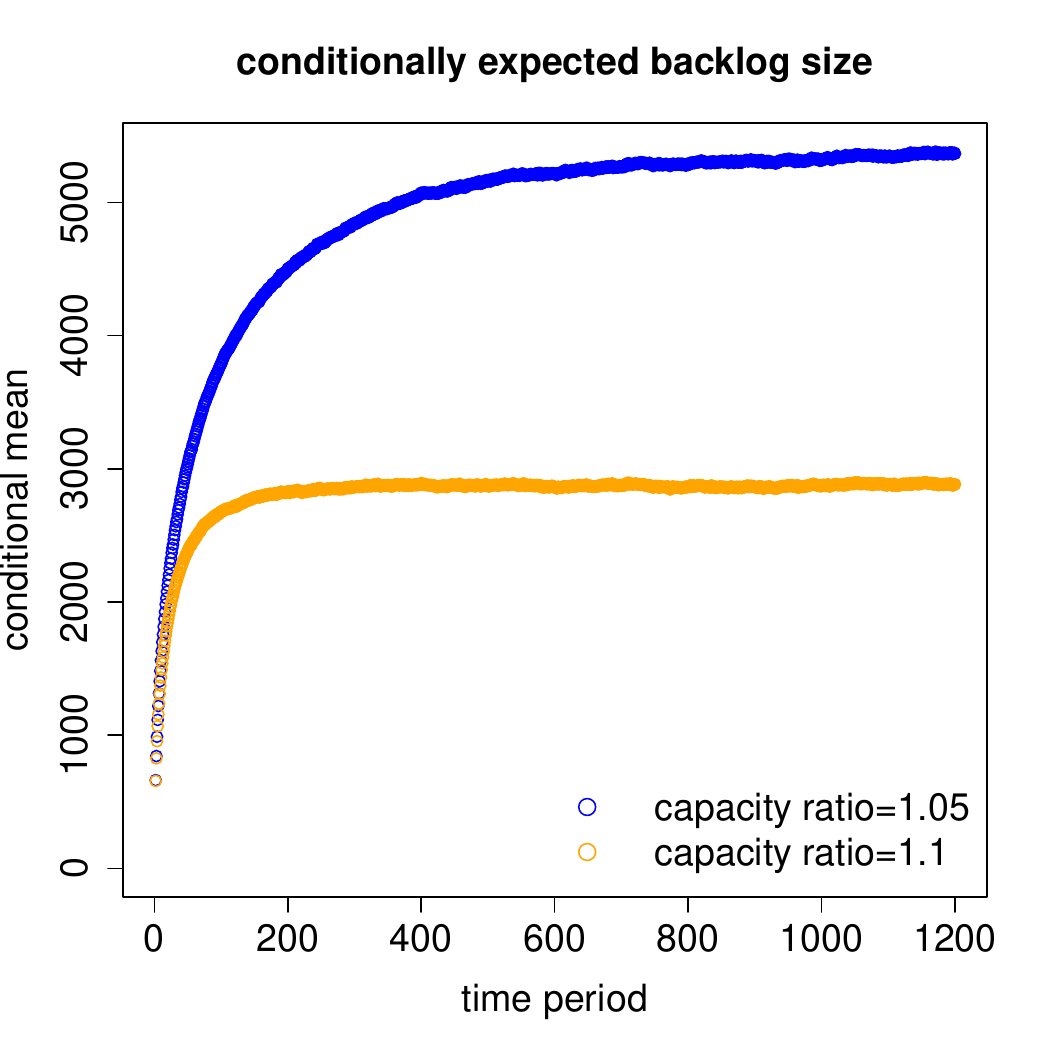}
\end{center}
\end{minipage}
\end{center}
\vspace{-.65cm}
\caption{Negative binomial model with capacity ratios $\eta=1.05, 1.10$:
analysis of burn-in to reach the stationary limit distribution
for an initial backlog of zero.}
\label{NB burn in}
\end{figure}

\begin{figure}[htb!]
\begin{center}
\begin{minipage}[t]{0.44\textwidth}
\begin{center}
\includegraphics[width=\textwidth]{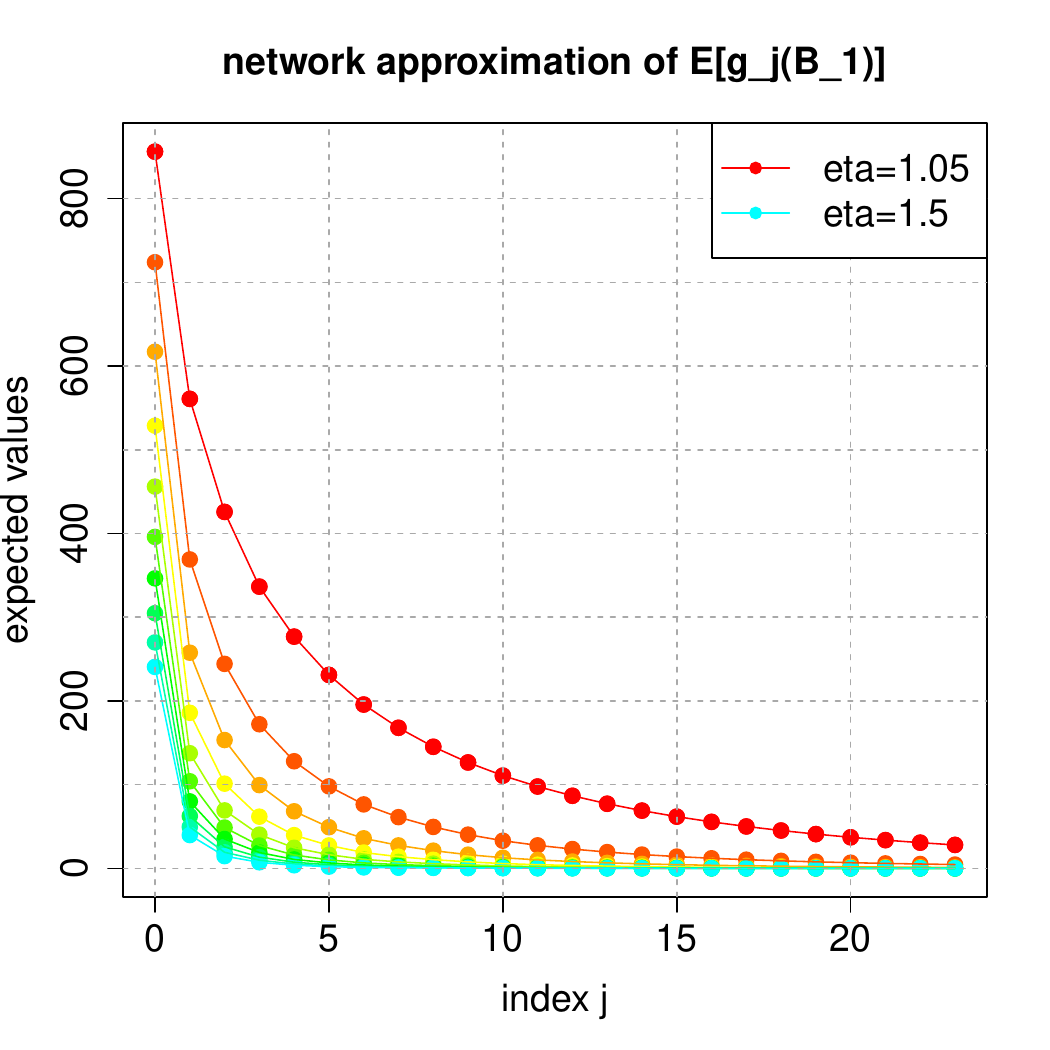}
\end{center}
\end{minipage}
\begin{minipage}[t]{0.44\textwidth}
\begin{center}
\includegraphics[width=\textwidth]{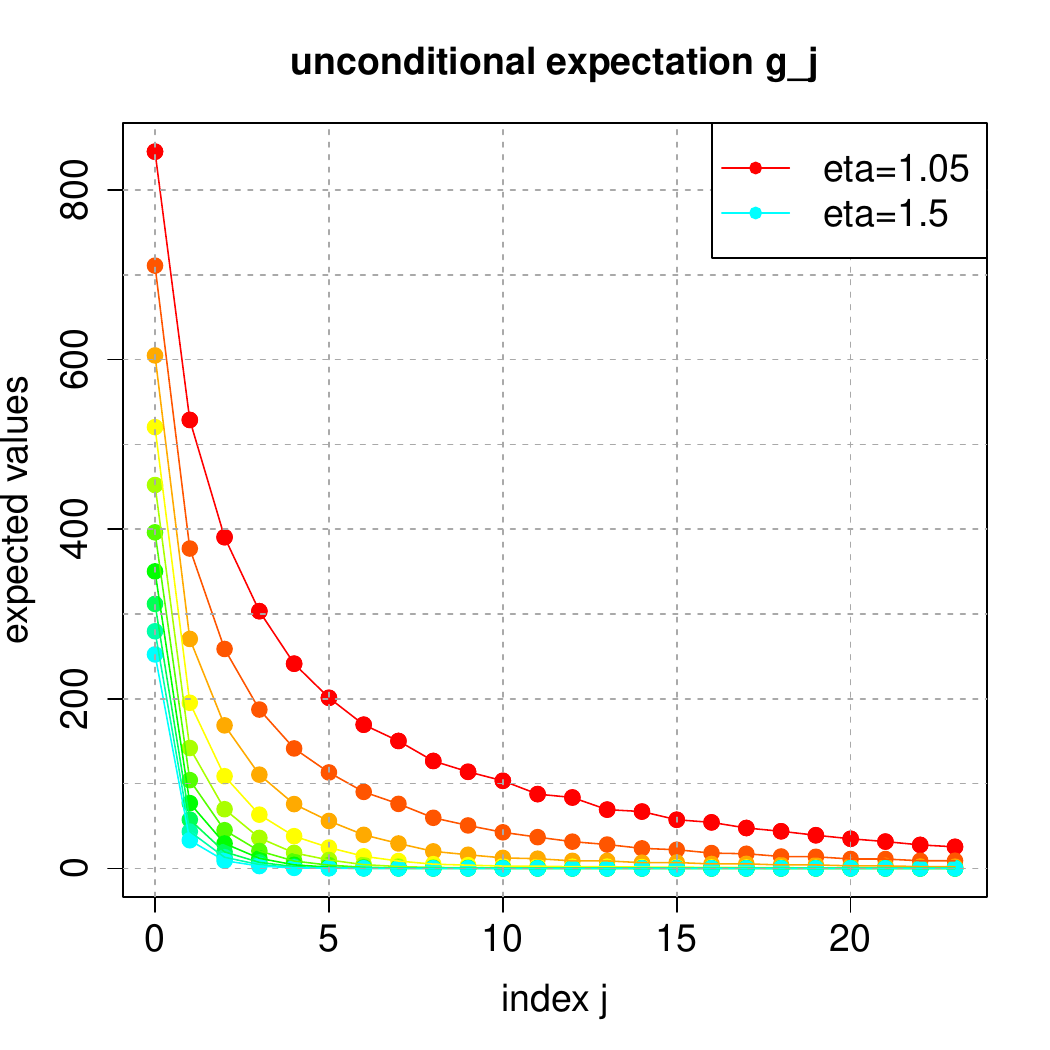}
\end{center}
\end{minipage}
\end{center}
\vspace{-.65cm}
\caption{RNN approximations of the means $(g_j(\eta))_{j \ge 0}$ for capacity ratios $\eta\in \{1.05, 1.10, \ldots, 1.50\}$: (lhs) this is identical to Figure \ref{NB network approximation 2} (top-left) received by empirically averaging over the conditional networks, and (rhs) direct fit of the unconditional mean using a single input network.}
\label{NB network approximation 2B}
\end{figure}

Based on this sampling and fitting strategy we fit a single input RRN to the function in \eqref{G expectations 2B}.
Figure \ref{NB network approximation 2B} (rhs) shows the results of this direct fitting of a RNN to the unconditional means $g_j(\eta)$ using stationary backlog time series for different
capacity ratios $\eta \in [1.05, 1.50]$. This is compared to the empirical average over the conditional means $g_j(B_1^{(\eta)};\eta)$ taken from Figure \ref{NB network approximation 2} (top-left). We see an excellent alignment of the results, telling that both networks have learned the same structure. The only (smaller) differences are visible for the smallest capacity ratio $\eta=1.05$. The issue here clearly is slow convergence to the stationary limiting distribution, which is also verified from Figure \ref{NB burn in}. For the unconditional cost optimization problem we use the unconditional network approximation to $(g_j(\eta))_{j \ge 0}$. 

\subsection{Approximation of conditional backlog expectations}

In order to numerically approximate the conditional backlog expectations in Corollary \ref{cor:cond_backlog_exp} we observe that we need to have time-delayed versions of the quantities $g_j(b;\eta)$ introduced in \eqref{G expectations 2}. 
We need to compute $h_j(b, m;\eta)$, where, for $j\geq 0$,  
\begin{align}\label{G expectations 3B}
h_j(b, 0;\eta)&:=\E \bigg[\prod_{l=1}^{j} G^{(\eta)}_{l}\,\bigg|\, B^{(\eta)}_1=b \bigg],
\\\label{G expectations 3A}
h_j(b, m;\eta)&:=\E \bigg[F_{m}^{(\eta)} \prod_{l=m+1}^{m+j} G^{(\eta)}_{l}\,\bigg|\, B^{(\eta)}_1=b \bigg], \quad m\geq 1. 
\end{align}
The functions $h_j(b, 0;\eta)$ are time-advanced versions of $g_j(b;\eta)$ by dropping $F_1^{(\eta)}$, the functions $h_j(b, m;\eta)$, $m \geq 1$, are time-delayed versions of $g_j(b;\eta)$. 
Using these we rewrite the expression for the conditional expectation in Corollary \ref{cor:cond_backlog_exp} as follows. 
For $k\geq 0$ such that $\tau-i+k\geq 0$, 
\begin{align}
\E\big[B^{(\eta)}_{i,\tau-i+k+1}\mid \calG^{(\eta)}_{\tau}\big]
&=\sum_{m=1}^{k}\mathds{1}_{\{i\leq \tau+m\leq i+ J\}}\,\frac{\mu_{\tau-i+m}}{\mu}\, h_{k-m}\Big(B^{(\eta)}_{\tau+1}, m;\eta\Big) \label{eq:cond_exp_mu_term} \\
&\quad+\mathds{1}_{\{i\leq \tau\leq i+ J\}}\,\frac{R_{i,\tau-i}}{R_{\tau}}\,F^{(\eta)}_{\tau}\,h_k\Big(B^{(\eta)}_{\tau+1}, 0;\eta\Big) \label{eq:cond_exp_F_term} \\
&\quad+\mathds{1}_{\{i\leq \tau\}}\,B^{(\eta)}_{i,\tau-i}\,G^{(\eta)}_{\tau}\,h_k\Big(B^{(\eta)}_{\tau+1}, 0;\eta\Big). \label{eq:cond_exp_G_term} 
\end{align} 
Since $h_j(b, m;\eta)$ lives on different scales
for $m=0$, formula \eqref{G expectations 3B}, and $m\ge 1$, formula \eqref{G expectations 3A},
we fit two different networks to these two cases. This is easier in training.
The case $m=0$ is then completely analogous to 
$g_j(b;\eta)$, see \eqref{gradient descent 1st}, and
the fitting results are given in Figure \ref{figure conditional 0} in the appendix.

Fitting $h_j(b, m;\eta)$, $m \ge 1$, is slightly more tricky because
the first term $F_1^{(\eta)}$ lives on a different scale compared
to $G_l^{(\eta)}$, $l\ge m+1$. This is similar to $g_j(b;\eta)$.
Note that we can rewrite \eqref{G expectations 3A} as follows 
\begin{align*}
h_j(b, m;\eta)&=\E \bigg[F_{m}^{(\eta)} \prod_{l=m+1}^{m+j} G^{(\eta)}_{l}\,\bigg|\, B^{(\eta)}_1=b \bigg]\\
&=\E \bigg[\E\bigg[F_{m}^{(\eta)} \prod_{l=m+1}^{m+j} G^{(\eta)}_{l}\,\bigg|\, B^{(\eta)}_m, B^{(\eta)}_1=b \bigg] \,\bigg|\, B^{(\eta)}_1=b \bigg]\\
&=\E \bigg[g_j\left(B^{(\eta)}_m;\eta\right)\,\bigg|\, B^{(\eta)}_1=b \bigg]. 
\end{align*}
This shows that we can employ a similar approximation and fitting strategy as in \eqref{gradient descent 1st} but we need to time-delay by $m-1$ periods.
This requires an additional input $m$ to the RNN $\bz^{\rm RNN}_{\theta}$, and then we can fit
\begin{align}\label{gradient descent 2nd}
\widehat{\theta} ~ \in ~
\underset{\theta}{\arg\min}\,
\frac{1}{n} \sum_{k=1}^n
\sum_{j=0}^{T-1}
\bigg(F_1^{k, m} \prod_{l=2}^{j+1} G_l^{k, m} - \bz^{\rm RNN}_{\theta}(B_1^k,\eta^k,m)_j \bigg)^2,
\end{align}
where the observations $(F_1^{k, m}, G^{k, m}_2, \ldots, G_{T-1}^{k,m })$, $1\le k \le n$, are received by first simulating the $m-1$ periods delayed starting point  $B_m^k$ from $B_1^k$, set new starting value $B_1^{k,m}:=B_m^k$, and then proceed as in \eqref{gradient descent 1st}.
This provides us with a fitted neural network $(b,\eta,m)\mapsto\bz^{\rm RNN}_{\widehat{\theta}}(b,\eta,m)$ that approximates
$h_j(b, m;\eta)$, $m \ge 1$, $j\ge 0$, $b\ge 0$ and $\eta \in (1.05, 1.50)$.
The results are shown in Figures \ref{figure conditional 1}-\ref{figure conditional 3} in the appendix for the three different starting values $b\in \{0, 1000, 5000\}$.

\section{Cost optimization}\label{sec:costs_negbin}
We now have prepared all the necessary numerical tools to study the optimal
capacity ratio $\eta>1$ for obtaining minimal costs. We distinguish
the unconditional and the conditional cases. The latter considers a
situation where we want to optimally plan the capacity under a given
starting backlog $B_1^{(\eta)}$, and the former unconditional
case is a global consideration of long-term optimal planing to receive minimal
costs. In the long run, the conditional version will converge to the unconditional one, regardless of the specific initial backlog $B_1^{(\eta)}$.

\subsection{Unconditional cost optimization}\label{sec:cost_optim_negbin}

Recall the expressions \eqref{eq:linear_cost_case} and \eqref{eq:inflating_cost_case} for expected combined delay-adjusted claim costs and capacity costs. Here we study these expected costs as functions of the capacity ratio $\eta$. We write 
\begin{align}\label{linear case}
\mu_i^{(\ell)}(\eta) := \kappa_g \,\mu + \kappa_b \, \E[B^{(\eta)}] + \kappa_c \,(c^{(\eta)}-\mu)
\end{align}
for the linear cost model, and 
\begin{align}\label{inflating cost case}
\mu_i^{(\iota)}(\eta): = \kappa_g \sum_{j\ge 0} \lambda_b^{j}\,
\E\big[B^{(\eta)}_{i,j}+R_{i,j}-B^{(\eta)}_{i,j+1}\big] + \kappa_c \,\big(c^{(\eta)}-\mu\big)
\end{align}
for the model with non-linear delay-inflated costs. 
We study these total costs as a function of the capacity ratio $\eta>1$ providing $c^{(\eta)}=\eta \mu$ and impacting the backlog via recursion \eqref{eq:backlog_recursion}. Using the network approximation to $(g_j(\eta))_{j\ge 0}$ we can explicitly compute these costs using relation (see Theorem \ref{thm:uncond_backlog_exp}) 
\begin{align}\label{expected backlogs per accident year}
\E\big[B_{i,j}^{(\eta)}\big]=\sum_{k=1}^{j \wedge (J+1)}\frac{\mu_{k-1}}{\mu}\,\E\bigg[F^{(\eta)}_{i+k-1}\prod_{l=k}^{j-1}G^{(\eta)}_{i+l}\bigg]
=\sum_{k=1}^{j \wedge (J+1)}\frac{\mu_{k-1}}{\mu}\, g_{j-k}(\eta),
\end{align}
where an empty sum is equal to $0$ and an empty product is equal to $1$. 
 
\begin{figure}[htb!]
\begin{center}
\begin{minipage}[t]{0.44\textwidth}
\begin{center}
\includegraphics[width=\textwidth]{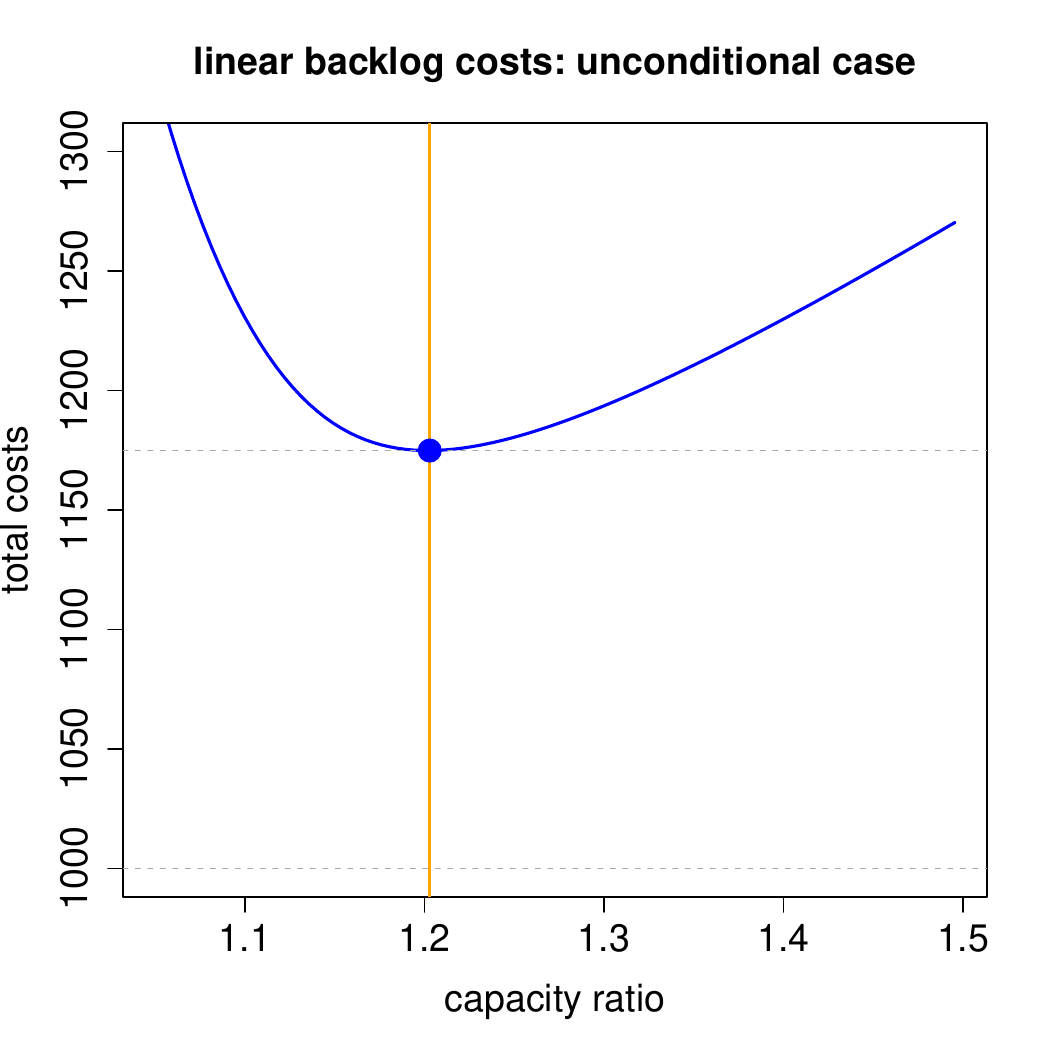}
\end{center}
\end{minipage}
\begin{minipage}[t]{0.44\textwidth}
\begin{center}
\includegraphics[width=\textwidth]{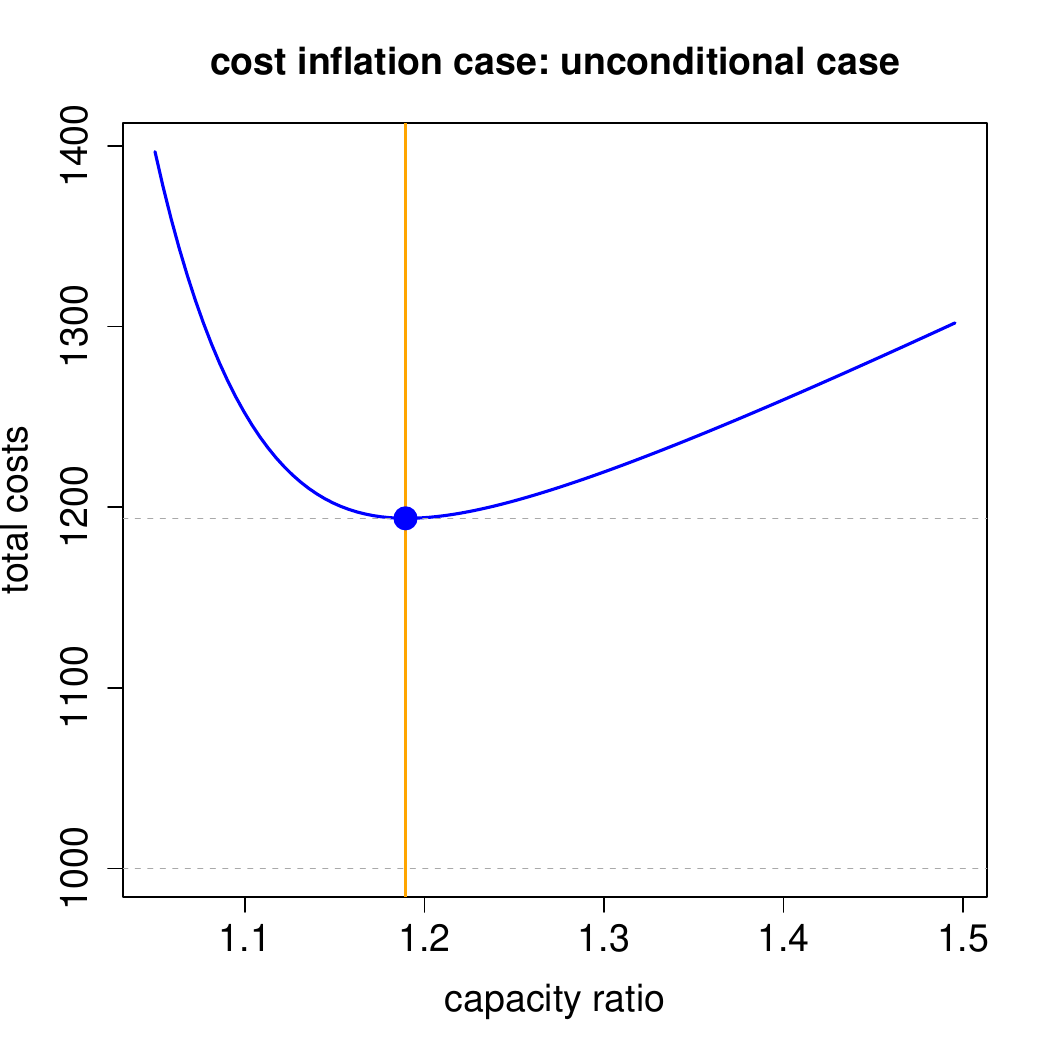}
\end{center}
\end{minipage}
\end{center}
\vspace{-.65cm}
\caption{Optimal capacity ratio: (lhs) linear backlog costs case \eqref{linear case}
and (rhs) inflating backlog costs case \eqref{inflating cost case}; the vertical
orange line shows the optimal capacity ratios $\eta^\ast$ and
the lower horizontal dotted darkgray line the ground-up costs $\kappa_g \mu$.}
\label{linear cost example}
\end{figure}

Figure \ref{linear cost example} (lhs) shows the optimal capacity ratio in the linear capacity cost case \eqref{linear case} with $\kappa_g=1$, which gives ground-up costs
of $\kappa_g \mu=1000$. To this we add backlog costs with $\kappa_b=0.075$ and excess capacity costs with $\kappa_c=0.5$. The plot shows the linear cost
case $\eta \mapsto \mu_i^{(\ell)}(\eta)$ as a function of the capacity ratio $\eta>1$. In this model (and parametrization) the optimal (total cost minimizing) ratio is $\eta^\ast=1.203$. Thus, in this case we have ground-up costs of $1000$, and the optimal capacity ratio $\eta^\ast=1.203$ adds another $175$ which is related to backlog costs. This results in the upper horizontal darkgray dotted line at $1175$ describing the total expected costs $\mu_i^{(\ell)}(\eta^\ast)$ allocated to occurrence period $i$.

Figure \ref{linear cost example} (rhs) shows the inflating backlog costs case \eqref{inflating cost case}, $\eta \mapsto \mu_i^{(\iota)}(\eta)$, where we inflate claims costs by an inflation rate of 5\% resulting in $\lambda_b=1.05$, and the remaining parameters are selected as above. In that case the optimal capacity ratio under our parametrization is $\eta^\ast=1.190$.

\begin{figure}[htb!]
\begin{center}
\begin{minipage}[t]{0.44\textwidth}
\begin{center}
\includegraphics[width=\textwidth]{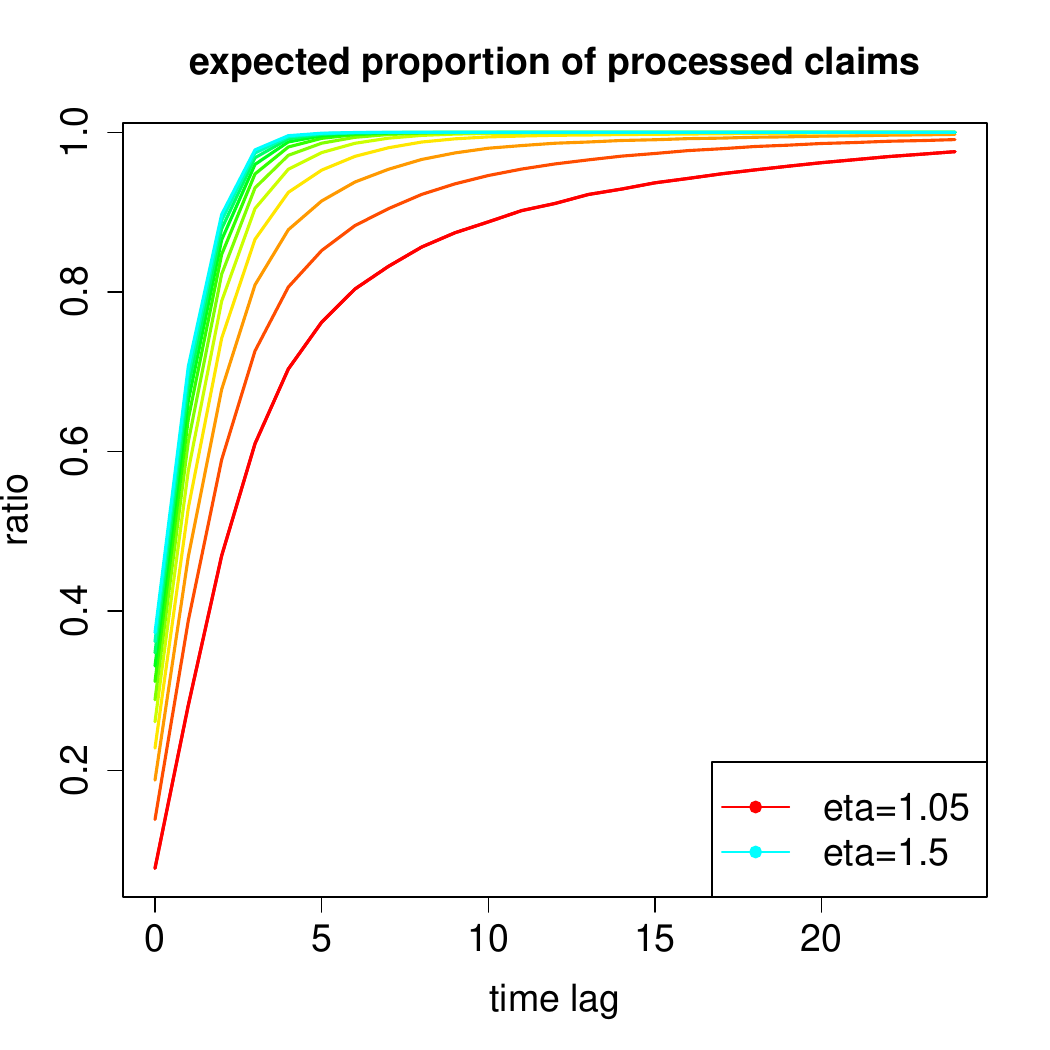}
\end{center}
\end{minipage}
\end{center}
\vspace{-.65cm}
\caption{Expected proportion of processed claims for a fixed
occurrence period for different capacity ratios $\eta \in \{1.05, \ldots, 1.50\}$.}
\label{process ratio}
\end{figure}

The explicit formula \eqref{expected backlogs per accident year}
for the expected backlogs $\E[B_{i,j}^{(\eta)}]$ for a fixed
occurrence period $i$ for different delays $j$ allows us
to compute the expected processing pattern 
$j\mapsto \E[P^{(\eta)}_{i,j}]$ of a fixed occurrence period, see also 
Corollary \ref{expected processing corollary}. Figure
\ref{process ratio} shows the cumulative proportion of expected
processed claims for different capacity ratios $\eta \in \{1.05, \ldots, 1.50\}$.
For the higher capacity ratios all claims are likely processed within 5 periods,
whereas for lower capacity ratios it may take up to 24 periods on average. This
precisely explains the cost differences implied by the
inflation part in the delay-inflated cost case \eqref{inflating cost case}.

\subsection{Conditional cost optimization}

Cost optimization in the conditional case is more involved because we cannot use convenient consequences of stationarity. 
Here we assume that customers pay a fixed premium and the cost optimization aims to maximize the profit on the business. 
We consider a going-concern view, assuming that the business continues as planned during a given planning horizon with a constant capacity that we aim to select optimally. 
We consider all costs in the time window $(\tau, \tau + T]$ without running off the open claims at the end of this time window. For $T\to \infty$, the results converge to the unconditional case because the impact of the starting configuration at time $\tau$ vanishes asymptotically as $T\to \infty$.

The optimization problem requires us to study $\E[B^{(\eta)}_{i,\tau-i+k+1}\mid \calG^{(\eta)}_{\tau}]$, where 
\begin{align*}
\calG^{(\eta)}_{\tau}=\sigma\big(R_{i,j},B^{(\eta)}_{i,j}:i+j\leq \tau, j\geq 0\big).
\end{align*}
Hence, we know all ingoing backlogs $(B^{(\eta)}_{i,j})_{i+j\le \tau}$ into 
period $\tau$ and all reported claims $(R_{i,j})_{i+j\le \tau}$ 
in period $\tau$. Having a constant capacity $C_t=c^{(\eta)}$, we 
therefore also know the aggregated ingoing backlog $B^{(\eta)}_{\tau +1}$
one period later. 
However, if $B^{(\eta)}_{\tau +1}>0$, then information $\calG^{(\eta)}_{\tau}$
does not provide its partition to $B^{(\eta)}_{i,\tau-i+1}$, i.e., the
individual origin periods. However, we will see
that if we aggregate over the occurrence periods $i$, this split is not necessary.

We consider the cases of linear backlog costs \eqref{linear case}. For the moment we drop the costs for the excess capacity $\kappa_c (c^{(\eta)}-\mu)$, since we do not allocate these costs to individual occurrence periods in the going-concern view, we can just add these costs at the end for every period up to the planning horizon $T$. In the linear cost case we need to study 
\begin{align}\label{linear case 2}
\kappa_g \,
\sum_{k }\E\big[R_{i,\tau-i+k+1} \mid \calG^{(\eta)}_{\tau}\big] 
+  \kappa_b \,
\sum_{k}\E\big[B^{(\eta)}_{i,\tau-i+k+1} \mid \calG^{(\eta)}_{\tau}\big],
\end{align}
the summation index $k$ is going to be discussed below.

For the subsequent considerations one should have the following matrix in mind. The rows in this matrix reflect the different occurrence periods from $\tau-J$ (first row) to $\tau + T$ in (last row). These are all occurrence periods that contribute to payments in the time window $(\tau, \tau + T]$. Thus, we have $T+J+1$ occurrence periods in this matrix. The columns in this matrix correspond to calendar periods. The first column reflects the reported claims in calendar period $\tau$, the second column the expected number of reported claims in calendar period $\tau+1$, and the last column the expected number of reported claims in calendar period $\tau +T$.
Thus, this matrix has $T+1$ columns. 
\begin{align}\label{runoff matrix}
M:=
\begin{pmatrix}
R_{\tau-J,J} & 0 & 0&  0&\cdots &  0&0 \\
R_{\tau-J+1,J-1} & \mu_J & 0&0& \cdots &  0&0 \\
R_{\tau-J+2,J-2} & \mu_{J-1}&\mu_J& 0& \cdots &  0&0 \\
\vdots & \vdots & \vdots &\vdots& \ddots & \vdots &\vdots \\
R_{\tau,0} & \mu_{1}&\mu_2&  \mu_3&\cdots  & 0& 0 \\
0 & \mu_{0}&\mu_1&  \mu_2&\cdots & 0& 0 \\
0 & 0&\mu_0& \mu_1& \cdots &  0&0 \\
0 & 0&0& \mu_0& \cdots &  0&0 \\
\vdots & \vdots & \vdots &\vdots&\ddots& \vdots& \vdots \\
0 & 0&0&  0&\cdots & \mu_0& \mu_1 \\
0 & 0&0&  0&\cdots & 0& \mu_0 \\
\end{pmatrix} 
\end{align}
One should notice that calendar period $\tau$ corresponding to the first column aggregates to the total number of reported claims $R_\tau$ in calendar period $\tau$, all other columns aggregate to $\mu$ being the expected number of claims of futures periods in this stationary model.

The reported claims terms in \eqref{linear case 2} are comparably simple because future reportings are independent of $\calG^{(\eta)}_{\tau}$. With finite planning horizon we have a total number of expected reported claims, this is the sum over columns 2 to $T+1$ in matrix \eqref{runoff matrix}, 
\begin{align*}
\sum_{k=0}^{T-1}\sum_{i=\tau+k+1-J}^{\tau+k+1}\E\big[ R_{i,\tau-i+k+1}\mid \calG^{(\eta)}_{\tau}\big]
=\sum_{c=2}^{T+1}\sum_{r=c}^{c+J}M_{r,c}=T\mu. 
\end{align*}
Note that occurrence periods $i\le \tau-J$ are fully reported at time $\tau$, therefore, they do not contribute to the above sum. We add the planning horizon $T$ and we discard all costs after this time window. 

Next we focus on the backlogs in \eqref{linear case 2}. For this we come back to the three terms \eqref{eq:cond_exp_mu_term}-\eqref{eq:cond_exp_G_term} which need to be summed over $k$ for a fixed occurrence period $i$.
We start with the term \eqref{eq:cond_exp_G_term}. This requires $i\le \tau$, otherwise this occurrence period cannot have any ingoing backlog. Their total contribution across all occurrence periods $i$ is then
\begin{align*}
\sum_{i \le \tau}
B^{(\eta)}_{i,\tau-i}\, \sum_{k=0}^{T-1}
G^{(\eta)}_{\tau}\,
h_k\Big(B^{(\eta)}_{\tau+1}, 0;\eta\Big)
=B^{(\eta)}_{\tau}\, \sum_{k=0}^{T-1}
G^{(\eta)}_{\tau}\,
h_k\Big(B^{(\eta)}_{\tau+1}, 0;\eta\Big),
\end{align*}
we add the planning horizon $T$ and we discard all costs after this time window. 

Next, we focus on the term \eqref{eq:cond_exp_F_term} newly reported claims in period $\tau$. This requires occurrence periods $i\leq \tau\leq i+ J$, thus, we have total expected backlogs from newly reported claims, this is the sum over the first column in matrix \eqref{runoff matrix},
\begin{align*}
\sum_{i =\tau-J}^\tau\frac{R_{i,\tau-i}}{R_{\tau}}\sum_{k=0}^{T-1}
F^{(\eta)}_{\tau}\,h_k\left(B^{(\eta)}_{\tau+1}, 0;\eta\right)
=\sum_{k=0}^{T-1}
F^{(\eta)}_{\tau}\,h_k\left(B^{(\eta)}_{\tau+1}, 0;\eta\right),
\end{align*}
this is again truncated at the planning horizon $T$.

Finally, we focus on the term \eqref{eq:cond_exp_mu_term} for future reportings for occurrence periods $ i>\tau-J$. To verify the following computation one should again note that we  aggregate over the columns 2 to $T+1$ in matrix \eqref{runoff matrix}, and each of these columns belongs to a different delay $m \in \{1,\ldots, T \}$,
\begin{align*}
&\sum_{i = \tau-J +1}^{\tau+T}
 \sum_{k=1}^{T}
\sum_{m=1}^{k}\mathds{1}_{\{i\leq \tau+m\leq i+ J\}}\,\frac{\mu_{\tau-i+m}}{\mu}\, h_{k-m}\left(B^{(\eta)}_{\tau+1}, m;\eta\right)
\\
&\quad=\sum_{m=1}^{T}
\sum_{i = \tau-J +1}^{\tau+T}
\mathds{1}_{\{\tau-J+m\leq i\leq \tau+m\}}\,\frac{\mu_{\tau+m-i}}{\mu}
 \sum_{k=m}^{T} h_{k-m}\left(B^{(\eta)}_{\tau+1}, m;\eta\right)
\\
&\quad=\sum_{m=1}^{T}
 \sum_{k=m}^{T} h_{k-m}\left(B^{(\eta)}_{\tau+1}, m;\eta\right)
=\sum_{m=1}^{T}
 \sum_{k=0}^{T-m} h_{k}\left(B^{(\eta)}_{\tau+1}, m;\eta\right).
\end{align*}
Collecting all terms and aggregating over all occurrence
periods $i \in \{\tau-J, \ldots, T\}$ that contribute to the costs in the time window
$(\tau, \tau+T]$, see \eqref{runoff matrix},  the linear costs case 
\eqref{linear case 2} requires
studying the aggregate costs
\begin{align}
T\mu^{(\ell)}(\eta, T, \calG^{(\eta)}_{\tau})
&:= \kappa_g \,T\,\mu
+  \kappa_b \,
B^{(\eta)}_{\tau} \sum_{k=0}^{T-1}
G^{(\eta)}_{\tau}\,
h_k\left(B^{(\eta)}_{\tau+1}, 0;\eta\right) \label{general linear cost formula} \\
&\quad+\,  \kappa_b \, \sum_{k=0}^{T-1}
F^{(\eta)}_{\tau}\,h_k\left(B^{(\eta)}_{\tau+1}, 0;\eta\right) \nonumber  \\
&\quad+\,  \kappa_b \,\sum_{m=1}^{T}
 \sum_{k=0}^{T-m} h_{k}\left(B^{(\eta)}_{\tau+1}, m;\eta\right)
+\kappa_c \, T\left(c^{(\eta)}-\mu\right). \nonumber
\end{align}
Thus, we have a fixed past history $\calG^{(\eta)}_{\tau}$, we have a fixed planing horizon $T\ge 1$, and we try to minimize to costs in time window $(\tau, \tau+T]$ in the cost capacity ratio $\eta$, for given cost parameters $\kappa_g, \kappa_b, \kappa_c>0$. To keep things simple, we assume that we start from a zero backlog $B^{(\eta)}_{\tau }=0$ at time $\tau$. This allows us to drop the backlog term in \eqref{general linear cost formula}, and we receive a next backlog at time $\tau+1$
\begin{align}\label{the next backlog}
B^{(\eta)}_{\tau+1} = \max(R_\tau-C^{(\eta)},0)
= \max\bigg(\sum_{i \le \tau}R_{i,\tau-i}-C^{(\eta)},0\bigg).
\end{align}
Thus, the zero initial backlog case $B^{(\eta)}_{\tau }=0$ results in studying 
\begin{align*}
T\mu^{(\ell)}(\eta, T, \calG^{(\eta)}_{\tau})
&= T\Big(\kappa_c\,  (c^{(\eta)}-\mu) + \kappa_g \,\mu\Big) 
+\, \kappa_b \, \sum_{k=0}^{T-1}
F^{(\eta)}_{\tau}\,h_k\Big(B^{(\eta)}_{\tau+1}, 0;\eta\Big)
\\
&\quad+\,  \kappa_b \,\sum_{m=1}^{T}
 \sum_{k=0}^{T-m} h_{k}\Big(B^{(\eta)}_{\tau+1}, m;\eta\Big)
.
\end{align*}
This is fairly simple now. It requires that starting from a zero backlog $B^{(\eta)}_{\tau }=0$ at time $\tau$, we need to simulate the number of reported claims $R_\tau$ in period $\tau$, from which we can compute the new backlog $B^{(\eta)}_{\tau+1}$ at time $\tau+1$ as well as $F^{(\eta)}_{\tau}$. In the following analysis we have $R_\tau=1310$ from which we can compute the next backlog \eqref{the next backlog} for the different capacity ratios $\eta >1$.

\begin{figure}[htb!]
\begin{center}
\begin{minipage}[t]{0.44\textwidth}
\begin{center}
\includegraphics[width=\textwidth]{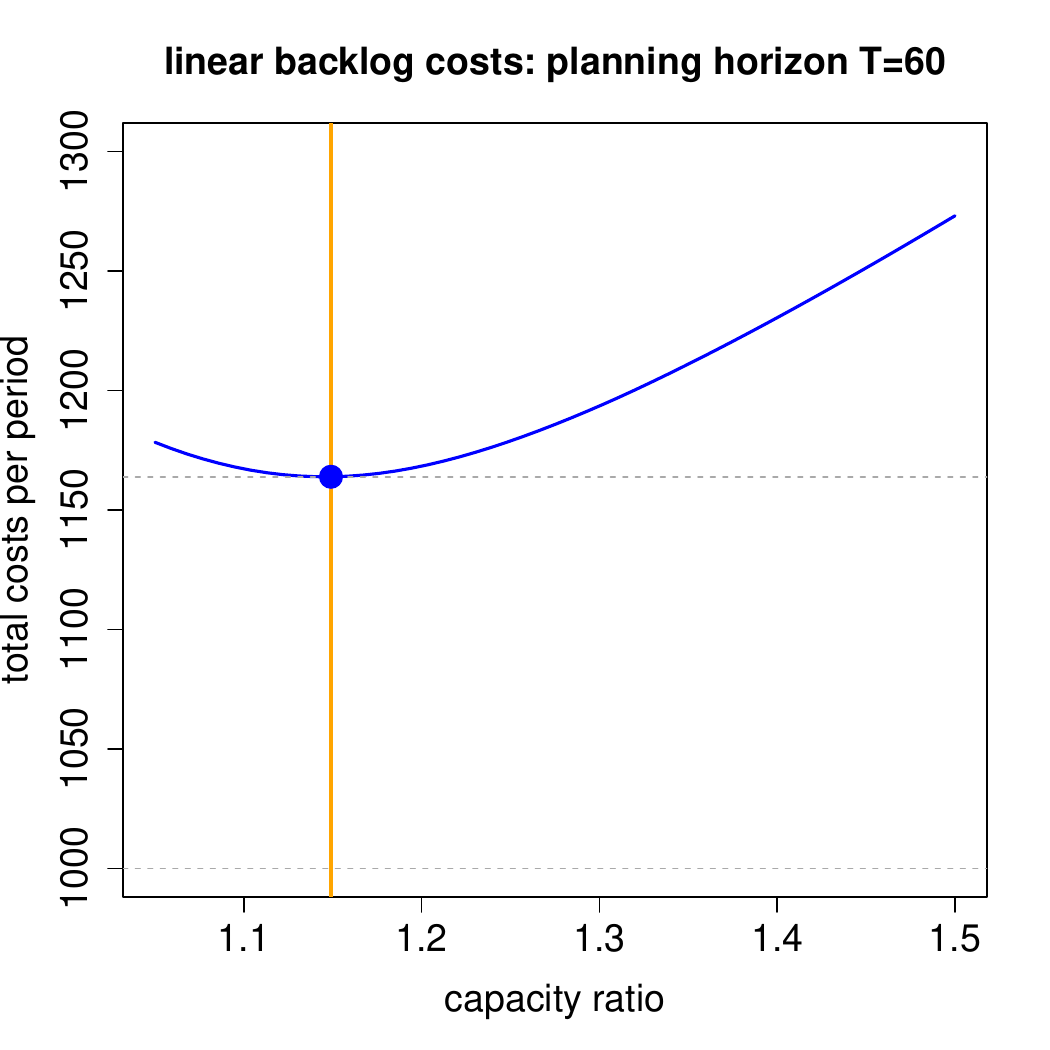}
\end{center}
\end{minipage}
\begin{minipage}[t]{0.44\textwidth}
\begin{center}
\includegraphics[width=\textwidth]{./Plots/LinearCapacity.pdf}
\end{center}
\end{minipage}
\end{center}
\vspace{-.65cm}
\caption{Optimal capacity ratios in the linear backlog costs case, conditional version, with a zero starting backlog $B^{(\eta)}_{\tau }=0$ at time $\tau$ showing:
(lhs) planning horizon $T=60$, (rhs) unconditional case taken from Figure \ref{linear cost example} (lhs), this corresponds to $T=\infty$.}
\label{linear cost example conditional}
\end{figure}

Figure \ref{linear cost example conditional} shows the resulting optimal capacity ratios in the conditional case if we start with a zero backlog $B^{(\eta)}_{\tau }=0$ at time $\tau$
for planning horizon $T=60$, 
and it is compared to the unconditional case given in Figure \ref{linear cost example} (lhs) using the same cost parameters $\kappa_g$, $\kappa_b$ and $\kappa_c$; note that the unconditional case corresponds to the infinite planning horizon. 

\begin{table}[htb]
\begin{center}
{\small
\begin{tabular}{|c|cc|}
\hline
planning horizon $T$ & optimal $\eta^\ast$ & costs $\mu^{(\ell)}(\eta^\ast, T, \calG^{(\eta^\ast)}_{\tau})$\\\hline
36 & 1.068 & 1152\\ 
60 & 1.149 & 1164\\
120 & 1.176 & 1172\\
$\infty$ & 1.203 & 1175\\
\hline
\end{tabular}}
\end{center}
\caption{Optimal capacity ratios $\eta^\ast$ for different
planning horizons $T$ and result total costs linear 
backlog cost case.}
\label{results different planning horizons}
\end{table}

Table \ref{results different planning horizons} gives the numbers for planning horizons including those in Figure \ref{linear cost example conditional}. 
With a planning horizon of 120 (monthly) periods (or 10 years) we are rather close to the unconditional case, having average optimal costs per period of 1172. For shorter planning horizons these costs are lower, this is because we start with a zero backlog $B^{(\eta)}_{\tau }=0$ at time $\tau$, thus, this is a more favorable starting position than the average over the stationary limit distribution (unconditional case). This results for a planning horizon of $T=36$ (montly) periods (or 3 years) to optimal average costs of 1152. 

Similar results with decreasing instead of increasing average costs are obtained if we start
from a large initial backlog $B^{(\eta)}_{\tau }$ at time $\tau$. 
This can be interpreted as a situation where the backlog has gone
out of control, e.g., due to a catastrophic claims event, and the management
is concerned about clearing the backlog at minimal mid-term costs.
We refrain from giving an explicit numerical example.

\section{Summary}\label{sec:summary}
We formalized the question of choosing optimal processing capacities for claims handling. On the one hand, the claims handling capacity needs to be limited because any insurance company has only finite financial resources available. 
On the other hand, the capacity should be sufficiently large because long processing delays (and large backlogs) also generate various costs. We studied this trade-off aiming at minimizing claims and claims processing costs. This problem has several features from queueing theory, but there are also some significant differences because claims are labeled by occurrence periods and arising expenses need to be allocated to occurrence periods to have a consistent and appropriate cost analysis of an insurance portfolio.

We formalized these questions and we solved a variant of this optimal cost and capacity problem. This variant describes a specific mechanism to work off a backlog, and it considers a specific super-imposed cost inflation factor for late claims processing. In this regard, there are many alternative ways to model these backlog cost items. Our choice is a realistic one that is still fairly well tractable, and the final intractable step was solved by a recurrent neural network approximation. 
This paper appears to be the first that considers this claims processing costs problem. Alternative ways to model delay-adjusted costs and other consequences of backlogs due to capacity constraints with shared capacity could be fruitful to explore. We invite interested scholars to contribute to this interesting problem.

\bigskip

{\bf Acknowledgement.}
Parts of this research was carried out while Mario W\"uthrich was a KAW guest professor at Stockholm University. 
Filip Lindskog acknowledges financial support from the Swedish Research Council, Project 2020-05065.

\newpage

\appendix

\section{Figures for neural network approximations}

\begin{figure}[htb!]
\begin{center}
\begin{minipage}[t]{0.4\textwidth}
\begin{center}
\includegraphics[width=\textwidth]{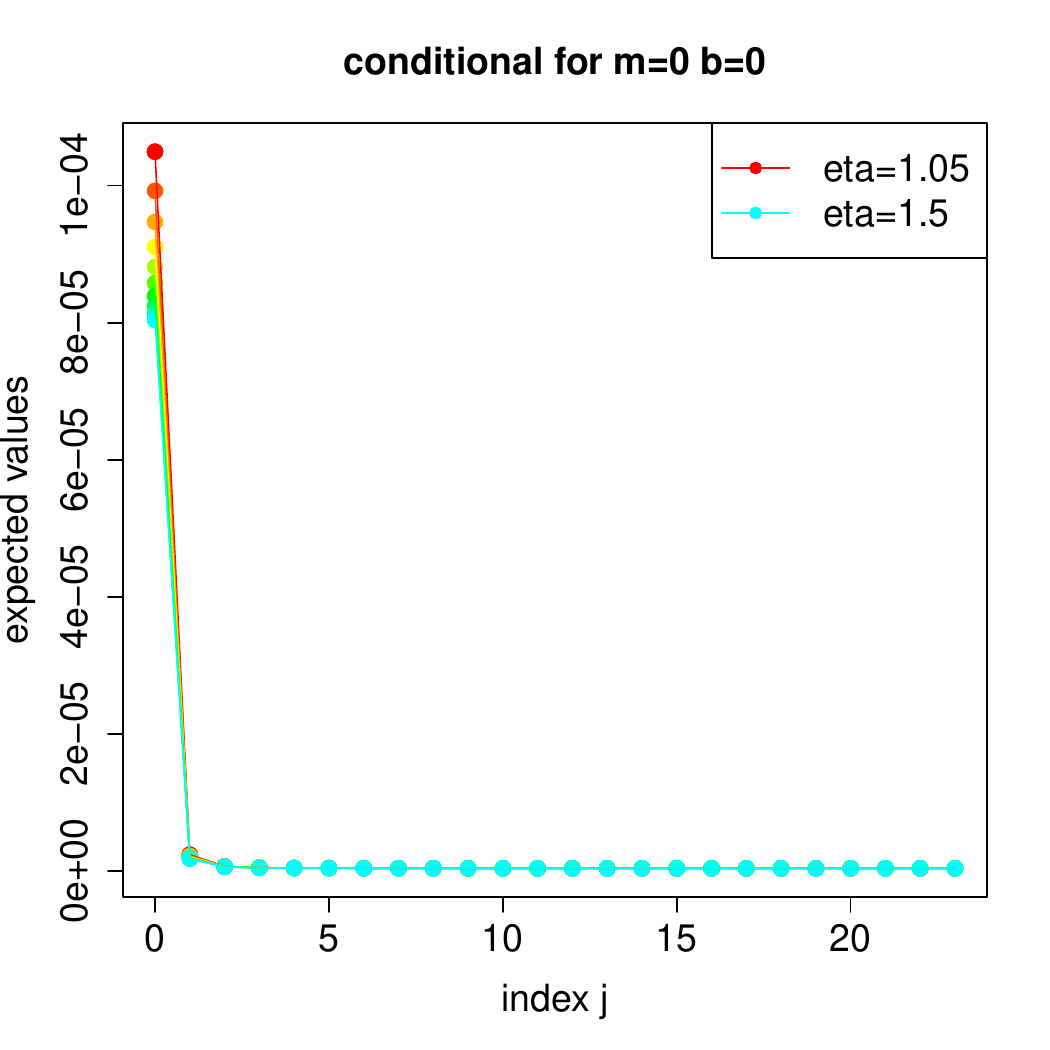}
\end{center}
\end{minipage}
\begin{minipage}[t]{0.4\textwidth}
\begin{center}
\includegraphics[width=\textwidth]{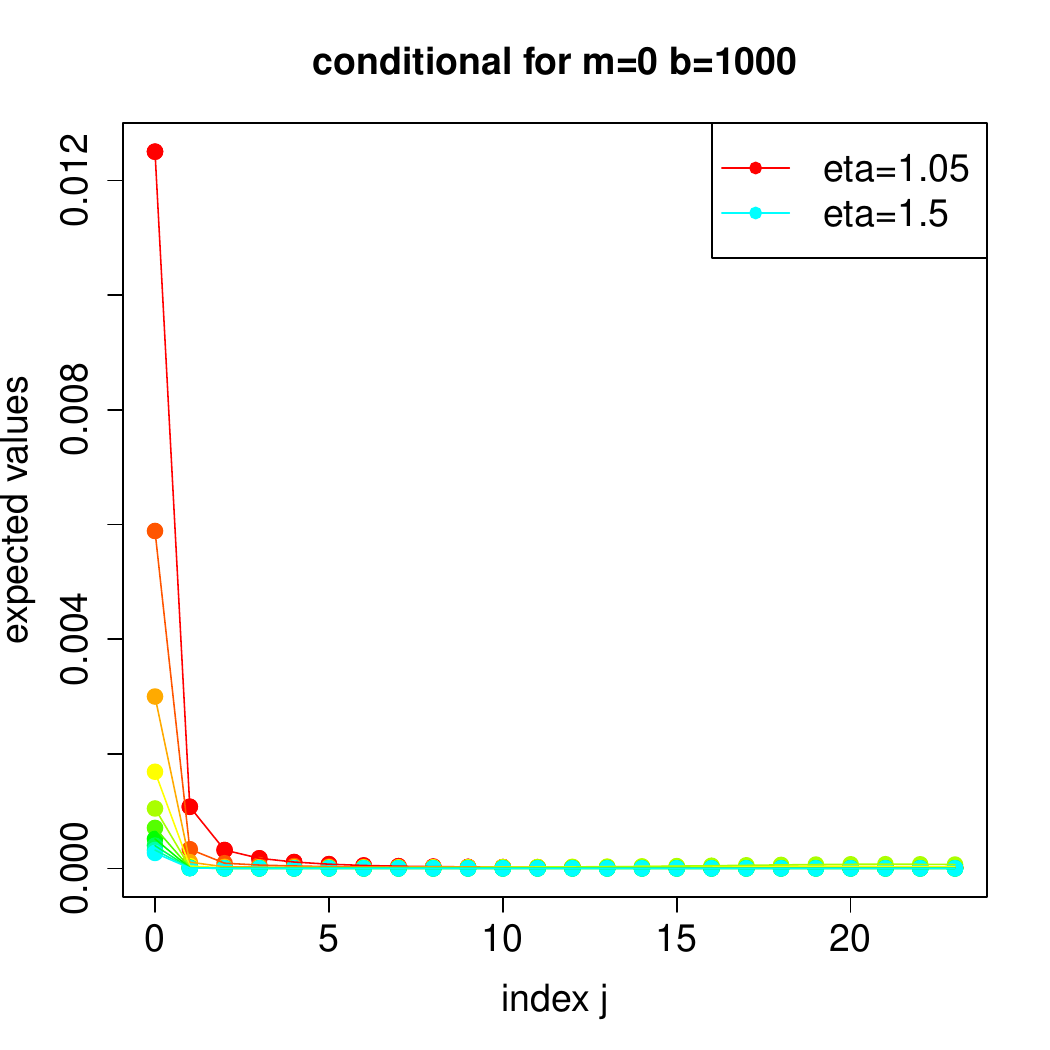}
\end{center}
\end{minipage}
\begin{minipage}[t]{0.4\textwidth}
\begin{center}
\includegraphics[width=\textwidth]{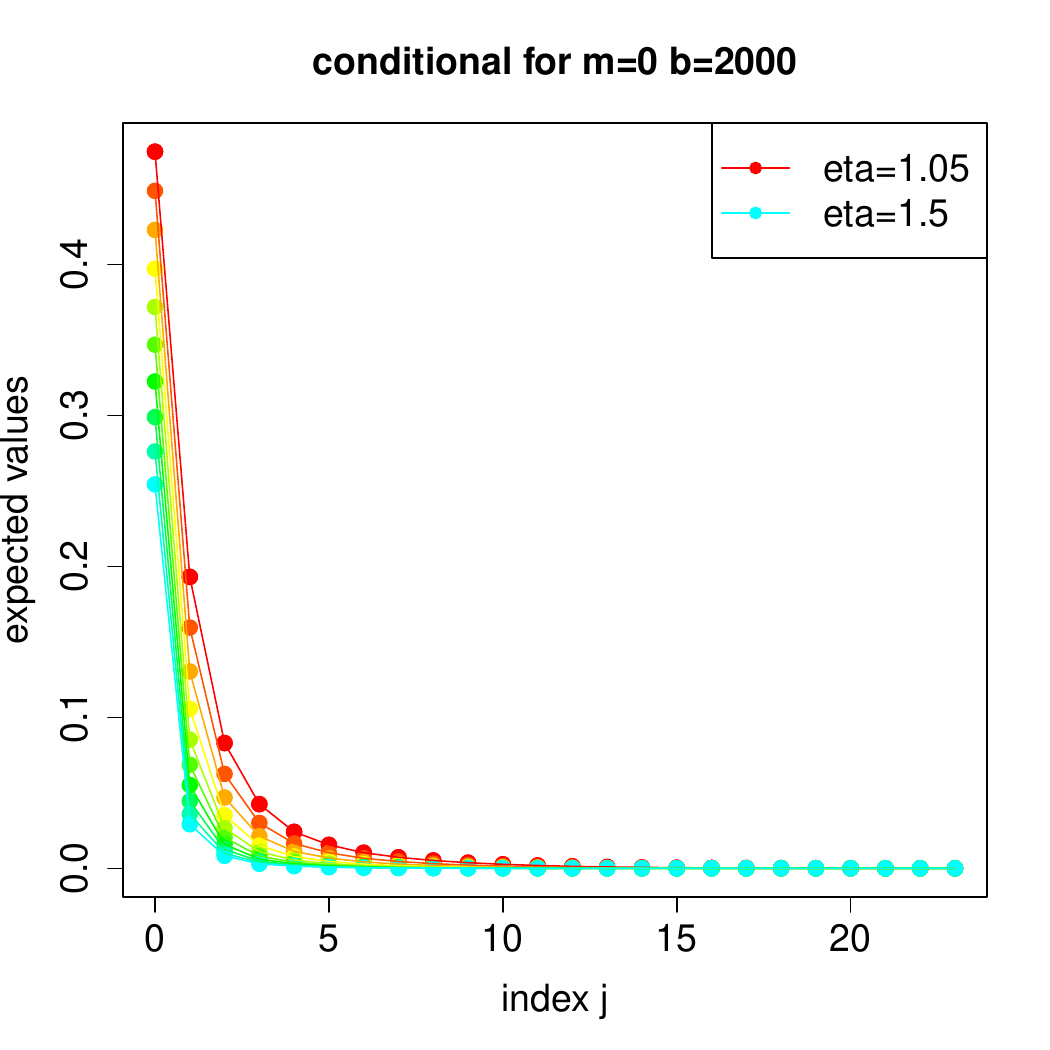}
\end{center}
\end{minipage}
\begin{minipage}[t]{0.4\textwidth}
\begin{center}
\includegraphics[width=\textwidth]{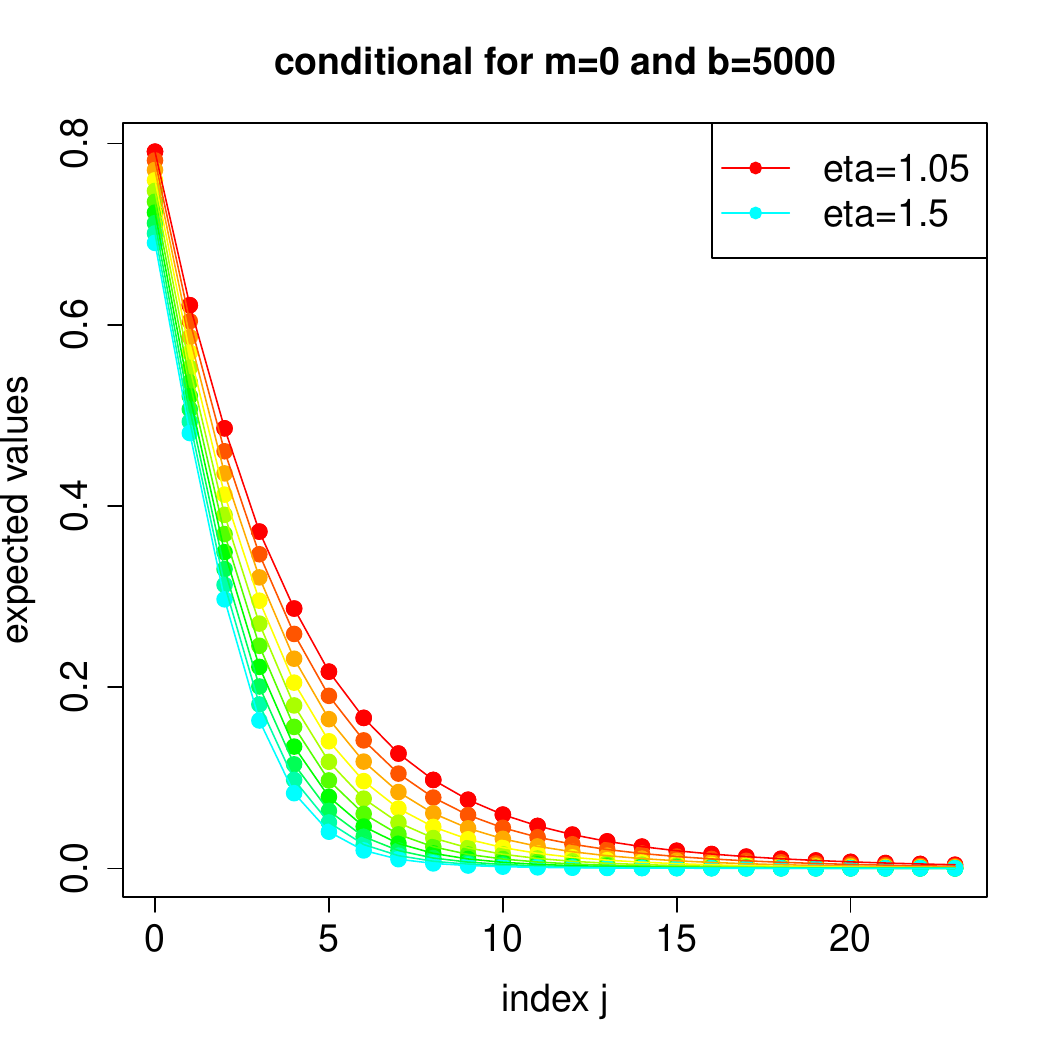}
\end{center}
\end{minipage}
\end{center}
\vspace{-.65cm}
\caption{RNN approximations of the conditional means
$h_j(b, 0;\eta)$, $j\ge 1$, for capacity
ratios $\eta\in \{1.05, 1.10, \ldots, 1.50\}$ (in different colors),
$b\in\{0,1000, 2000, 5000\}$ (different plots) and $m=0$.}
\label{figure conditional 0}
\end{figure}

\begin{figure}[htb!]
\begin{center}
\begin{minipage}[t]{0.4\textwidth}
\begin{center}
\includegraphics[width=\textwidth]{./Plots/NetworkGELU_NB_meanFG_eta0.pdf}
\end{center}
\end{minipage}
\begin{minipage}[t]{0.4\textwidth}
\begin{center}
\includegraphics[width=\textwidth]{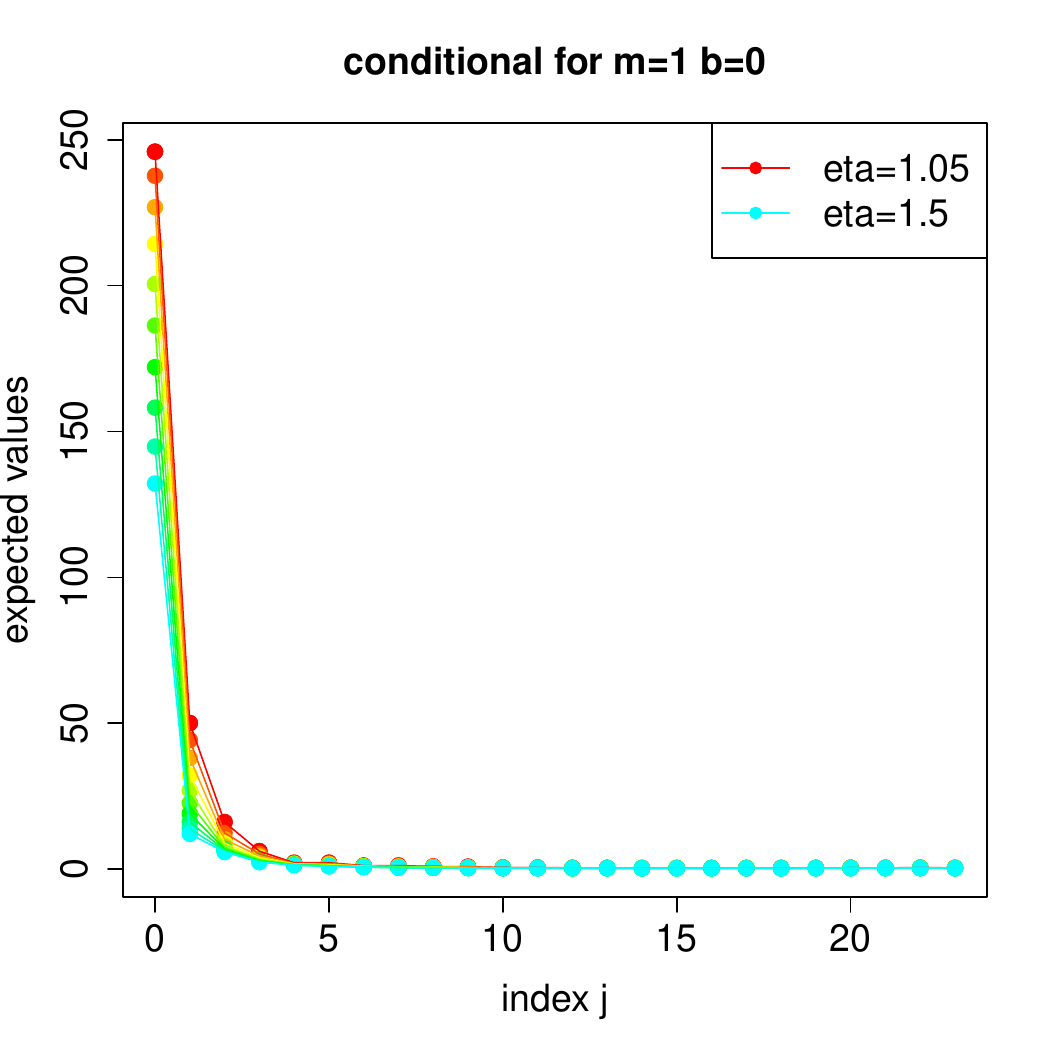}
\end{center}
\end{minipage}
\begin{minipage}[t]{0.4\textwidth}
\begin{center}
\includegraphics[width=\textwidth]{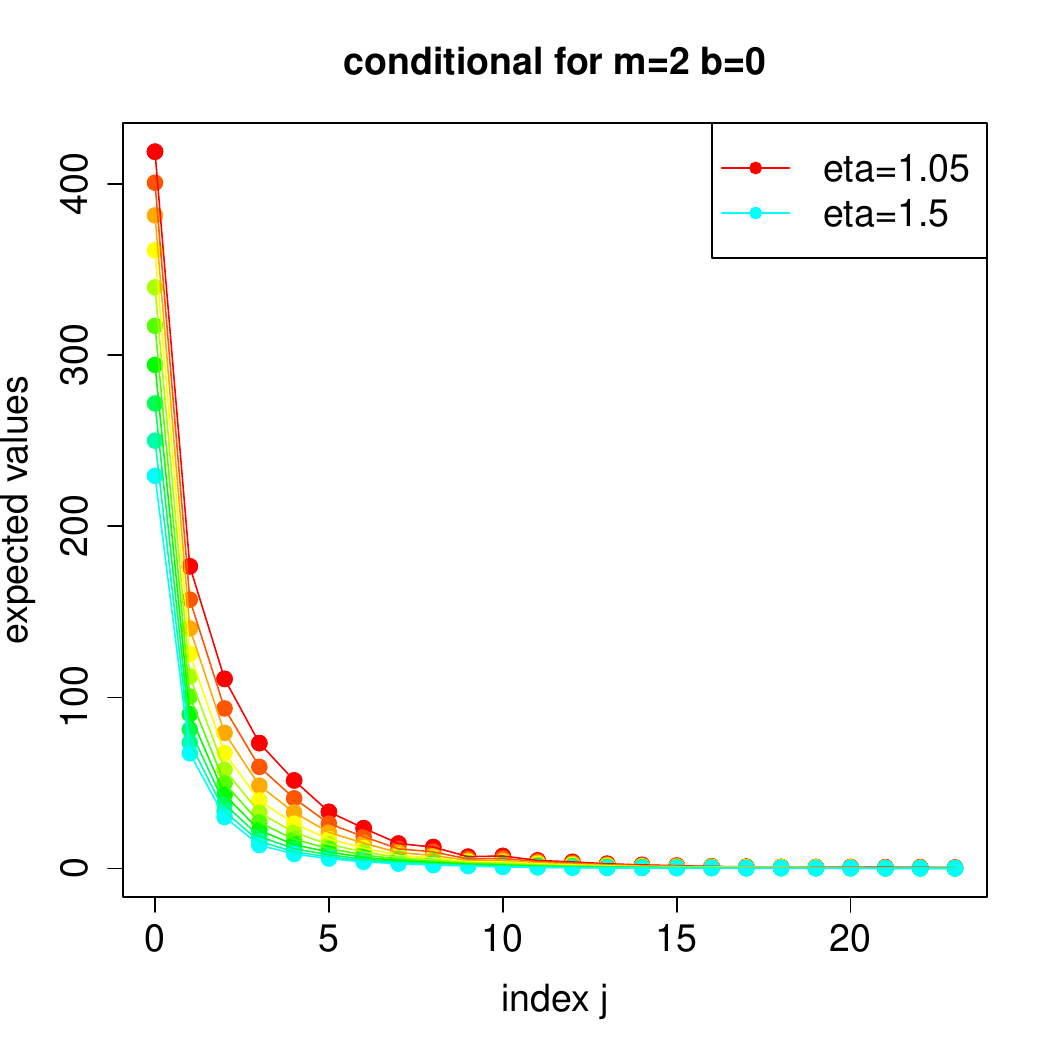}
\end{center}
\end{minipage}
\begin{minipage}[t]{0.4\textwidth}
\begin{center}
\includegraphics[width=\textwidth]{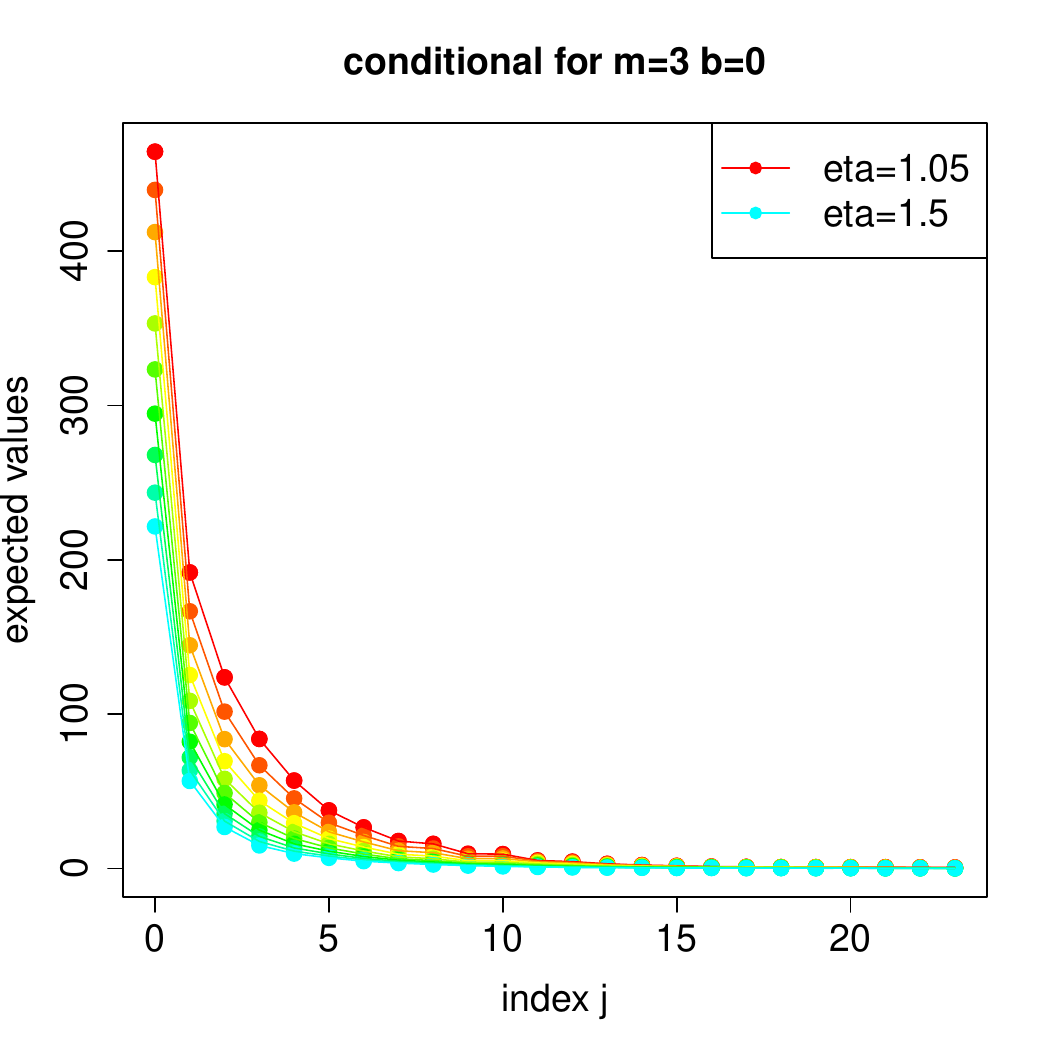}
\end{center}
\end{minipage}
\begin{minipage}[t]{0.4\textwidth}
\begin{center}
\includegraphics[width=\textwidth]{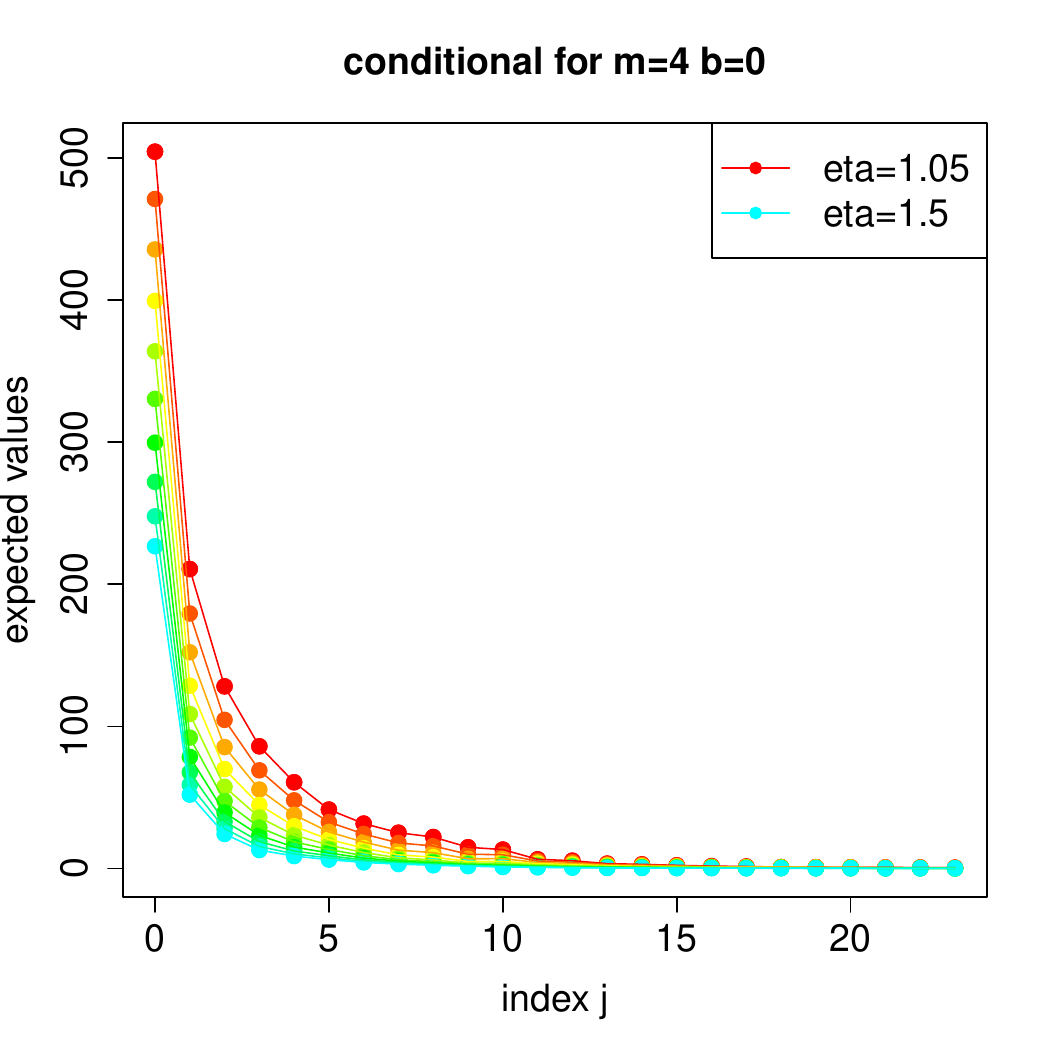}
\end{center}
\end{minipage}
\begin{minipage}[t]{0.4\textwidth}
\begin{center}
\includegraphics[width=\textwidth]{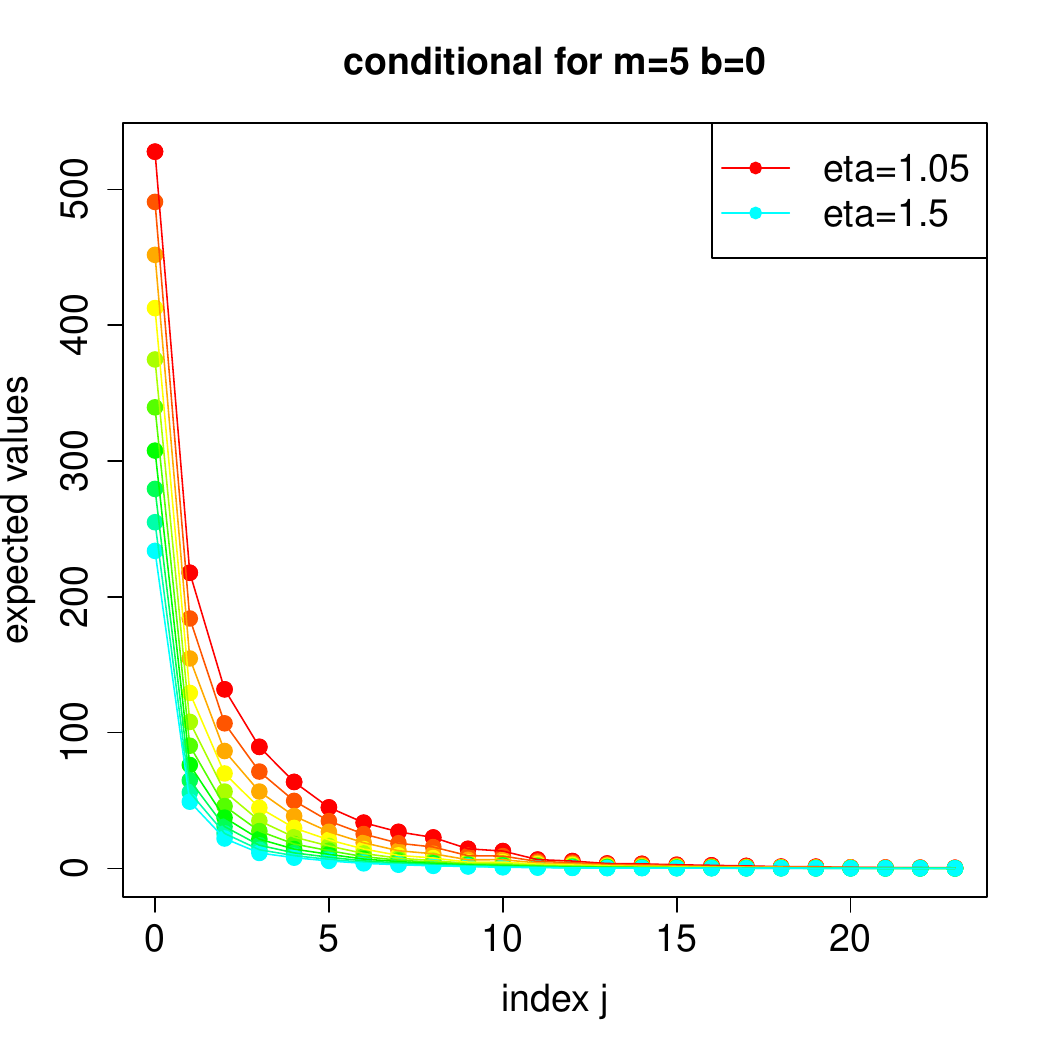}
\end{center}
\end{minipage}
\end{center}
\vspace{-.65cm}
\caption{RNN approximations of the conditional means
$h_j(b, m;\eta)$, $j\ge 0$,  for capacity
ratios $\eta\in \{1.05, 1.10, \ldots, 1.50\}$ (in different colors),
$b=0$ and $m\in \{1,\ldots, 5\}$ (different plots),
top-left is taken from Figure \ref{NB network approximation 2}.}
\label{figure conditional 1}
\end{figure}

\begin{figure}[htb!]
\begin{center}
\begin{minipage}[t]{0.4\textwidth}
\begin{center}
\includegraphics[width=\textwidth]{./Plots/NetworkGELU_NB_meanFG_1000.pdf}
\end{center}
\end{minipage}
\begin{minipage}[t]{0.4\textwidth}
\begin{center}
\includegraphics[width=\textwidth]{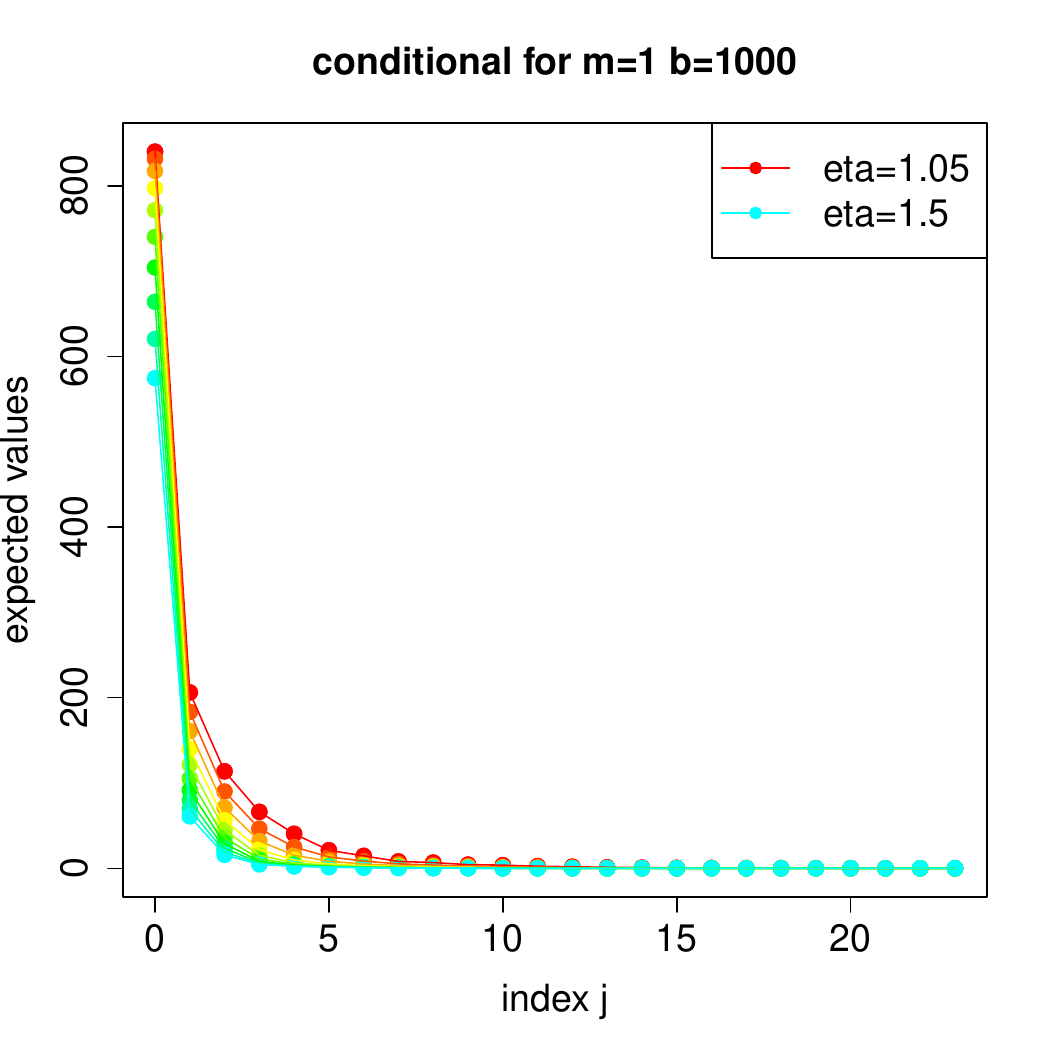}
\end{center}
\end{minipage}
\begin{minipage}[t]{0.4\textwidth}
\begin{center}
\includegraphics[width=\textwidth]{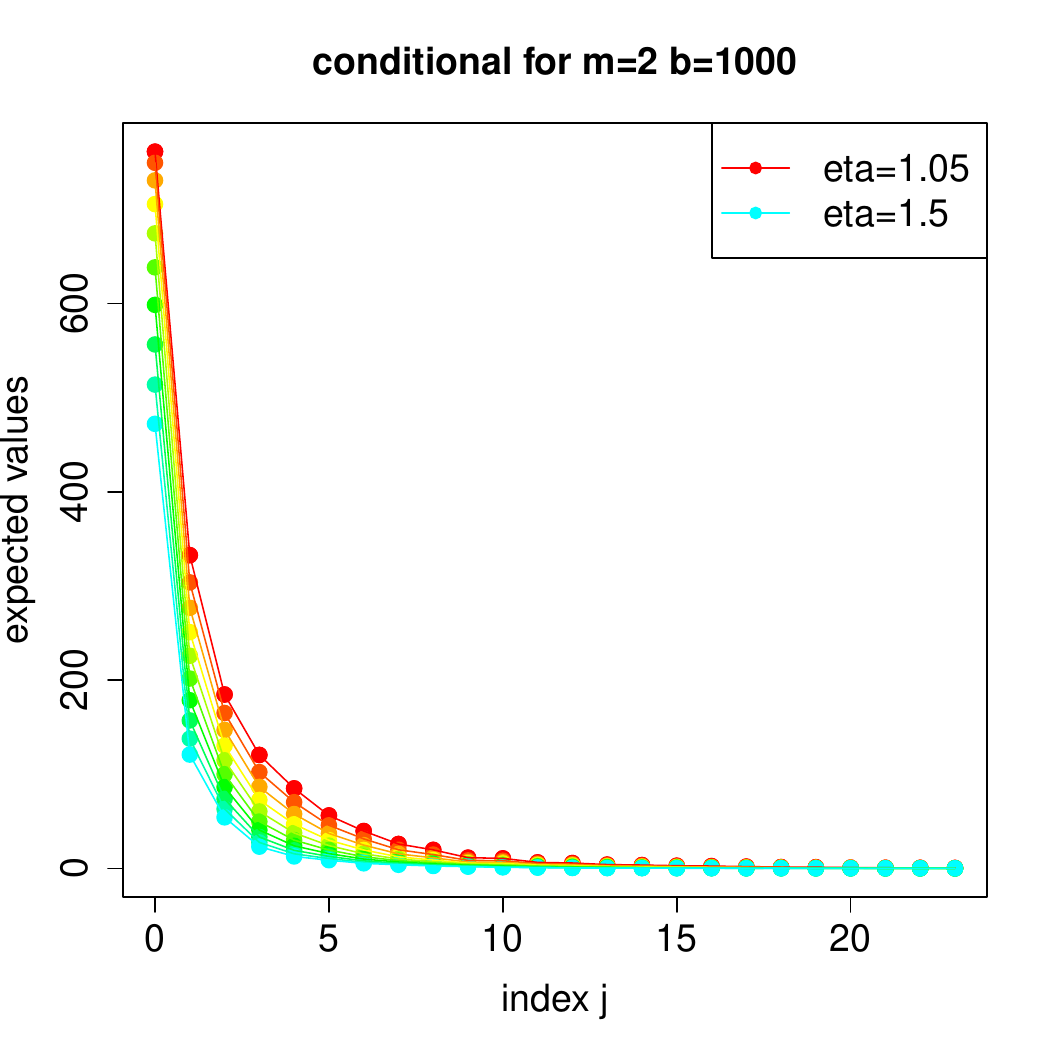}
\end{center}
\end{minipage}
\begin{minipage}[t]{0.4\textwidth}
\begin{center}
\includegraphics[width=\textwidth]{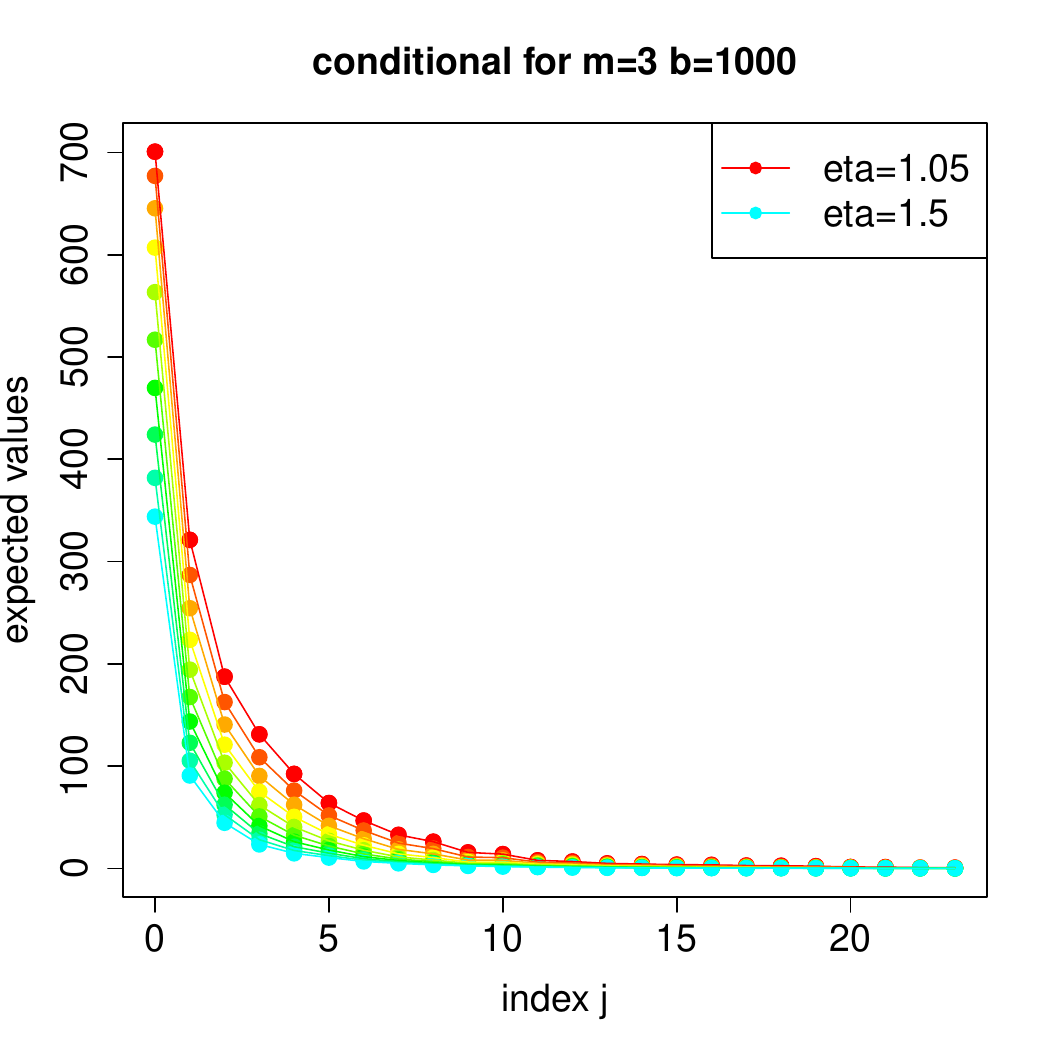}
\end{center}
\end{minipage}
\begin{minipage}[t]{0.4\textwidth}
\begin{center}
\includegraphics[width=\textwidth]{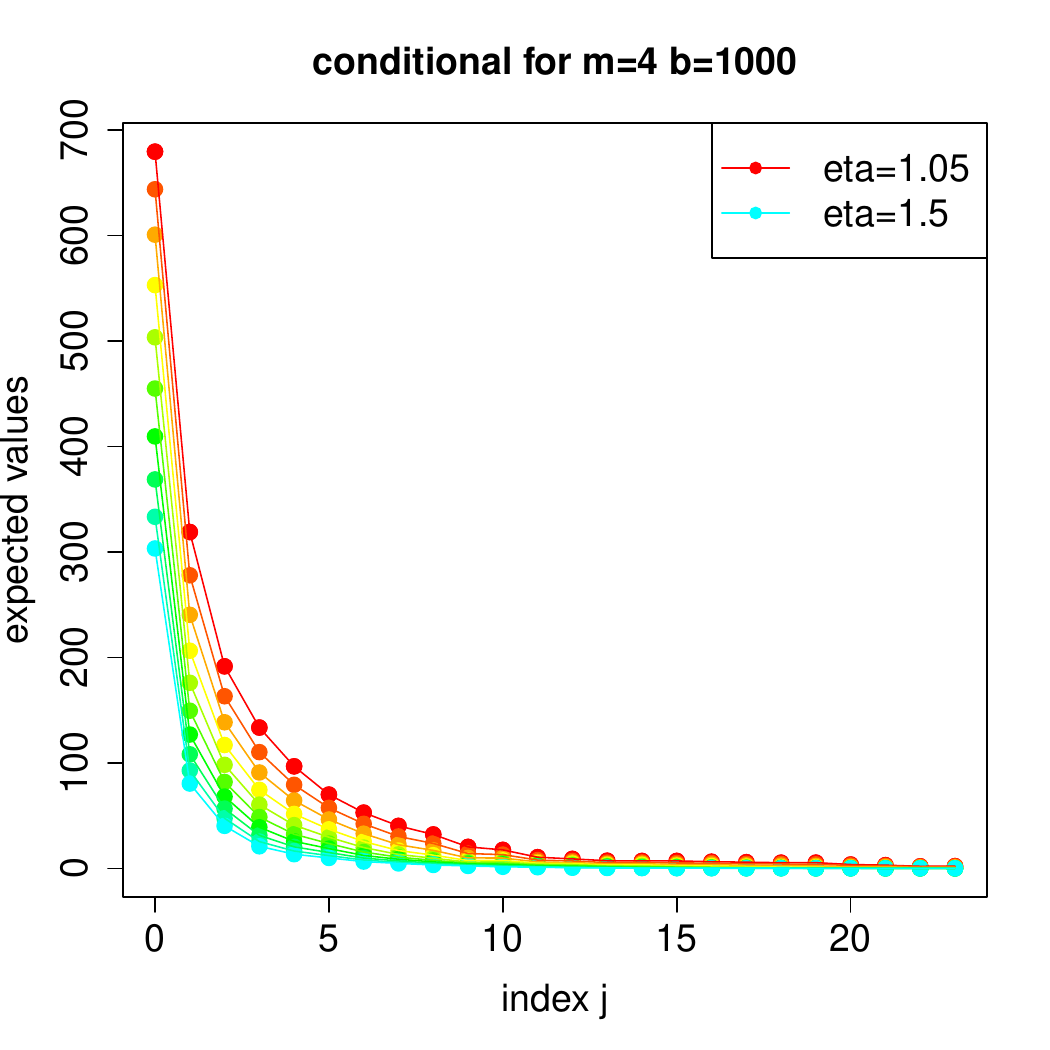}
\end{center}
\end{minipage}
\begin{minipage}[t]{0.4\textwidth}
\begin{center}
\includegraphics[width=\textwidth]{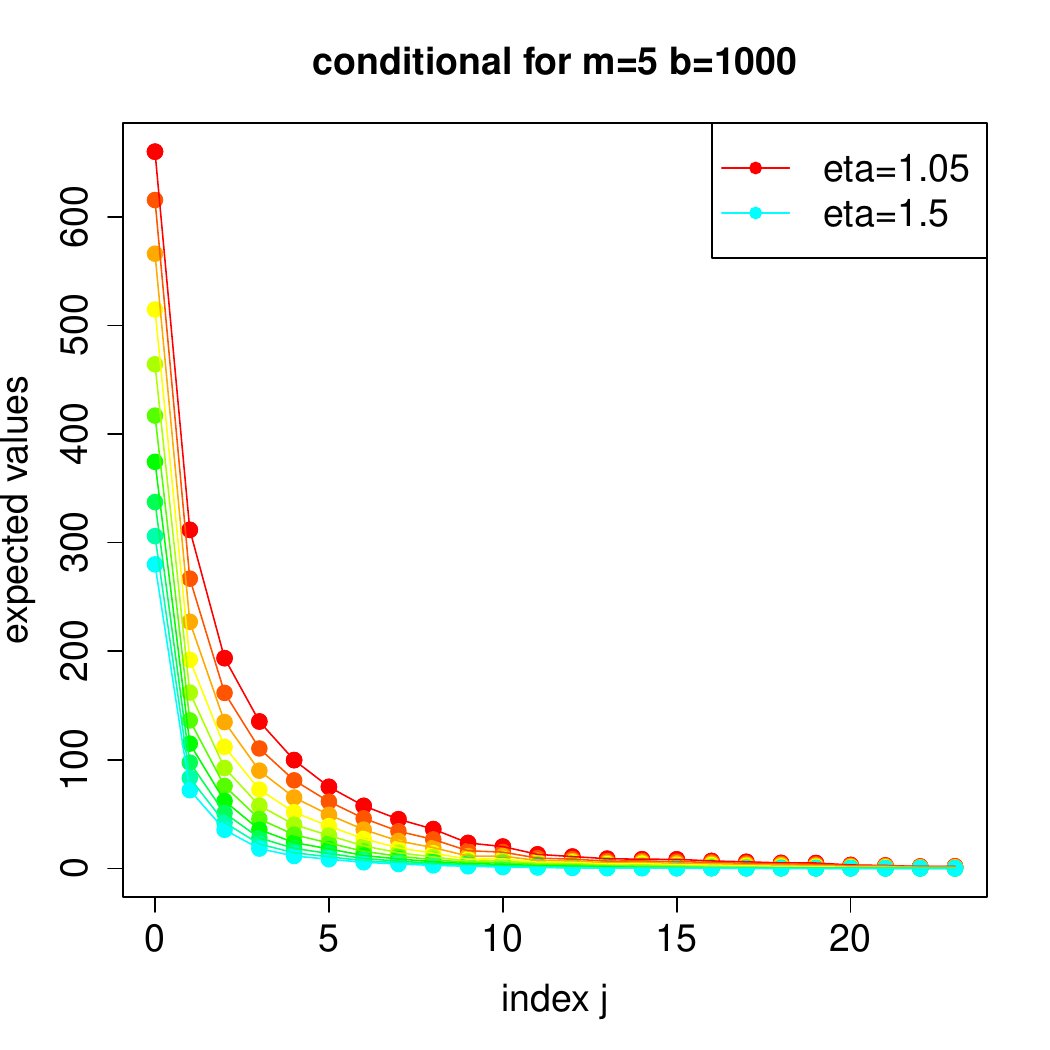}
\end{center}
\end{minipage}
\end{center}
\vspace{-.65cm}
\caption{RNN approximations of the conditional means
$h_j(b, m;\eta)$, $j\ge 0$,  for capacity
ratios $\eta\in \{1.05, 1.10, \ldots, 1.50\}$ (in different colors),
$b=1000$ and $m\in \{1,\ldots, 5\}$ (different plots),
top-left is taken from Figure \ref{NB network approximation 2}.}
\label{figure conditional 2}
\end{figure}

\begin{figure}[htb!]
\begin{center}
\begin{minipage}[t]{0.4\textwidth}
\begin{center}
\includegraphics[width=\textwidth]{./Plots/NetworkGELU_NB_meanFG_5000.pdf}
\end{center}
\end{minipage}
\begin{minipage}[t]{0.4\textwidth}
\begin{center}
\includegraphics[width=\textwidth]{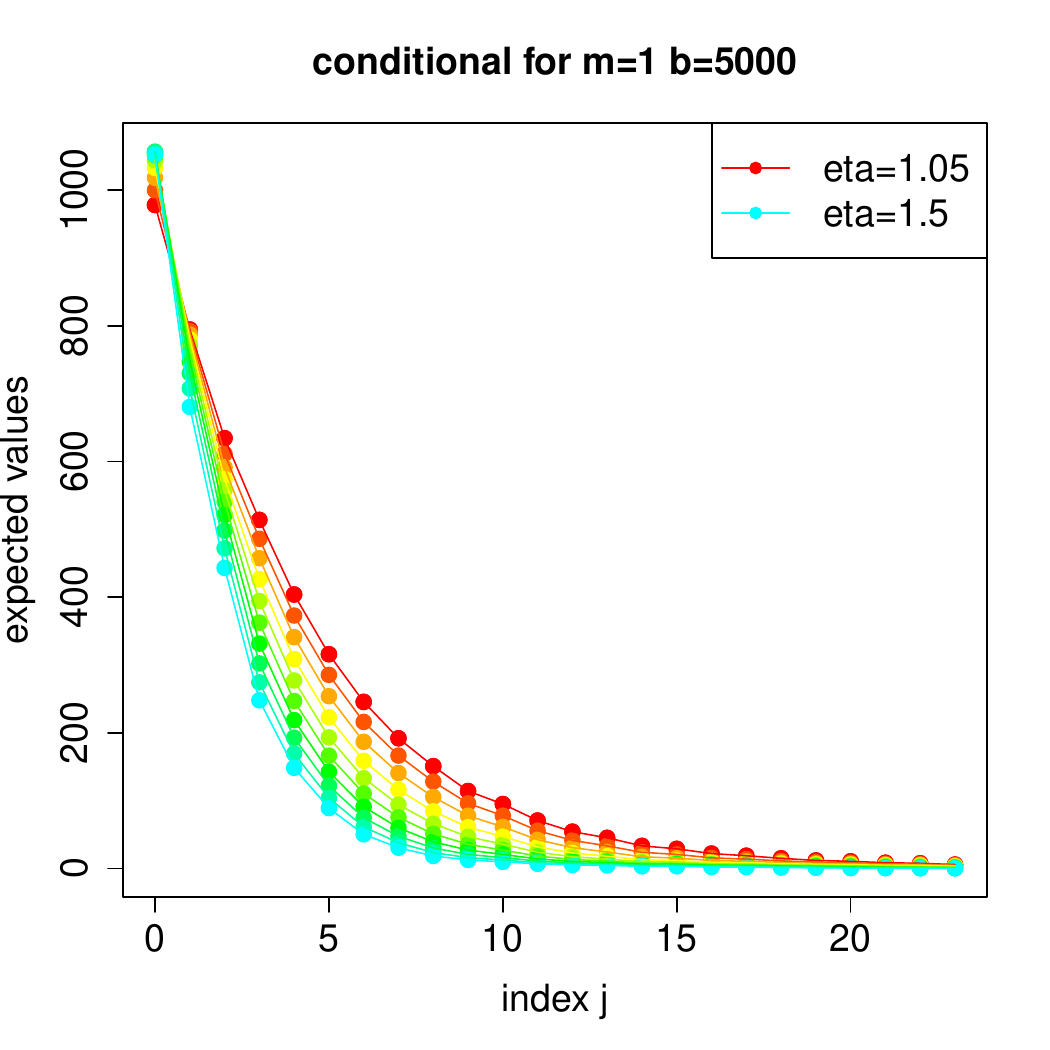}
\end{center}
\end{minipage}
\begin{minipage}[t]{0.4\textwidth}
\begin{center}
\includegraphics[width=\textwidth]{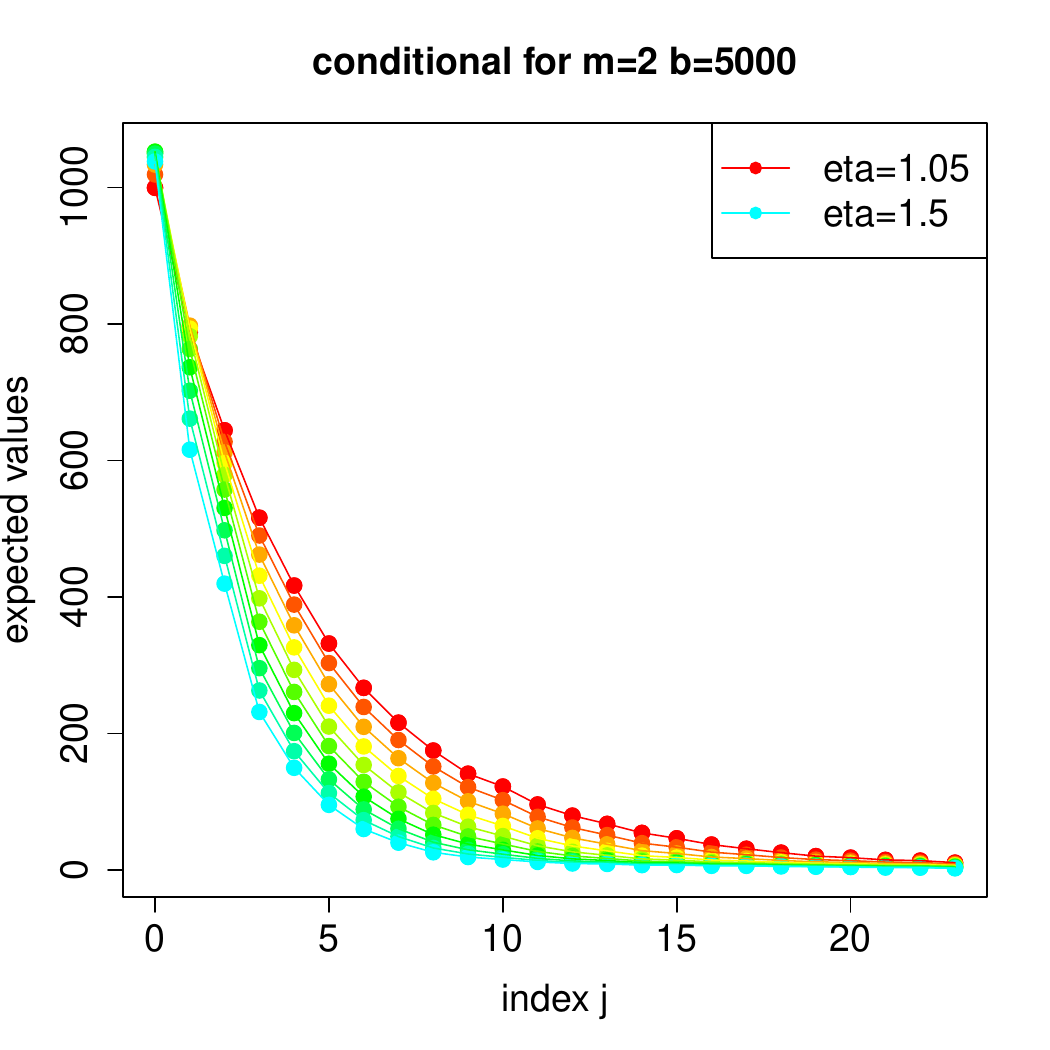}
\end{center}
\end{minipage}
\begin{minipage}[t]{0.4\textwidth}
\begin{center}
\includegraphics[width=\textwidth]{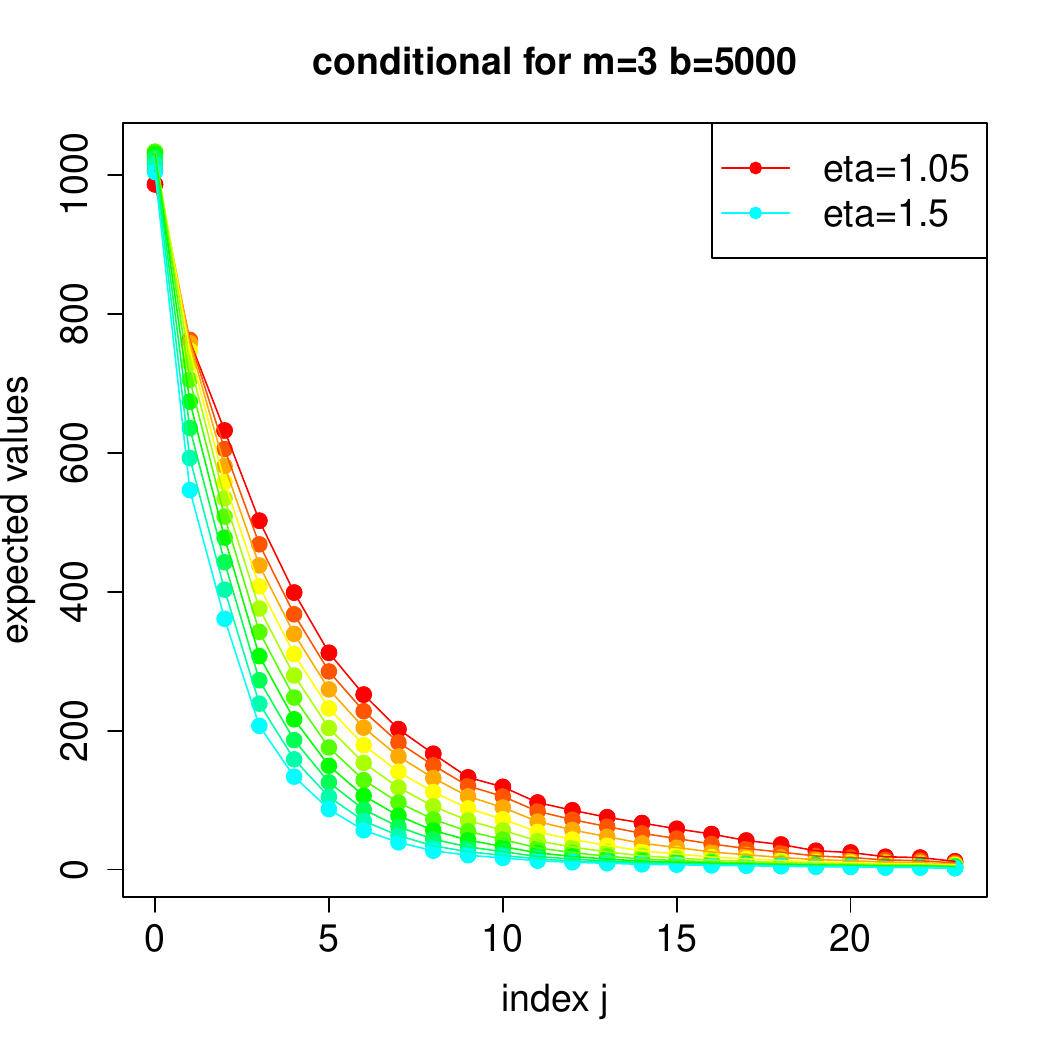}
\end{center}
\end{minipage}
\begin{minipage}[t]{0.4\textwidth}
\begin{center}
\includegraphics[width=\textwidth]{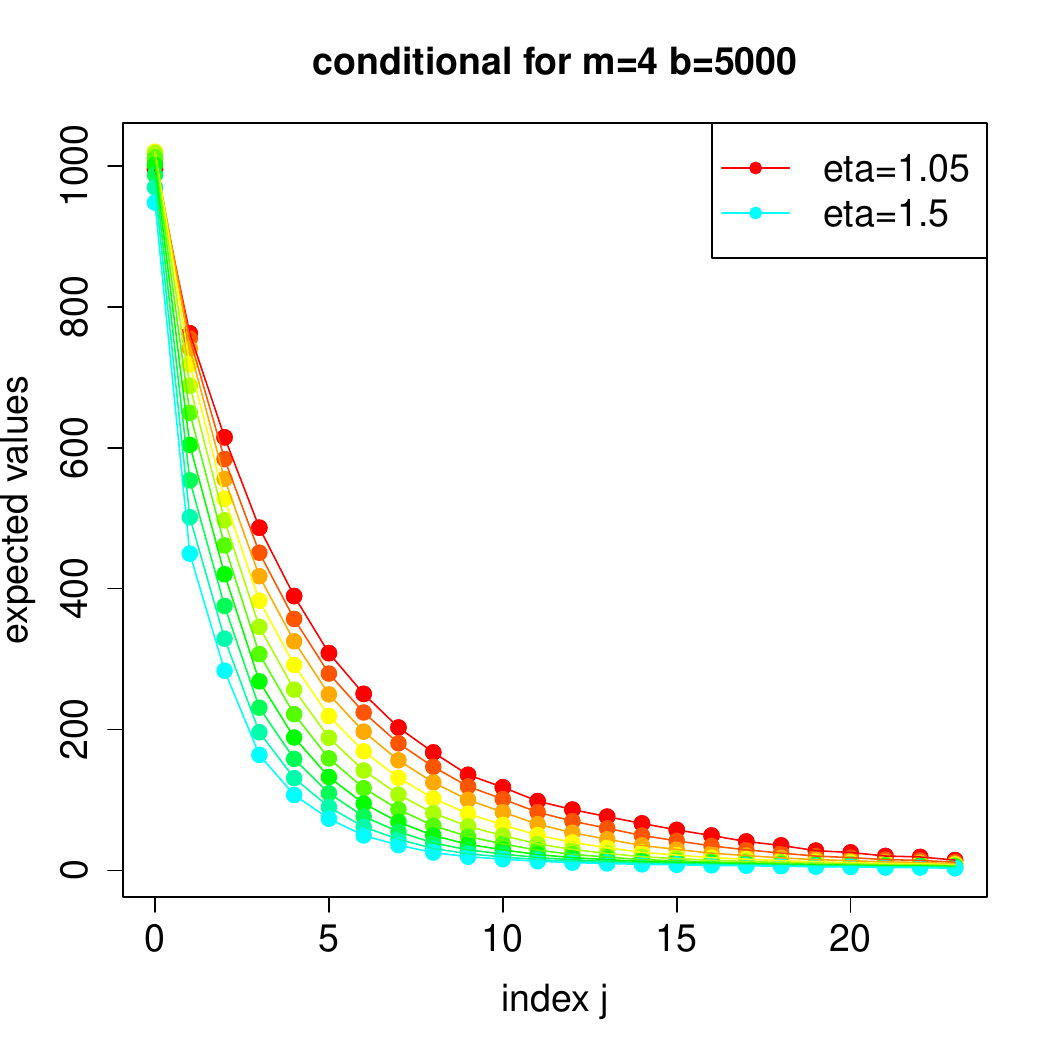}
\end{center}
\end{minipage}
\begin{minipage}[t]{0.4\textwidth}
\begin{center}
\includegraphics[width=\textwidth]{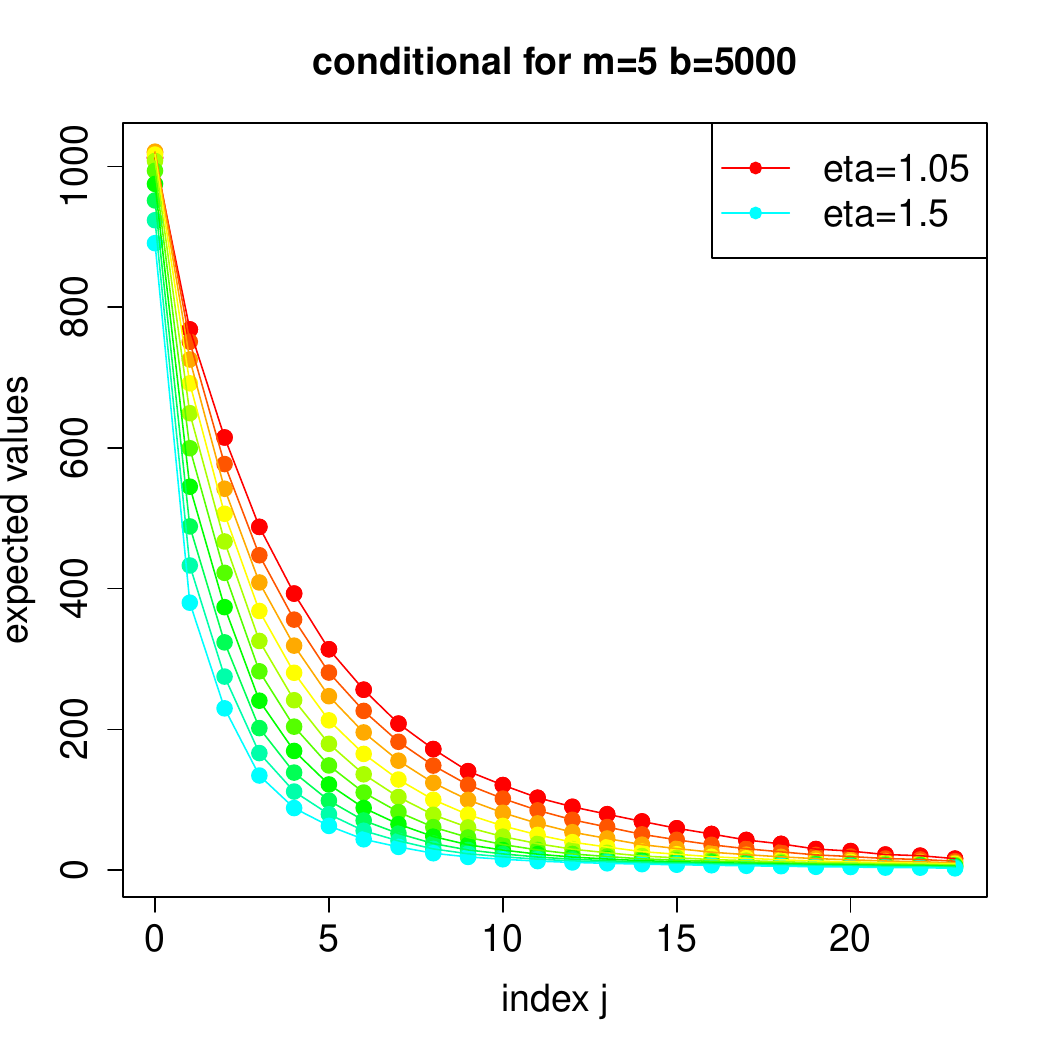}
\end{center}
\end{minipage}
\end{center}
\vspace{-.65cm}
\caption{RNN approximations of the conditional means
$h_j(b, m;\eta)$, $j\ge 0$,  for capacity
ratios $\eta\in \{1.05, 1.10, \ldots, 1.50\}$ (in different colors),
$b=5000$ and $m\in \{1,\ldots, 5\}$ (different plots),
top-left is taken from Figure \ref{NB network approximation 2}.}
\label{figure conditional 3}
\end{figure}


\begin{thebibliography}{99}
\bibitem{Albrecher-et-al-2011}
Albrecher, H., Denuit, M., Trufin, J.~(2011).
Ruin problems under IBNR dynamics.
\textit{Applied Stochastic Models in Business and Industry} 27(6), 619-632.  
\bibitem{Asmussen-2003}
Asmussen, S.~(2003). 
\textit{Applied Probability and Queues}, second edn. Springer, New York.  
\bibitem{Billingsley-1995}
Billingsley, P.~(1995).
\textit{Probability and Measure}, third edn. Wiley, New York. 
\bibitem{Boogaert-Haezendonck-1989}
Boogaert, P., Haezendonck, J.~(1989).
Delay in claim settlement.
\textit{Insurance: Mathematics and Economics} 8, 321-330.
\bibitem{Buchwalder}
Buchwalder, M., B\"uhlmann, H., Merz, M., W\"uthrich, M.V.~(2006).
Estimation of unallocated loss adjustment expenses.
\textit{Bulletin of the Swiss Association of Actuaries} 2006(1), 43-53. 
\bibitem{Chen-Whitt-2020}
Chen, Y., Whitt, W.~(2020).
Algorithms for the upper bound mean waiting time in the GI/GI/1 queue.
\textit{Queueing Systems} 94, 327-356. 
\bibitem{Daley-1977}
Daley, D.J.~(1977).  
Inequalities for moments of tails of random variables, with queueing applications.
\textit{Zeitschrift fur Wahrscheinlichkeitsetheorie Verw. Gebiete} 41, 139-143. 
\bibitem{Harel-1990}
Harel, A.~(1990)
Convexity results for single-serves queues and multiserver queues with constant service times. 
\textit{Journal of Applied Probability} 27(2),  465-468. 
\bibitem{Huynh-et-al-2015}
Huynh, M., Landriault, D., Shi, T., Willmot, G.E.~(2015).
On a risk model with claim investigation.
\textit{Insurance: Mathematics and Economics} 65, 37-45.
\bibitem{Janssen-Leeuwaarden-2005}
Janssen, A.J.E.M., Van Leeuwaarden, J.S.H.~(2005).
Relaxation time for the discrete D/G/1 queue. 
\textit{Queueing Systems} 50, 53-80. 
\bibitem{Janssen-Leeuwaarden-2018}
Janssen, A.J.E.M., Van Leeuwaarden, J.S.H.~(2018).
Spitzer's identity for discrete random walks. 
\textit{Operations Research Letters} 46, 168-172. 
\bibitem{Kingman-1962}
Kingman, J.F.C.~(1962). 
Inequalities for the queue GI/G/1. 
\textit{Biometrika} 49(3/4), 315-324. 
\bibitem{Maglaras-Zeevi-2003}
Maglaras, C., Zeevi, A.~(2003).
Pricing and capacity sizing for systems with shared resources: approximate solutions and scaling relations. 
\textit{Management Science} 49(8), 1018-1038. 
\bibitem{Sato-1999}
Sato, K.-I.~(1999).
\textit{L\'evy Processes and Infinitely Divisible Distributions}. Cambridge Studies in Advanced Mathematics 68, Cambridge University Press. 
\bibitem{Stadje-1997}
Stadje, W.~(1997).
A new approach to the Lindley recursion. 
\textit{Statistics \& Probability Letters} 31, 169-175.
\bibitem{Waters-Papatriandafylou-1985}
Waters, H.R., Papatriandafylou, A.~(1985).
Ruin probabilities allowing for delay in claims settlement.
\textit{Insurance: Mathematics and Economics} 4, 113-122. 
\end{thebibliography}
\end{document}